\documentclass[11pt,a4paper]{article}

\usepackage{subfiles} 
\usepackage[utf8]{inputenc}
\usepackage[T1]{fontenc}
\usepackage{amsmath, amssymb, amsthm, amsfonts, mathrsfs}
\usepackage{mathrsfs}
\usepackage{bm}
\usepackage{braket}
\usepackage{geometry}
\usepackage{graphicx}
\usepackage{subcaption}
\usepackage{booktabs}
\usepackage{multirow}
\usepackage{adjustbox}
\usepackage{xcolor}
\usepackage[colorlinks=true, linkcolor=blue!50!black, citecolor=blue!50!black, urlcolor=blue!50!black]{hyperref}
\usepackage{natbib}
\usepackage{enumitem}
\usepackage{tikz}
\usepackage{standalone}
\usetikzlibrary{positioning}

\theoremstyle{plain}
\newtheorem{theorem}{Theorem}
\newtheorem{proposition}{Proposition}
\newtheorem{corollary}{Corollary}
\newtheorem{lemma}{Lemma}
\theoremstyle{definition}
\newtheorem{definition}{Definition}
\theoremstyle{remark}
\newtheorem{remark}{Remark}

\newcommand{\R}{\mathbb{R}}
\newcommand{\Z}{\mathbb{Z}}

\newcommand{\bO}{\mathcal{O}}

\newcommand{\TV}{\mathrm{TV}}
\DeclareMathOperator{\diam}{diam}

\allowdisplaybreaks

\usepackage{mathtools}
\renewcommand{\braket}[2]{\langle #1|#2\rangle}
\newcommand{\norm}[1]{\|#1\|}
\newcommand{\abs}[1]{|#1|}
\newcommand{\op}{\mathrm{op}}
\newcommand{\tr}{\mathrm{tr}}
\newcommand{\Real}{\mathbb{R}}
\newcommand{\Complex}{\mathbb{C}}
\newtheorem{assumption}[theorem]{Assumption}
\newtheorem{example}[theorem]{Example}

\begin{document}

\begin{center}
    {\LARGE \textbf{Q-Edge: Symmetry-Reduced Quantum Simulation of Structured Extreme Dependence}}\\[1.2em]
    {\large Hongrui Zhang,\ Paolo Recchia,\ Jianlong Lu,\  Ying Chen$^{\ast}$}\\[0.4em]
    {\small National University of Singapore}\\[0.3em]
    {\small $^{\ast}$Corresponding author: matcheny@nus.edu.sg}
\end{center}

\vspace{1em}
\vspace{1em}

\begin{abstract}
\noindent
High-dimensional simulation of multivariate extremes is fundamentally limited by the combinatorial complexity of dependence, often more than by the scarcity of extreme observations. We show that symmetry admits a lossless orbit-space representation that preserves structured extreme dependence while replacing an exponentially large dependence space with a compact set of symmetry classes. Based on this principle, we develop Q-Edge (Quantum Extreme Dependence Engine), a symmetry-reduced quantum framework that operates directly in orbit space, enabling scalable simulation and digital twins of structured extreme systems. By transferring symmetry into the data representation rather than the quantum circuit, Q-Edge allows unconstrained quantum generative models to exploit dramatically reduced state spaces. For a 30-dimensional problem, approximately 1.6 million angular states collapse to 256 orbit states, reducing the required quantum representation from about 21 qubits to 8. Our results establish a general computational principle for scalable quantum simulation of structured extreme dependence.
\end{abstract}

\vspace{0.5em}

\vspace{1.5em}

\section{Introduction}\label{sec:intro}

High-dimensional stochastic simulation underpins scientific discovery and decision making across finance, climate science, biology and engineering. In many applications, the events of greatest interest are not typical observations but rare, coordinated extremes, such as systemic financial crises, compound climate hazards, cascading infrastructure failures and simultaneous biological responses. Accurate simulation of these events requires reproducing not only their marginal behaviour but, more importantly, their dependence structure, which governs how extreme outcomes occur jointly across many interacting components. Despite remarkable advances in generative artificial intelligence, scalable simulation of multivariate extremes remains a fundamental computational challenge.

The primary obstacle is not the scarcity of extreme observations. Rather, it is the combinatorial growth of dependence complexity. In multivariate extreme value theory, dependence is represented by the angular distribution, whose discretized state space increases exponentially with dimension. As a consequence, both classical and quantum generative models face the same challenges, making faithful simulation progressively more difficult as dimensionality grows. Existing approaches therefore struggle to model nonlinear, heterogeneous and geometrically structured dependence in high dimensions.

Yet many real-world dependence structures are not arbitrary. Financial institutions, climate regions, biological populations and engineered systems frequently contain repeated components, modular organisation or statistically exchangeable groups. These structures imply that many nominally different dependence configurations are mathematically equivalent. Consequently, the effective complexity of the problem is determined not solely by the size of the ambient state space but also by the symmetry of the underlying dependence distribution. This observation suggests a fundamental question: \emph{can symmetry transform the complexity of simulating multivariate extremes?}

Here, in the context of synthetic data generation through a quantum machine, we answer this question affirmatively. We show that symmetry admits a lossless orbit-space representation that preserves structured extreme dependence while replacing an exponentially large dependence space with a compact set of symmetry classes. Rather than treating symmetry as an architectural constraint on a learning model, we establish it as a representation principle that fundamentally changes the complexity of the simulation problem itself. The reduction is exact within the corresponding symmetry class and applies before any learning algorithm is introduced.

Building on this principle, we develop Q-Edge (Quantum Extreme Dependence Engine), a symmetry-reduced quantum framework for simulating structured extreme dependence directly in orbit space. Unlike existing symmetry-aware quantum machine learning methods, which encode symmetry within quantum circuits through equivariant architectures or parameter constraints, Q-Edge transfers symmetry into the data representation itself. The resulting quantum model operates on a dramatically smaller state space while retaining the expressive power of unconstrained quantum generative models. For a representative 30-dimensional problem, approximately 1.6 million angular states reduce to only 256 orbit states, decreasing the required quantum representation from approximately 21 qubits to 8.

Across diverse synthetic and high-dimensional benchmarks, Q-Edge consistently reconstructs complex nonlinear, heterogeneous, multimodal and geometrically structured dependence more accurately than state-of-the-art classical and quantum generative models. More broadly, our work demonstrates that scalable simulation of multivariate extremes is fundamentally a problem of representation. By exploiting symmetry before learning, Q-Edge establishes a general computational framework for quantum simulation, generative modelling and digital twins of structured extreme systems.

\subsection{Related literature}

Our work lies at the intersection of four strands of research: generative simulation of complex distributions, multivariate extreme-value theory, symmetry-aware quantum machine learning, and quantum generative modelling. These strands can be organized around three unresolved questions concerning the scalable simulation of structured multivariate extremes.

\subsubsection{Can modern generative models faithfully simulate structured extreme dependence?} Variational autoencoders, generative adversarial networks, normalizing flows, and diffusion models have substantially expanded the ability to learn and simulate high-dimensional probability distributions \citep{kingma2013auto,goodfellow2020generative,
rezende2015variational,yang2023diffusion}.
These advances provide important computational foundations for scientific simulators and data-driven digital twins. Their strongest performance, however, generally concerns the bulk of a distribution, whereas multivariate extremes present a qualitatively different learning problem. Extreme observations are scarce, joint-tail probabilities are highly sensitive to errors in the dependence structure, and the relevant probability mass may be asymmetric, multimodal, or concentrated near lower-dimensional geometric
structures.

A growing literature adapts deep generative models specifically to extreme events. Copula-based normalizing flows have been developed to capture heavy-tailed marginals and asymmetric tail dependence
\citep{mcdonald2022comet}; extreme-value-aware variational autoencoders
separate radial and angular generation
\citep{lafon2026vae}; and recent approaches combine extreme-value theory with normalizing flows, geometric representations, generative adversarial networks, or deep learning
\citep{murphybarltrop2024deep,demonte2025generative,
lhaut2026simulation,hu2025gpdflow,wessel2025generative}. These methods significantly improve the flexibility of extreme-event simulation. Nevertheless, they primarily seek more expressive models on the original dependence space. Additionally, they do not directly resolve the growth of that space itself as the dimension increases. This leaves a central question: 
\emph{can structured extreme dependence be simulated scalably without first learning the full high-dimensional dependence space?}

\subsubsection{Can the intrinsic representation complexity of multivariate extremes be reduced?}

Multivariate extreme-value theory provides a rigorous framework for describing rare events and their joint dependence
\citep{rootzen2006multivariate,resnick2007heavy}.
Under multivariate regular variation, an extreme observation can be decomposed
into radial magnitude and angular direction, with the angular measure characterizing how extreme events occur jointly across dimensions.
This decomposition separates marginal tail intensity from extremal dependence
and provides the theoretical basis for many parametric and semiparametric
models.

Classical constructions include logistic, asymmetric logistic, Dirichlet,
H\"usler--Reiss, and related models. More recently, conditional-independence,
sparsity, and graphical structures have been used to construct parsimonious
high-dimensional extreme-value models
\citep{engelke2020graphical}. Deep-learning approaches have further increased
flexibility by learning angular distributions or geometric limit sets directly
\citep{lafon2026vae,murphybarltrop2024deep,
demonte2025generative, wessel2025generative, lhaut2026simulation}.
These developments address statistical and computational tractability through
parametric assumptions, graphical factorization, sparsity, or flexible
function approximation.

However, a more fundamental representational question remains largely open.
The angular state space contains many configurations that may be statistically
equivalent whenever the underlying system has repeated components,
exchangeable populations, modular organization, or other group symmetries.
Existing approaches generally model these configurations separately unless a
specific parametric or graphical restriction is imposed. It is therefore
unclear whether the dependence space itself can be reduced, prior to model
estimation, without discarding information relevant to the structured target
distribution. In other words, \emph{does structured extreme dependence admit a
smaller, exact representation whose complexity is determined by intrinsic
symmetry rather than ambient dimension?}

\subsubsection{Can symmetry reduce quantum learning complexity?}

Quantum machine learning has pursued two complementary directions for
improving the representation of complex probability distributions.
The first exploits symmetry as an inductive bias through geometric and
equivariant quantum machine learning
\citep{larocca2022group,nguyen2024theory}. In this paradigm,
parameterized quantum circuits are designed so that data encodings,
trainable gates, or observables transform consistently under a prescribed
group action. Such symmetry-aware architectures can improve
generalization, parameter efficiency, and trainability. For permutation
symmetry, equivariant quantum circuits have been shown to reach
overparameterized regimes efficiently and, under suitable conditions, to
avoid barren plateaus
\citep{meyer2023exploiting,schatzki2024theoretical}. Representation-theoretic
and dynamical-Lie-algebra analyses further clarify how symmetry
constraints influence circuit expressivity and optimization
\citep{larocca2022diagnosing,ragone2024lie}.

The second direction seeks more expressive quantum generative models.
Quantum circuit Born machines \citep{liu2018differentiable}, quantum
generative adversarial networks \citep{zoufal2019quantum}, and
variational quantum circuits \citep{mitarai2018quantum} exploit the Born
rule to represent and sample complex probability distributions directly
from quantum states. These models provide flexible implicit
representations and have demonstrated considerable promise for
high-dimensional generative learning. Nevertheless, they typically learn
over the original exponentially large sample space, where trainability,
measurement complexity, and model expressivity remain increasingly
challenging as the dimension grows.

Together, these two directions highlight a fundamental trade-off.
Symmetry-aware quantum learning improves efficiency by constraining the
learning architecture, whereas unconstrained quantum generative models
maximize expressive power while retaining the full complexity of the
original state space. An important question therefore remains unresolved:
\emph{can symmetry reduce quantum learning complexity without
constraining the quantum model itself?} More specifically, \emph{can symmetry
be exploited before learning begins to transform the target probability
distribution into a smaller, information-preserving representation,
allowing an unconstrained quantum generative model to operate directly on
the reduced state space?}

\subsubsection{A unified answer through symmetry-reduced representation}

Q-Edge provides a unified answer to the three questions posed above. The central distinction of Q-Edge is that it exploits symmetry before learning, replacing the original angular sample space with a lossless orbit-space representation. Rather than treating symmetry as an inductive bias encoded in the learning architecture, Q-Edge establishes symmetry as a
representation principle that fundamentally changes the complexity of the
simulation problem. This contrasts with statistical approaches that improve
tractability through parametric assumptions, graphical factorization, or
sparsity, and with quantum approaches that impose symmetry through
equivariant circuit architectures.

The theoretical foundation of Q-Edge is that a group action partitions the
discretized angular space into symmetry orbits, yielding an exact
representation of structured extreme dependence on the corresponding orbit
space. Consequently, symmetry-equivalent configurations need not be learned
individually. Instead, Q-Edge maps each orbit to a single
computational-basis state and trains a quantum circuit Born machine \citep{liu2018differentiable} directly
on the reduced probability space. The quantum circuit therefore remains
completely unconstrained, retaining the expressive power of a generic
quantum generative model while operating on a substantially smaller state
space.

This representation-first strategy simultaneously addresses the limitations
identified in the preceding literature. It enables scalable simulation of
structured extreme dependence without separately learning every
symmetry-equivalent configuration; it reduces representation complexity
exactly rather than through additional modelling assumptions; and it lowers
the quantum resources required for simulation by reducing the effective size
of the target probability space. The resulting conceptual framework is depicted in Figure~\ref{fig:placeholder}

\begin{figure}
    \centering
    \includegraphics[width=\linewidth]{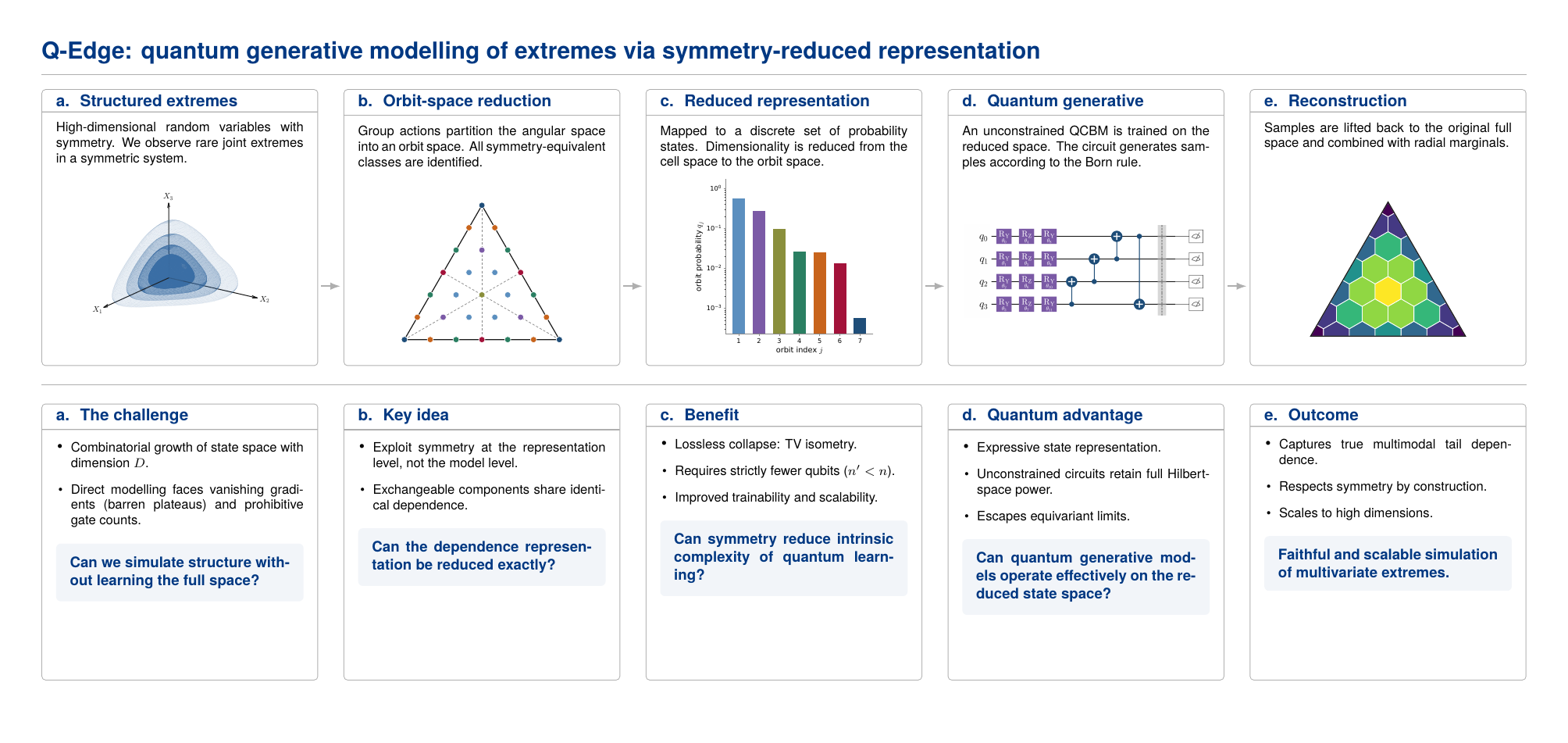}
    \caption{Conceptual framework of Q-Edge. Structured symmetry induces a lossless orbit-space representation, reducing the effective state space before quantum learning. An unconstrained quantum generative model is trained on the reduced representation to enable scalable simulation of structured multivariate extremes.}
    \label{fig:placeholder}
\end{figure}

Q-Edge therefore connects multivariate extreme-value theory, symmetry-aware representation, and quantum generative modelling through a common
computational principle: when dependence possesses exploitable structure, the appropriate object to learn is not the original high-dimensional distribution, but its exact representation over symmetry classes.

\subsection{Structure of the Paper}
The remainder of the paper is organized as follows. Section \ref{sec:method} introduces the Q-Edge framework and presents the symmetry-preserving orbit-space representation for multivariate extreme dependence. Section \ref{sec:model} develops the quantum generative learning architecture and discusses its implementation. Section \ref{sec:results} presents empirical results across a range of dependence structures, including nonlinear, heterogeneous, multimodal, geometrically constrained and high-dimensional settings. Section \ref{sec:discussion} concludes with implications for quantum digital twins, symmetry-aware learning and future research directions.

\section{A Representation Principle for Structured Extreme Dependence}
\label{sec:method}

This section develops the theoretical foundation of Q-Edge. We show that exchangeability induces an exact orbit-space representation of
multivariate extreme dependence, allowing the original dependence
distribution to be represented losslessly on a substantially smaller
state space. This symmetry-preserving representation changes the
intrinsic complexity of the learning problem while preserving the
dependence structure exactly. The resulting representation principle
provides the basis for the quantum learning framework developed in the
following section.

\subsection{The Representation Bottleneck}
\label{subsec:angular}

The representation of multivariate extremes begins with the classical
polar decomposition of multivariate extreme value theory. For a $D$-dimensional $(D\geq2)$ random
non-zero vector
$\mathbf{X}\in\mathbb{R}_{\ge0}^{D}$,
an extreme event can be decomposed into its overall magnitude and its
dependence structure through

\[
\mathbf{X}
=
R\,\boldsymbol{\Theta},
\quad
R=\|\mathbf X\|_1,
\quad
\boldsymbol{\Theta}
=
\frac{\mathbf X}{\|\mathbf X\|_1}.
\]

The radial variable $R$ describes the severity of the extreme event,
whereas the angular variable $\boldsymbol{\Theta}$ characterizes how the extremal magnitude is distributed across the $D$ variables. Since, under the regularly varying assumption, the radial component can be modelled independently using standard
univariate extreme-value techniques, our interest lies entirely in the
angular component, which completely specifies the dependence structure
of multivariate extremes. The angular variable takes values on the probability simplex
\[
\Delta^{D-1}
=
\left\{
\boldsymbol{\theta}\in\mathbb R_{\ge0}^{D}:
\sum_{i=1}^{D}\theta_i=1
\right\},
\]whose dimension is $D-1$. Each point on the simplex represents one
possible allocation of the extremal magnitude among the variables,
thereby providing a geometric representation of multivariate extreme
dependence.

To enable statistical learning and quantum computation, the continuous
simplex are usually approximated by a finite state space. Following the
standard lattice discretization, we partition $\Delta^{D-1}$ using the
natural lattice

\begin{equation}
\mathcal X_N
=
\left\{
\mathbf k=(k_1,\ldots,k_D)\in\mathbb Z_{\ge0}^{D}:
\sum_{i=1}^{D}k_i=N
\right\},
\label{eq:lattice}
\end{equation}where $N$ denotes the discretization resolution. Normalizing each
lattice point by $N$ establishes a one-to-one correspondence between
$\mathcal X_N$ and the rational lattice points on
$\Delta^{D-1}$. Consequently, every lattice point represents one Euclidean Voronoi cell of a partition of the simplex, and a probability distribution over
cells provides a finite-dimensional approximation of the
underlying angular measure. The consistency of this discretization is
established in Appendix~\ref{sec:consistency}.

The complexity of this representation is determined by the cardinality
of the natural lattice,
\begin{equation}
|\mathcal X_N|
=
\binom{N+D-1}{D-1},
\label{eq:lattice_size}
\end{equation}
which grows combinatorially with both the dimension $D$ and the
discretization resolution $N$. Even moderate dimensions produce an extremely large configuration space, making direct statistical learning increasingly difficult. Consequently, improving the fidelity of the approximation rapidly enlarges the number of admissible dependence configurations that must be represented. This growth occurs before any learning or optimization is performed, indicating that the principal challenge is fundamentally one of representation rather than
computation. 

\subsection{The Combinatorial Barrier}
\label{subsec:barrier}

The natural lattice provides a finite representation of multivariate extreme dependence. The remaining
question is whether this representation can be learned efficiently. As we show below, the answer is generally negative. The combinatorial growth of the discretized dependence space propagates directly to the
quantum representation, resulting in exponential state-preparation costs and severe optimization difficulties. These computational bottlenecks all stem from a common source: the assumption that every dependence
configuration must be represented as an independent quantum state.

To enable quantum learning, the discretized dependence space is embedded
into the computational basis of a quantum register. An $n$-qubit system
possesses exactly $2^n$ computational basis states, and therefore the
smallest compatible register satisfies

\begin{equation}
n
=
\left\lceil
\log_2 |\mathcal X_N|
\right\rceil.
\label{eq:nqubits}
\end{equation}
Since the lattice cardinality is generally not a power of two, we
augment the natural lattice with additional representative centroids $\mathcal{P}_N$,
obtaining an augmented configuration space.
\begin{equation}
\mathfrak Q_N=\mathcal X_N \cup \mathcal P_N, 
\quad
|\mathfrak Q_N|
=
2^n.
\label{eq:augmented_grid}
\end{equation}
The augmentation establishes a one-to-one correspondence between lattice configurations and computational basis states while preserving the geometric structure of the discretized simplex.

This embedding exposes the computational consequences of the
representation. A generic probability distribution over
$\mathfrak Q_N$ occupies a $2^n$-dimensional Hilbert space. Preparing an
arbitrary quantum state in such a space requires
$O(2^n)$ elementary quantum
gates~\citep{grover2000synthesis,plesch2011quantum,sanders2019black},
reflecting the exponential number of independent amplitudes that must be
specified. Consequently, state preparation alone rapidly becomes the
dominant computational cost as the representation grows.

The challenge extends beyond state preparation to quantum optimization.
Highly expressive parameterized quantum circuits approach Haar-random
unitaries, for which the variance of the training gradient decreases
exponentially with the number of qubits, giving rise to the
\emph{barren plateau}
phenomenon~\citep{mcclean2018barren,cerezo2021cost}. As the
representation size increases, optimization becomes progressively more
difficult, while finite coherence times and gate noise further restrict
the depth of trainable quantum circuits on noisy intermediate-scale
quantum (NISQ) devices.

Importantly, these computational bottlenecks share a common origin.
They do not arise because quantum circuits are inherently inefficient,
but because every lattice configuration is represented as an independent
quantum state. In structured multivariate extremes, however, many
configurations encode identical dependence patterns under permutations
of statistically equivalent variables. Treating all such configurations
separately introduces substantial representational redundancy, which is
then inherited by the quantum model.

The implication is fundamental. Scalability cannot be achieved merely by designing more expressive quantum circuits, because the computational complexity is inherited from the representation itself rather than the learning architecture. This naturally raises the following question: are all dependence configurations genuinely distinct? As we show in the next subsection, for multivariate data exhibiting exchangeability, statistically equivalent dependence configurations can be identified
and represented by a much smaller orbit space without sacrificing any statistical information.

\subsection{Exchangeability and the Orbit-Space Representation}
\label{subsec:orbit}

For many multivariate, treating every discretized dependence
configuration as distinct is unnecessary. Financial assets within the
same sector
\citep{poon2004extreme,engelke2020graphical}, geographically or
physically similar sites in environmental and hydrological studies
\citep{tawn1990modelling,colestawn1991modelling}, and interchangeable
components in compositional data arising in geology, biology and the
geosciences \citep{aitchison1982statistical} may play statistically
equivalent roles. In such settings, permuting the variable labels
changes the representation but not the underlying dependence structure.
This symmetry is captured by \emph{exchangeability}.

A simple example illustrates the resulting redundancy. At resolution
$N=6$ in dimension $D=3$, the six configurations
\[
(1,2,3),\ (1,3,2),\ \ldots,\ (3,2,1)
\]
differ only by permutations of their coordinates and therefore describe
the same allocation pattern under the action of the symmetric group
$S_3$. Treating them as six independent learning targets enlarges the
representation without adding statistical information. The same
permutation symmetry arises in standard dependence models, including
the symmetric Dirichlet family, the symmetric logistic model for
multivariate extremes
\citep{gumbel1960bivariate,tawn1990modelling}, and symmetric forms of
the logistic-normal distribution
\citep{aitchison1980logistic,aitchison1982statistical}.

To formalize this idea, let $S_D$ denote the symmetric group acting on
the $D$ coordinates by permutation. A distribution $P$ on the action closed augmented
configuration space $\mathcal Q_N$ is exchangeable when its probability
assignment is unchanged by this action.

\begin{definition}[$S_D$-symmetric distribution]
\label{def:symmetric_dist_main}
A probability distribution $P$ on $\mathcal Q_N$ is
\emph{$S_D$-symmetric}, or equivalently $S_D$-exchangeable, if
\begin{equation}
P(\sigma\cdot\mathbf{x})
=
P(\mathbf{x})
\qquad
\text{for all }
\sigma\in S_D
\text{ and }
\mathbf{x}\in\mathcal Q_N.
\label{eq:symmetric_distribution}
\end{equation}
The set of all such distributions is denoted
$\Delta_{\mathrm{sym}}$.
\end{definition}

Exchangeability allows configurations related by permutations to be
grouped together. The \emph{orbit} of
$\mathbf{x}\in\mathcal Q_N$ is
\begin{equation}
\mathrm{Orb}(\mathbf{x})
=
\{
\sigma\cdot\mathbf{x}:\sigma\in S_D
\}.
\label{eq:orbit_main}
\end{equation}
The distinct orbits partition $\mathcal Q_N$ into equivalence classes,
and each orbit represents one dependence pattern together with all of
its possible coordinate labellings.

The resulting collection of symmetry classes is the orbit space
\[
\mathcal Q_N/S_D.
\]
We denote its cardinality by \(
K
=
|\mathcal Q_N/S_D|
=
\mathcal O_D(\mathcal Q_N).
\)
Explicit orbit counts and the extension to subgroups
$G\leq S_D$, including block-symmetric groups associated with partial
exchangeability, are provided in the Appendix
(Theorem~\ref{thm:general_burnside},
Theorem~\ref{thm:augmented_orbits}, and
Proposition~\ref{prop:augmented_orbit_general}).

Because an exchangeable distribution is constant within each orbit, it
is completely determined by one probability value per orbit rather than
one value per configuration. Symmetry therefore changes the intrinsic
dimension of the statistical model.

\begin{proposition}[Dimension under exchangeability]
\label{prop:sym_dim_main}
The space of $S_D$-symmetric probability distributions on
$\mathcal Q_N$ has dimension
\begin{equation}
\dim(\mathrm{aff} \ \Delta_{\mathrm{sym}})
=
K-1
=
\mathcal O_D(\mathcal Q_N)-1.
\label{eq:sym_dim_main}
\end{equation}
\end{proposition}
A proof is given in the Appendix  (Proposition~\ref{prop:sym_dim}). Proposition~\ref{prop:sym_dim_main}
shows that the relevant representation complexity is governed by the
number of distinct symmetry classes, rather than by the total number of
lattice configurations. When the orbits are large, this produces a
substantial reduction in dimension before any learning algorithm is
applied.

\subsection{Lossless Orbit-Space Representation}
\label{subsec:orbitrepresentation}

The orbit construction reduces the number of configurations that must be
represented. The remaining question is whether this reduction discards
statistical information. For exchangeable distributions, it does not:
the orbit space provides an exact reformulation of the original learning
problem.

Let $\mathrm{Orb}_1,\ldots,\mathrm{Orb}_K$ denote the distinct orbits in $\mathcal Q_N$, and let$s_j
=|\mathrm{Orb}_j|$ be the size of the $j$th orbit. The probability simplex over the orbit
space is
\begin{equation}
\Delta^{K-1}
=
\left\{
\mathbf q=(q_1,\ldots,q_K):
q_j\geq 0,\ 
\sum_{j=1}^{K}q_j=1
\right\}.
\label{eq:orbit_simplex}
\end{equation}

Here, $q_j$ denotes the total probability mass assigned to the entire
orbit $\mathrm{Orb}_j$. To recover a distribution on the original
configuration space, this mass is distributed uniformly among the
configurations belonging to the orbit. Let $j(\mathbf{x})$ denote the
index of the orbit containing $\mathbf{x}$. We define the orbit map
\begin{equation}
\Phi:
\Delta^{K-1}
\longrightarrow
\Delta_{\mathrm{sym}},
\quad
[\Phi(\mathbf q)](\mathbf{x})
=
\frac{q_{j(\mathbf{x})}}
{s_{j(\mathbf{x})}}.
\label{eq:phi_def}
\end{equation}
The map $\Phi$ converts an orbit-level probability distribution into an
exchangeable distribution on $\mathcal Q_N$. The following result shows
that this representation is lossless.

\begin{theorem}[Lossless orbit-space representation]
\label{thm:phi_isometry}
The orbit map
\[
\Phi:
\Delta^{K-1}
\longrightarrow
\Delta_{\mathrm{sym}}
\]
is a bijection. Thus, every $S_D$-symmetric distribution on
$\mathcal Q_N$ has a unique representation in the orbit simplex.

Moreover, $\Phi$ preserves total variation distance:
\begin{equation}
\mathrm{TV}
\bigl(
\Phi(\mathbf q),
\Phi(\mathbf q')
\bigr)
=
\mathrm{TV}
\bigl(
\mathbf q,
\mathbf q'
\bigr)
\label{eq:tv_isometry}
\end{equation}
for all
$\mathbf q,\mathbf q'\in\Delta^{K-1}$.
It also preserves every $f$-divergence for which the divergence is
well defined, including the symmetric Kullback--Leibler divergence used during
training.
\end{theorem}
The proof, together with the corresponding weighted-inner-product
isometry, is given in the Appendix
(Theorem~\ref{thm:orbit_symmetric}).

Theorem~\ref{thm:phi_isometry} establishes that orbit-space
representation is not an approximation or a lossy compression. Every
exchangeable distribution corresponds to exactly one orbit-space
distribution, and the statistical discrepancies used for learning are
preserved. Optimizing total variation or KL divergence in the orbit
space is therefore equivalent to optimizing the same objective in the
full configuration space.

The reduction is purely representational. A general distribution on
$\mathcal Q_N$ requires
\[
|\mathcal Q_N|-1
\]
degrees of freedom, whereas an exchangeable distribution requires only
\[
K-1
=
\mathcal O_D(\mathcal Q_N)-1.
\]
The corresponding representation ratio is
\begin{equation}
\rho_N
=
\frac{
\mathcal O_D(\mathcal Q_N)
}{
|\mathcal Q_N|
}.
\label{eq:reductionratio}
\end{equation}
When many configurations belong to the same symmetry orbit,
$\rho_N< 1$, yielding a substantial reduction in both classical and
quantum representation complexity.

This result identifies exchangeability as a representation principle,
rather than merely an inductive bias imposed on a learning architecture.
Q-Edge removes symmetry-induced redundancy before learning begins and
operates directly on the smaller orbit simplex. The next section
develops the quantum realization of this lossless representation.

\section{Q-Edge: A Symmetry-Aware Quantum Representation}
\label{sec:model}

The previous section established the representation principle of
Q-Edge. Exchangeability identifies statistically equivalent dependence
configurations, the orbit-space representation removes this redundancy
without loss of information, and the resulting orbit simplex provides a
compact representation of multivariate extreme dependence. We now show
how this representation is realized on a quantum computer.

Q-Edge consists of four stages (Fig.~\ref{fig:flowchart}). First, extreme observations
are represented through the angular component of multivariate extreme
value theory. Second, the angular simplex is discretized into a quantum
representation compatible with finite qubit registers. Third, we show
that representing arbitrary dependence distributions is fundamentally
impossible for hardware-efficient quantum circuits because of an
intrinsic representation barrier. Finally, Q-Edge overcomes this barrier
by learning directly in the orbit space, where exchangeability reduces
the statistical complexity before any quantum resources are used.

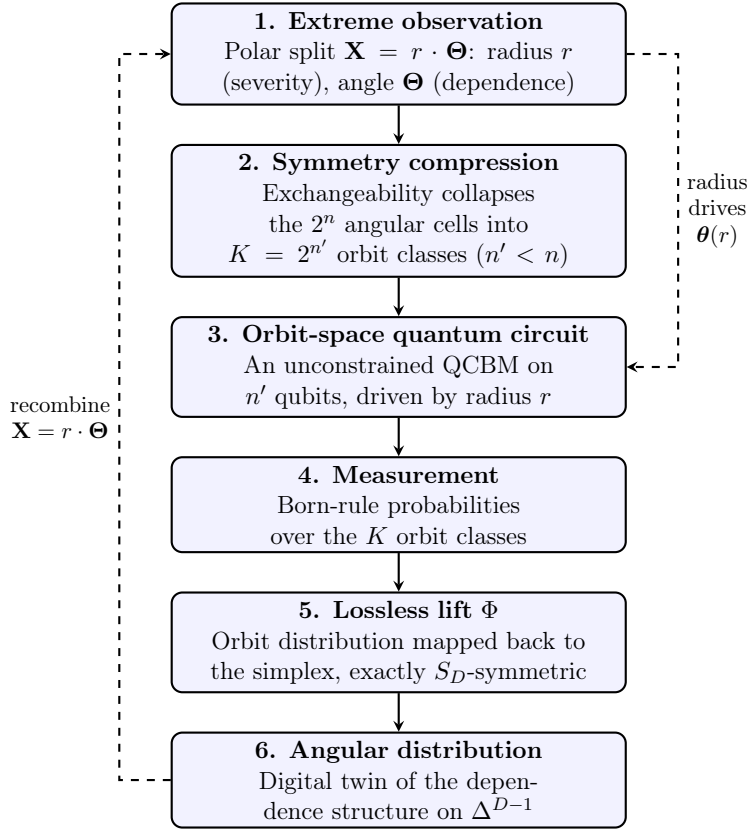
\begin{figure}[htbp]
    \centering
    \begin{tikzpicture}[
        scale=0.85, transform shape,
        node distance=0.6cm,
        box/.style={rectangle, draw, thick, align=center,
            minimum height=1cm, text width=6.8cm, fill=blue!5,
            rounded corners},
        emph/.style={rectangle, draw, very thick, align=center,
            minimum height=1cm, text width=6.8cm, fill=orange!12,
            rounded corners},
        arrow/.style={->, thick, >=stealth}
    ]
        \node[box] (input) {\textbf{1. Extreme observation}\\
            Polar split $\mathbf{X} = r\cdot\boldsymbol{\Theta}$:
            radius $r$ (severity), angle $\boldsymbol{\Theta}$ (dependence)};
        \node[box, below=of input] (compress) {\textbf{2. Symmetry compression}\\
            Exchangeability collapses the $2^n$ angular cells
            into $K=2^{n'}$ orbit classes ($n'\!<\! n$)};
        \node[box, below=of compress] (circuit) {\textbf{3. Orbit-space quantum circuit}\\
            An unconstrained QCBM on $n'$ qubits, driven by radius $r$};
        \node[box, below=of circuit] (probs) {\textbf{4. Measurement}\\
            Born-rule probabilities over the $K$ orbit classes};
        \node[box, below=of probs] (lift) {\textbf{5. Lossless lift} $\Phi$\\
            Orbit distribution mapped back to the simplex,
            exactly $S_D$-symmetric};
        \node[box, below=of lift] (output) {\textbf{6. Angular distribution}\\
            Digital twin of the dependence structure on $\Delta^{D-1}$};

        \draw[arrow] (input) -- (compress);
        \draw[arrow] (compress) -- (circuit);
        \draw[arrow] (circuit) -- (probs);
        \draw[arrow] (probs) -- (lift);
        \draw[arrow] (lift) -- (output);

        \draw[arrow, dashed] (input.east) -- ++(0.8,0)
            |- (circuit.east)
            node[pos=0.25, right, align=center]
            {\small radius\\[-0.4ex]\small drives\\[-0.4ex]\small $\bm\theta(r)$};
        \draw[arrow, dashed] (output.west) -- ++(-0.8,0)
            |- (input.west)
            node[pos=0.25, left, align=center]
            {\small recombine\\[-0.4ex]\small $\mathbf{X}=r\cdot\boldsymbol{\Theta}$};
    \end{tikzpicture}
    \caption{The Q-Edge pipeline. An extreme observation is split by
    polar decomposition into a radial severity $r$ and an angular
    dependence component $\boldsymbol{\Theta}$. Exchangeability compresses
    the angular configuration space into a much smaller set of orbit
    classes, on which an unconstrained quantum circuit
    Born machine---driven by the radius---is trained; its output is
    lifted back to the simplex by the total-variation isometry $\Phi$,
    yielding an exactly symmetric digital twin of the dependence
    structure.}
    \label{fig:flowchart}
\end{figure}
\subsection{Representing Extremes}
\label{subsec:extreme}

Following the polar decomposition of multivariate extreme value theory,
Q-Edge models only the angular component
$\Theta\in\Delta^{D-1}$ describing extremal dependence.
The radial magnitude $R$ is asymptotically independent of $\Theta$ and
is therefore modelled separately. The complete extreme sample is
recovered as
\[
\mathbf X = R\Theta.
\]

The angular distribution is represented by a quantum circuit Born
machine (QCBM),
\[
p_{\boldsymbol\theta}(x)
=
\left|
\langle x|
U(\boldsymbol\theta(r))
|0\rangle^{\otimes n}
\right|^2,
\]
where the circuit parameters depend smoothly on the radial magnitude
through
\[
\boldsymbol\theta(r)
=
\sum_{i=0}^{T}
\boldsymbol\theta^{(i)}g_i(r),
\qquad
g_i(r)=\frac1{(1+r)^i}.
\]

This parameterization ensures that the learned angular distribution 
approaches a radius-independent limit as
$r\rightarrow\infty$, consistent with extreme value theory.
As a design choice we therefore restrict to a
hardware-efficient ansatz of gate complexity $O(n^2)$.

\begin{theorem}[Lipschitz continuity]
\label{thm:lipschitz}
Let $P_r$ be the Born distribution of our parameterized circuit. The Born distribution varies continuously with the radial magnitude.
Specifically,
\[
\mathrm{TV}(P_{r_1},P_{r_2})
\le
\frac12\Lambda_T|r_1-r_2|,
\]
where $\Lambda_T=\sum_{i=1}^{T} i\,\|\bm\theta^{(i)}\|_1$ depends only on the parameterization coefficients.
A refined bound and proof are given in the Appendix.
\end{theorem}

The theorem provides a theoretical guarantee that nearby extreme events
produce nearby angular distributions, enabling stable conditioning on
event.

\subsection{Representing the Angular Simplex}
\label{subsec:simplex}

To represent angular dependence on a quantum computer, the continuous
simplex is discretized into a finite collection of Voronoi cells, each corresponding to one computational basis state.

A natural lattice of resolution $N$ contains
$\binom{N+D-1}{D-1}$ cells, which in general is not a power of two,
satisfying
$
2^{n-1}<\binom{N+D-1}{D-1}\le2^n.
$
Consequently, a direct encoding either discards surplus basis states or
introduces penalty terms, both of which reduce representation
efficiency.

Q-Edge resolves this mismatch by augmenting the lattice with
geometrically canonical centroid points. For a base lattice point
$\mathbf{k}\in\mathbb{Z}_{\ge0}^D$ and a type-$m$ elementary simplex,
the corresponding centroid is

\begin{equation}
\label{eq:centroid_main}
G_m(\mathbf{k})
=
\mathbf{k}
+
\frac{m}{D}\mathbf1, \quad m\in\{1,\dots,D-1\},
\end{equation}
where the derivation is given in Appendix~\ref{sec:equivariance}. Selecting an appropriate
subset of these centroids produces an augmented grid
$\mathcal Q_N$ satisfying
$
|\mathcal Q_N|=2^n,
$
thereby establishing a bijection
$
\phi:\mathcal Q_N\leftrightarrow\{0,1\}^n.
$
Every computational basis state is therefore assigned to exactly one
simplex cell, eliminating both unused basis states and artificial
penalty terms. The resulting discretization is statistically consistent.

\begin{theorem}[$L^1$ convergence of the discretization]
\label{thm:l1}
Let $\mu\ll\lambda_{D-1}$ with density
$f\in L^1(\Delta^{D-1})$, and let $f_N$ denote the piecewise-constant
projection induced by $\mathcal Q_N$. Then
\[
\|f_N-f\|_{L^1(\Delta^{D-1})}\rightarrow0,
\qquad N\rightarrow\infty.
\]
\end{theorem}
The proof is given in Appendix~\ref{sec:consistency}. Consequently, the total approximation error separates naturally into

\begin{equation}
\label{eq:tvdecomp}
d_{\mathrm{TV}}
(\tilde\mu_N^{\mathrm{QCBM}},\mu)
\le
\underbrace{
d_{\mathrm{TV}}
(\tilde\mu_N^{\mathrm{QCBM}},\tilde\mu_N)
}_{\text{representation}}
+
\underbrace{
d_{\mathrm{TV}}
(\tilde\mu_N,\mu)
}_{\text{discretization}}.
\end{equation}

By Theorem~\ref{thm:l1}, the discretization error vanishes
asymptotically. The remaining limitation therefore lies not in the
representation of the simplex, but in the expressive power of the
quantum model.
\subsection{The Fundamental Representation Barrier}
\label{subsec:barrier}

The augmented discretization developed in the previous subsection
eliminates approximation error in the representation of the angular
simplex. A natural question therefore arises: can a hardware-efficient
quantum circuit accurately represent an arbitrary dependence
distribution?

The following theorem shows that the answer is fundamentally negative.
The limitation is not a consequence of a particular circuit design or
optimization algorithm, but of the intrinsic dimension of the
probability simplex itself.

\begin{theorem}[Fundamental representation barrier]
\label{thm:inapprox}

Let $\mathcal{B}$ denote the family of probability distributions
reachable by a parameterized quantum circuit with $N_p$ trainable
parameters acting on $n$ qubits, and let $M=2^n$. If
$N_p<M-1$, then there exists a probability distribution
$\mu\in\Delta^{M-1}$ satisfying
\[
d_{\mathrm{TV}}(\mu,\mathcal B)
\ge
\varepsilon^*,
\]
where \(
\varepsilon^*
=
(1/8)^{1/(1-\alpha)}
(\pi N_p+1)^{-\alpha/(1-\alpha)},\;
\alpha=\frac{N_p}{M-1}<1.
\)

\end{theorem}
The proof is given in Appendix~\ref{sec:reachable}.

Theorem~\ref{thm:inapprox} establishes a fundamental representation
barrier for quantum generative models. For the hardware-efficient scaling $N_p = O(\mathrm{poly}(n))$, $\varepsilon^{*}\to 1/8$ as $n\to\infty$. Any circuit with only
polynomially many trainable parameters is necessarily unable to
represent arbitrary probability distributions on the full
$2^n$-state simplex. This limitation is independent of the choice of
ansatz, optimizer or training algorithm; it arises solely from the
intrinsic dimension of the statistical model.

Consequently, improving the quantum circuit alone cannot overcome this
barrier. The representation itself must change. The next subsection
shows that exchangeability provides exactly such a change: by replacing
the original configuration space with its orbit-space representation,
the effective dimension is reduced from $2^n-1$ to the orbit count,
making accurate learning compatible with hardware-efficient quantum
circuits.
\subsection{Escaping the Barrier through Symmetry}
\label{subsec:escape}

Q-Edge overcomes the representation barrier by exploiting the
exchangeability developed in Section~2.
Rather than learning over the original configuration space,
the quantum circuit operates directly on the orbit simplex.

The orbit-space representation replaces the original
$2^n$
configuration states by
\[
K
=
\mathcal O_D(\mathcal Q_N)
\]
orbit representatives, requiring only
$n'= \log_2K$ qubits.
Each computational basis state now represents an entire symmetry class
rather than a single lattice configuration.

This leads to a fundamentally different learning problem.
Instead of approximating distributions on the full simplex, the circuit
learns only the orbit-space distribution.
By Theorem~\ref{thm:phi_isometry}, this representation is lossless, so
optimization in orbit space is mathematically equivalent to learning the
original exchangeable distribution.

The resulting architecture differs fundamentally from existing
equivariant quantum circuits.
Rather than constraining the circuit to preserve symmetry,
Q-Edge removes symmetry-induced redundancy before learning begins.
The quantum circuit therefore remains completely unconstrained while
producing an exactly exchangeable distribution through the orbit-space
representation.

\begin{theorem}[Orbit-space dominance]
\label{thm:dominance}
At the level of abstract unitary ansatz families, every
$S_D$-equivariant circuit family with a fixed invariant pure initial
state admits an orbit-space realization using no additional trainable
parameters. For $K=2^{n'}$, this realization uses $n'$ qubits and
reproduces the same grid distributions under $\Phi$, with strictly
fewer qubits whenever $K<2^n$.
\end{theorem}

The construction is given in Appendix~\ref{sec:equivariance}, Proposition~\ref{prop:dominance} and Remark~\ref{rem:no_fixed_point};
universal orbit-space coverage is established in Corollary~\ref{cor:dominance}.
Compiled gate counts and depth are discussed separately in Remark~\ref{rem:compile}.

The degree of compression depends only on the underlying symmetry.
Fully symmetric dependence models achieve more than forty-fold
reductions in representation size, whereas asymmetric models recover the
original representation (Table~\ref{tab:compression}).

The key innovation of Q-Edge is therefore not a more expressive quantum
circuit but a different representation on which the circuit operates.
By exploiting exchangeability before learning begins, Q-Edge transforms
an intractable representation problem into one compatible with
near-term quantum hardware while preserving the statistical structure of
multivariate extremes exactly.

\subsection{Symmetry as an escape: the orbit-space construction}\label{sec:orbit}

The representation barrier established in
Theorem~\ref{thm:inapprox} applies to arbitrary probability
distributions on the full simplex. Many dependence distributions arising in practice, however, exhibit exchangeability or partial
exchangeability. The orbit-space representation
introduced in Section~\ref{sec:method} changes the statistical model itself by
replacing the full simplex with the symmetric simplex
$\Delta_{\rm sym}$, whose dimension is

\begin{equation}
\label{eq:sym_dim_G}
\dim(\mathrm{aff} \ \Delta_{\rm sym})
=
\mathcal O_G(\mathcal Q_N)-1,
\end{equation}
where $\mathcal O_G(\mathcal Q_N)$ denotes the number of symmetry
orbits. Since
$\mathcal O_G(\mathcal Q_N)<2^n$
for highly symmetric problems, the intrinsic dimension of the target
distribution is dramatically reduced.

The consequence is immediate. The representation barrier of
Theorem~\ref{thm:inapprox} applies to the intrinsic dimension of the
target simplex. Replacing the original simplex by
$\Delta_{\rm sym}$ changes the necessary parameter requirement from
\(
N_p\ge2^n-1
\)
to
\(
N_p
\ge
\mathcal O_G(\mathcal Q_N)-1.
\)
Thus the escape from the representation barrier comes not from a more
expressive quantum circuit, but from representing a lower-dimensional
statistical model.

The amount of compression depends only on the orbit count. For the
Dirichlet family, the orbit count can be computed explicitly using
Burnside's lemma (Appendix~\ref{sec:equivariance} Theorem~\ref{thm:dirichlet_orbits}). Table~\ref{tab:compression} illustrates the
continuous transition from fully asymmetric to fully symmetric targets.
For example, when
$D=30$
and
$N=6$
with block symmetry
$G=S_{10}^3$,
the original discretization contains approximately
$1.6\times10^6$
cells (about 21 qubits), whereas the orbit-space representation
contains only
$256=2^8$
states.

\begin{table}[t]
\centering
\caption{Orbit-space compression across the symmetry spectrum
($D=6$, $N=6$, $|\mathcal{X}_6|=462$).}
\label{tab:compression}
\small
\begin{tabular}{lcccc}
\toprule
Target & Symmetry $G$ & $|G|$ & $\bO_G$ & Compression \\
\midrule
$\mathrm{Dir}(\alpha_1,\alpha_1,\alpha_1,\alpha_1,\alpha_1,\alpha_1)$ & $S_6$ & $720$ & $11$ & $\mathbf{42\times}$ \\
$\mathrm{Dir}(\alpha_1,\alpha_1,\alpha_1,\alpha_2,\alpha_2,\alpha_2)$ & $S_3\times S_3$ & $36$ & $49$ & $9.4\times$ \\
$\mathrm{Dir}(\alpha_1,\alpha_1,\alpha_2,\alpha_2,\alpha_3,\alpha_3)$ & $S_2^3$ & $8$ & $119$ & $3.9\times$ \\
$\mathrm{Dir}(\alpha_1,\alpha_2,\alpha_3,\alpha_3,\alpha_4,\alpha_4)$ & $S_2^2$ & $4$ & $188$ & $2.5\times$ \\
$\mathrm{Dir}(\alpha_1,\alpha_2,\alpha_3,\alpha_4,\alpha_5,\alpha_6)$ & $\{e\}$ & $1$ & $462$ & $1\times$ \\
\bottomrule
\end{tabular}
\end{table}

To realize this reduced representation on a quantum computer, we encode
the orbit indices directly rather than the original simplex cells.
Choosing
$K=\mathcal O_G(\mathcal Q_N)=2^{n'}$,
we introduce a bijection $
\psi:\{1,\ldots,K\}
\leftrightarrow
\{0,1\}^{n'},$ and represent the Born distribution over orbit indices,
$
q_{\boldsymbol\theta}(j)
=
\left|
\langle
\psi(j)
|
U(\boldsymbol\theta)
|0\rangle^{\otimes n'}
\right|^2.
$ The probability assigned to each orbit is then distributed uniformly
over its members,

\[
P_{\boldsymbol\theta}(\mathbf x)
=
\frac{
q_{\boldsymbol\theta}(j(\mathbf x))
}
{
|\mathcal O_{j(\mathbf x)}|
}.
\]
Unlike equivariant quantum circuits, this construction imposes no
symmetry constraint on the circuit itself. Symmetry is enforced entirely
through the statistical representation.

\section{Results}
\label{sec:results}

We evaluate Q-Edge through three questions that correspond directly to
the theoretical developments. First, does the orbit-space
representation itself overcome the representation barrier? Second, does
Q-Edge outperform existing classical generative models for structured
extreme dependence? Finally, does the representation remain effective in
high-dimensional settings where direct quantum representations become
impractical? To answer these questions, we evaluate Q-Edge on five
synthetic $S_3$-symmetric dependence distributions in three dimensions
and on a challenging 30-dimensional
$S_{10}^3$-symmetric Dirichlet distribution. Throughout, Q-Edge is
compared against a direct-cell quantum circuit Born machine using the
same circuit architecture and training procedure and, where
appropriate, a parameter-matched conditional variational autoencoder
(VAE). Performance is assessed using three complementary metrics:
root mean squared error (RMSE), symmetric Kullback--Leibler (KL)
divergence, and total variation distance (TVD).  RMSE measures the pointwise accuracy of the reconstructed density, symmetric KL divergence quantifies distributional fidelity, and TVD measures the global probability mass discrepancy. Together, these complementary metrics show that the improvement is systematic rather than metric-specific. Implementation details, objective function, model architectures, optimization
procedures, hyperparameter settings, and hardware implementation are deferred to
Appendix~\ref{sec:methods}.

\subsection{Does Orbit-Space Representation Overcome the Representation Barrier?}
\label{sec:ablation}

The central claim of Q-Edge is that the primary limitation is not the quantum circuit itself, but the statistical representation on which it operates. To isolate this effect, we compare two quantum circuit Born machines that share exactly the same circuit architecture, optimization procedure and target distributions. The only difference is the representation: Direct - one circuit learns directly on the discretized simplex,
whereas Q-Edge learns on the symmetry-reduced orbit simplex.

We consider five radius-conditioned,
$S_3$-symmetric dependence distributions defined on the 2-simplex (Appendix \ref{sec:exe_setup}),
covering a broad spectrum of extremal dependence geometries, including
unimodal, multimodal, vertex-concentrated,
barycentre-concentrated and curved-support distributions. We focus on three-dimensional dependence structures because they provide the simplest non-trivial setting in which exchangeability, multimodality, and complex extremal dependence can all be systematically evaluated while remaining fully visualizable. Both models employ the same
$R_yR_zR_y$-CNOT ansatz and identical radius-conditioned training. At
resolution $N=6$, the augmented simplex contains $64$ cells, which the
$S_3$ symmetry reduces to $16$ orbit classes. Consequently, the
orbit-space model operates on only four qubits, whereas the direct-cell
representation requires six.

Q-Edge consistently improves reconstruction accuracy across all five target distributions, providing lower
reconstruction error than the direct-cell QCBM across all three
evaluation metrics (Table~\ref{tab:orbit-ablation}). Even for the Dirichlet distribution, where both models achieve highly accurate reconstructions, Q-Edge reduces RMSE by 31\%, symmetric KL divergence by 50\%, and TVD by 35\%. The improvement becomes even more pronounced for geometrically complex dependence structures. For the logistic-normal and ring distributions, symmetric KL divergence is reduced by 73\% and 75\%, respectively, accompanied by substantial reductions in both RMSE and TVD. Averaged across all five
targets, Q-Edge reduces RMSE by 30\%, symmetric KL divergence by 53\%, and TVD by 36\%. Q-Edge more faithfully
preserves multimodal and curved-support dependence structures while the
direct-cell model exhibits visible symmetry-breaking artefacts.

These improvements cannot be attributed to increased model capacity.
Indeed, the direct-cell QCBM employs a larger Hilbert space, two
additional qubits and a correspondingly larger parameter space. The only
difference between the two models is the statistical representation on
which optimization is performed. Q-Edge operates directly in the
exchangeable probability simplex, where every feasible point satisfies
the required permutation symmetry. In contrast, the direct-cell model
searches the full probability simplex, whose vast majority of candidate
distributions violate the symmetry of the target. Consequently, a
substantial fraction of the optimization effort is spent exploring
regions that are statistically irrelevant.

\begin{table}[t]
\centering
\caption{Orbit-space representation overcomes the representation barrier.
Comparison between Q-Edge and Direct-QCBM using identical quantum
circuits, optimization procedures and target distributions. Entries are
reported as \textbf{Q-Edge / Direct-QCBM}, followed by the relative error
reduction ($\downarrow$, lower is better). RMSE quantifies pointwise
reconstruction error, symmetric KL divergence measures distributional
fidelity, and TVD measures the overall probability mass discrepancy.
}
\label{tab:orbit-ablation}

\small
\begin{adjustbox}{max width=\textwidth}
\begin{tabular}{lccc}
\toprule
&\multicolumn{3}{c}{\textbf{Q-Edge / Direct-QCBM (relative improvement)}}\\
&
\textbf{RMSE ($\times10^{-3}$)}
&
\textbf{Sym. KL ($\times10^{-2}$)}
&
\textbf{TVD ($\times10^{-2}$)}
\\

\midrule

Dirichlet &
$\mathbf{0.09}/0.13$ ($\downarrow31\%$) &
$\mathbf{0.002}/0.004$ ($\downarrow50\%$) &
$\mathbf{0.215}/0.333$ ($\downarrow35\%$) \\

Logistic spectral &
$\mathbf{5.16}/6.30$ ($\downarrow18\%$) &
$\mathbf{1.54}/2.55$ ($\downarrow40\%$) &
$\mathbf{5.70}/7.58$ ($\downarrow25\%$) \\

Logistic-normal &
$\mathbf{0.39}/0.67$ ($\downarrow42\%$) &
$\mathbf{0.029}/0.109$ ($\downarrow73\%$) &
$\mathbf{0.584}/1.276$ ($\downarrow54\%$) \\

Three-mode &
$\mathbf{0.24}/0.31$ ($\downarrow23\%$) &
$\mathbf{0.015}/0.021$ ($\downarrow29\%$) &
$\mathbf{0.563}/0.702$ ($\downarrow20\%$) \\

Ring &
$\mathbf{4.26}/6.78$ ($\downarrow37\%$) &
$\mathbf{1.54}/6.12$ ($\downarrow75\%$) &
$\mathbf{4.44}/8.17$ ($\downarrow46\%$) \\
\bottomrule
\end{tabular}
\end{adjustbox}
\end{table}

These results provide direct empirical evidence for the
representation principle. The observed
performance gain does not arise from a more expressive quantum circuit,
but from performing optimization in the correct statistical space.
Orbit-space compression simultaneously reduces the computational
dimension, enforces exchangeability exactly, and improves reconstruction
accuracy, thereby overcoming the representation barrier predicted by
Theorem~\ref{thm:inapprox}.

\subsection{Can Orbit-Space Quantum Learning Outperform Classical Generative Models?}
\label{sec:benchmark}

Having established that orbit-space representation overcomes the
representation barrier independently of circuit architecture, we next compare Q-Edge with a classical deep generative model. The objective is not to compare quantum and classical hardware, but to determine whether the combination of symmetry-aware representation and quantum generative learning provides a statistical advantage over a parameter-matched
classical baseline.

We evaluate both models on the same five radius-conditioned,
$S_3$-symmetric dependence distributions defined on the 2-simplex. At resolution
$N=14$, the augmented partition contains $316$ cells, which collapse to
$64$ orbit states under the $S_3$ symmetry, allowing the entire problem
to be represented on six qubits. Q-Edge learns directly on the orbit
simplex, whereas the baseline is a conditional variational autoencoder
(VAE) with a Dirichlet decoder and a comparable parameter budget. Both
models are trained using the same adaptive partition and evaluated on
eleven held-out radii spanning
$r\in[0.01,10^{3}]$.

\begin{table}[t]
\centering
\caption{Q-Edge versus a parameter-matched conditional VAE on
five $S_3$-symmetric, radius-conditioned target distributions.
Entries are reported as \textbf{Q-Edge / VAE}, followed by the relative
error reduction achieved by Q-Edge
($\downarrow$, lower is better). RMSE measures pointwise reconstruction
accuracy, symmetric KL divergence measures distributional fidelity, and
TVD quantifies the overall probability mass discrepancy. Values are
averaged over eleven held-out radii.}
\label{tab:orbit-bench}

\small
\begin{adjustbox}{max width=\textwidth}
\begin{tabular}{lccc}
\toprule
&\multicolumn{3}{c}{\textbf{Q-Edge / VAE (relative improvement)}}\\
&
\textbf{RMSE ($\times10^{-3}$)}
&
\textbf{Sym. KL ($\times10^{-1}$)}
&
\textbf{TVD ($\times10^{-1}$)}
\\

\midrule

Dirichlet &
$0.05/\mathbf{0.01}$ ($\uparrow 400\%$)&
$0.001/\mathbf{0.001}$ ($= $) &
$0.053/\mathbf{0.017}$ ($\uparrow 212\%$)\\

Logistic spectral &
$\mathbf{1.40}/2.35$ ($\downarrow40\%$) &
$\mathbf{0.709}/1.190$ ($\downarrow40\%$) &
$\mathbf{0.821}/1.534$ ($\downarrow46\%$) \\

Logistic-normal &
$\mathbf{0.27}/0.79$ ($\downarrow66\%$) &
$\mathbf{0.032}/0.159$ ($\downarrow80\%$) &
$\mathbf{0.230}/0.614$ ($\downarrow63\%$) \\

Three-mode &
$\mathbf{0.14}/2.58$ ($\downarrow95\%$) &
$\mathbf{0.009}/5.976$ ($\downarrow99\%$) &
$\mathbf{0.126}/2.921$ ($\downarrow96\%$) \\

Ring &
$\mathbf{0.93}/4.17$ ($\downarrow78\%$) &
$\mathbf{1.309}/11.762$ ($\downarrow89\%$) &
$\mathbf{0.549}/2.701$ ($\downarrow80\%$) \\
\bottomrule
\end{tabular}
\end{adjustbox}
\end{table}

The Dirichlet distribution represents a special case because it belongs to the native decoder family of the VAE. Consequently, the classical model achieves better reconstruction. Nevertheless, both methods already attain near-perfect accuracy, with all three error metrics remaining close to zero, indicating that the practical difference between the two approaches is negligible. Outside this favourable setting, Q-Edge
consistently outperforms the classical baseline. The largest gains are
observed for the logistic-normal, ring and three-mode distributions,
whose complex geometric structures deviate substantially from the
Dirichlet family. For these challenging targets, Q-Edge reduces
distributional errors by up to $89\%$ and improves pointwise
reconstruction accuracy by up to $95\%$.

These improvements arise from fundamentally different inductive biases. The VAE inherits a strong preference for Dirichlet-like distributions through its decoder, explaining its excellent performance on the Dirichlet benchmark but limiting its ability to represent multimodal and curved-support dependence. In contrast, Q-Edge imposes no parametric assumption on the target distribution. Instead, it learns directly in the exchangeable orbit space, allowing the quantum circuit to represent the full class of symmetric probability distributions.

The reconstructed densities (Fig.~\ref{fig:density-panels}) reveal qualitative differences that are not fully captured by scalar error metrics. While both methods accurately recover the simple Dirichlet distribution, the VAE progressively smooths multimodal and curved-support dependence structures towards simpler unimodal distributions. In contrast, Q-Edge faithfully reconstructs multiple peaks, ring-shaped support and sharp geometric features, reflecting its ability to learn directly within the space of exchangeable probability distributions without imposing a restrictive parametric decoder.

Figure~\ref{fig:metrics-vs-r} explains why these qualitative differences arise. As the radius increases, the three-mode distribution develops increasingly separated modes, while the ring distribution concentrates probability mass along a curved manifold. These geometric changes have little effect on Q-Edge, whose reconstruction errors remain consistently low across the full range of radii. By contrast, the VAE errors increase substantially once the target departs from the Dirichlet family, indicating that its inductive bias becomes increasingly restrictive as the dependence geometry grows more complex.

\begin{figure}[htbp]
    \centering
    \begin{subfigure}[b]{0.195\textwidth}
        \centering
        \includegraphics[width=\linewidth]{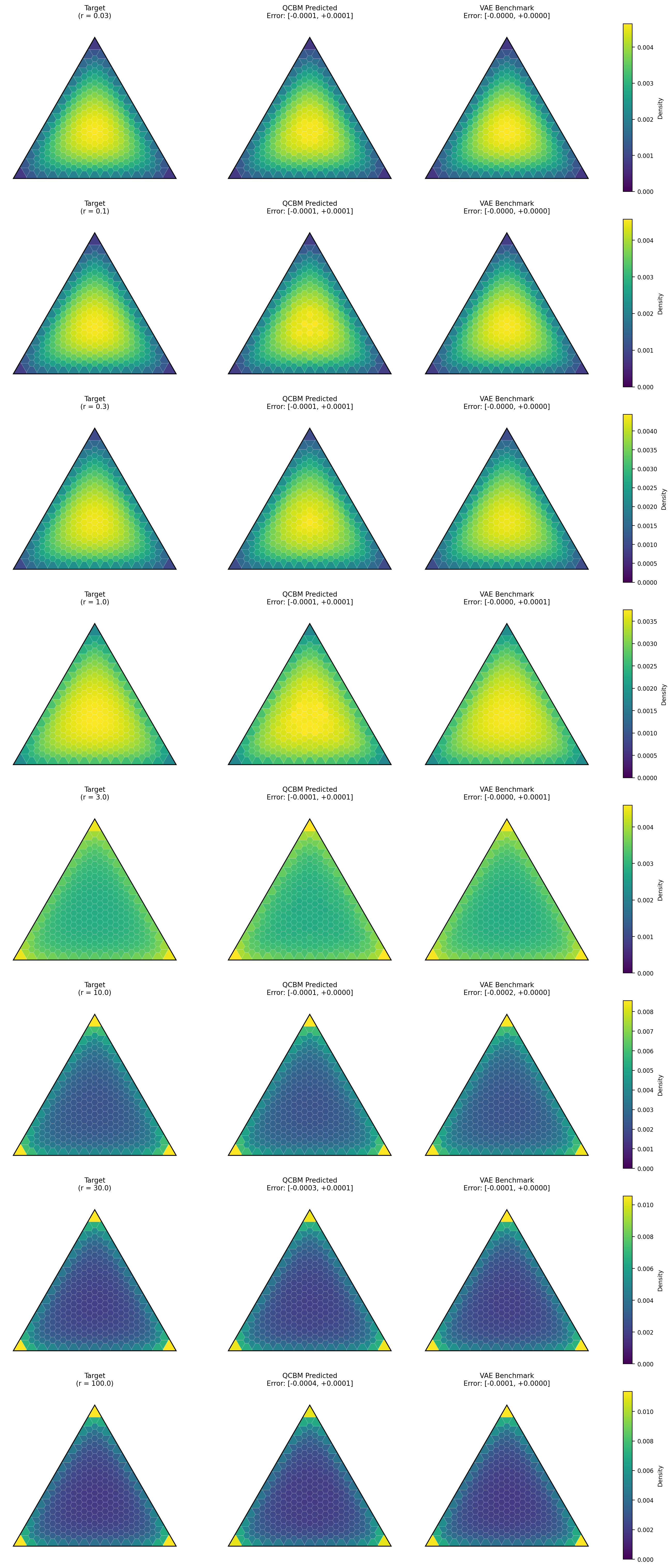}
        \caption{Dirichlet}
        \label{fig:3a}
    \end{subfigure}\hfill
    \begin{subfigure}[b]{0.195\textwidth}
        \centering
        \includegraphics[width=\linewidth]{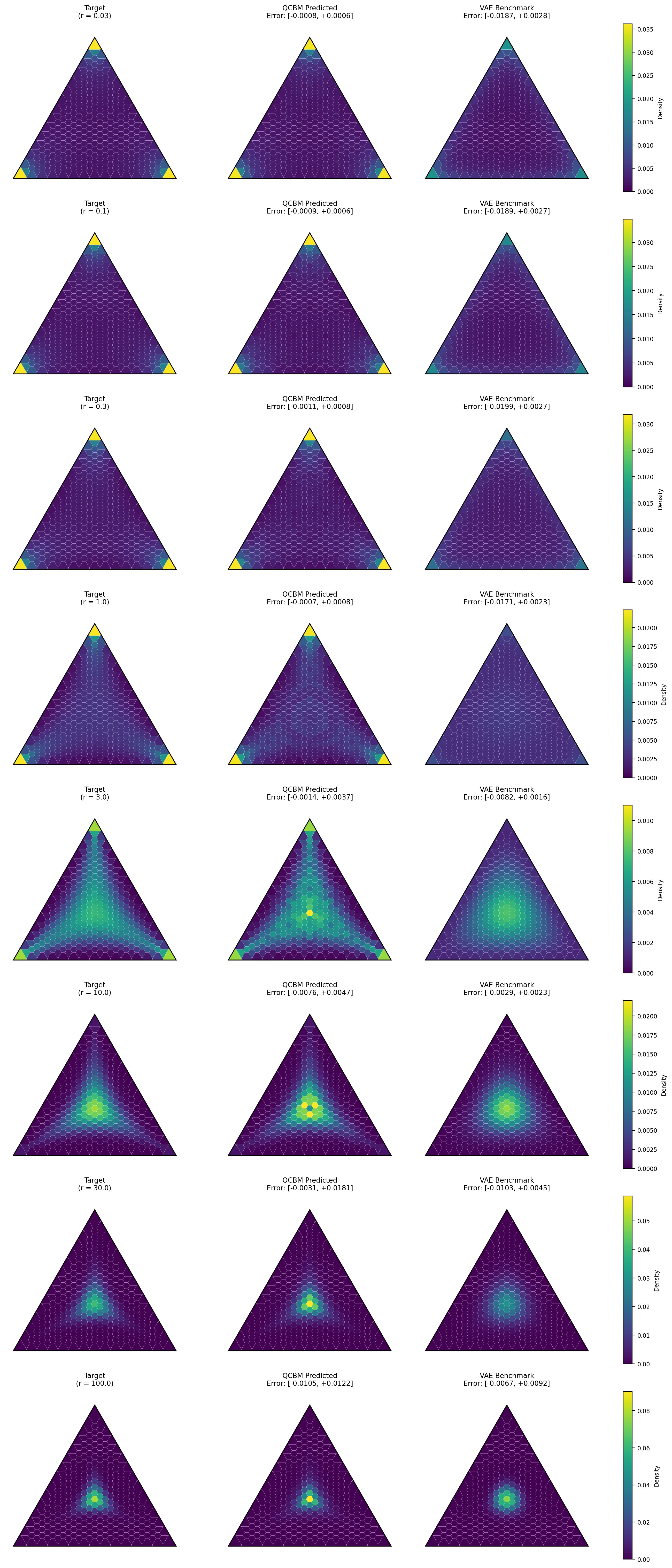}
        \caption{Logistic}
        \label{fig:3b}
    \end{subfigure}\hfill
    \begin{subfigure}[b]{0.195\textwidth}
        \centering
        \includegraphics[width=\linewidth]{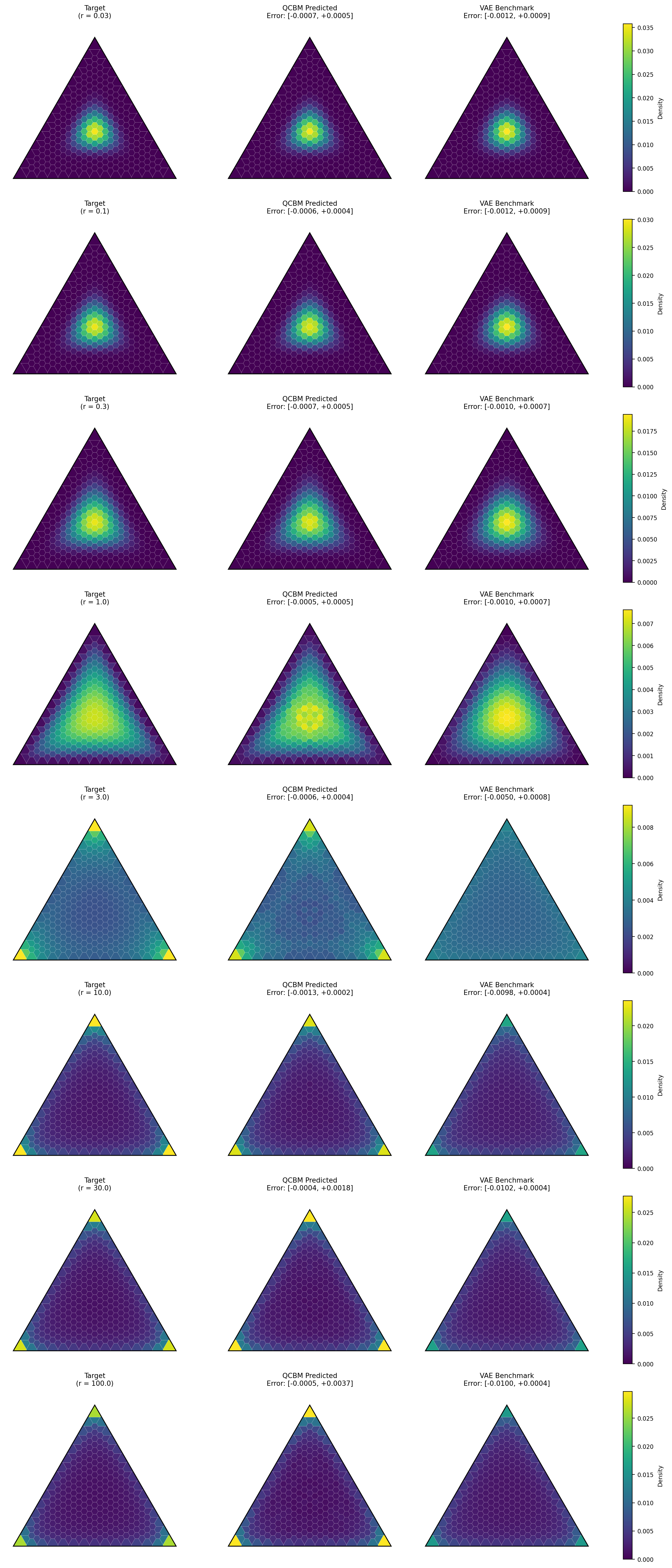}
        \caption{Logistic-normal}
        \label{fig:3c}
    \end{subfigure}\hfill
    \begin{subfigure}[b]{0.195\textwidth}
        \centering
        \includegraphics[width=\linewidth]{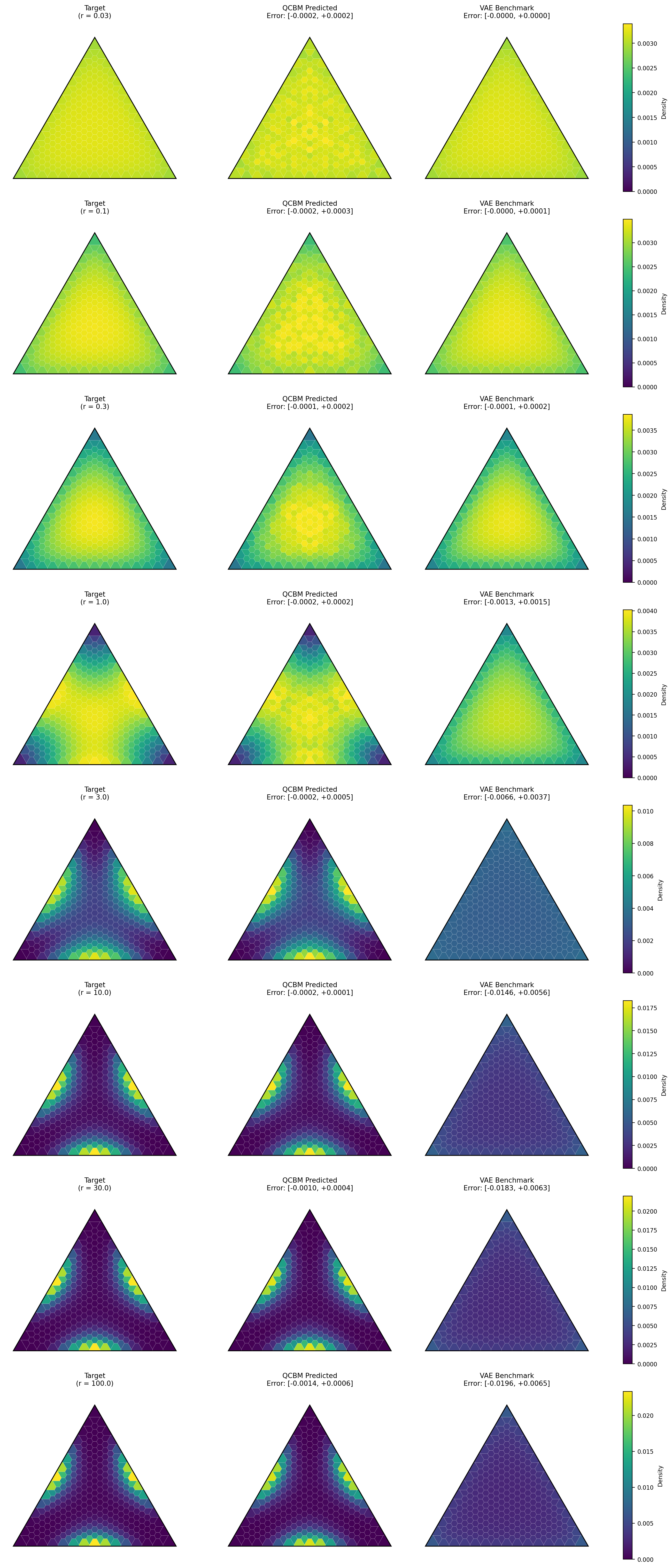}
        \caption{Three-mode}
        \label{fig:3d}
    \end{subfigure}\hfill
    \begin{subfigure}[b]{0.195\textwidth}
        \centering
        \includegraphics[width=\linewidth]{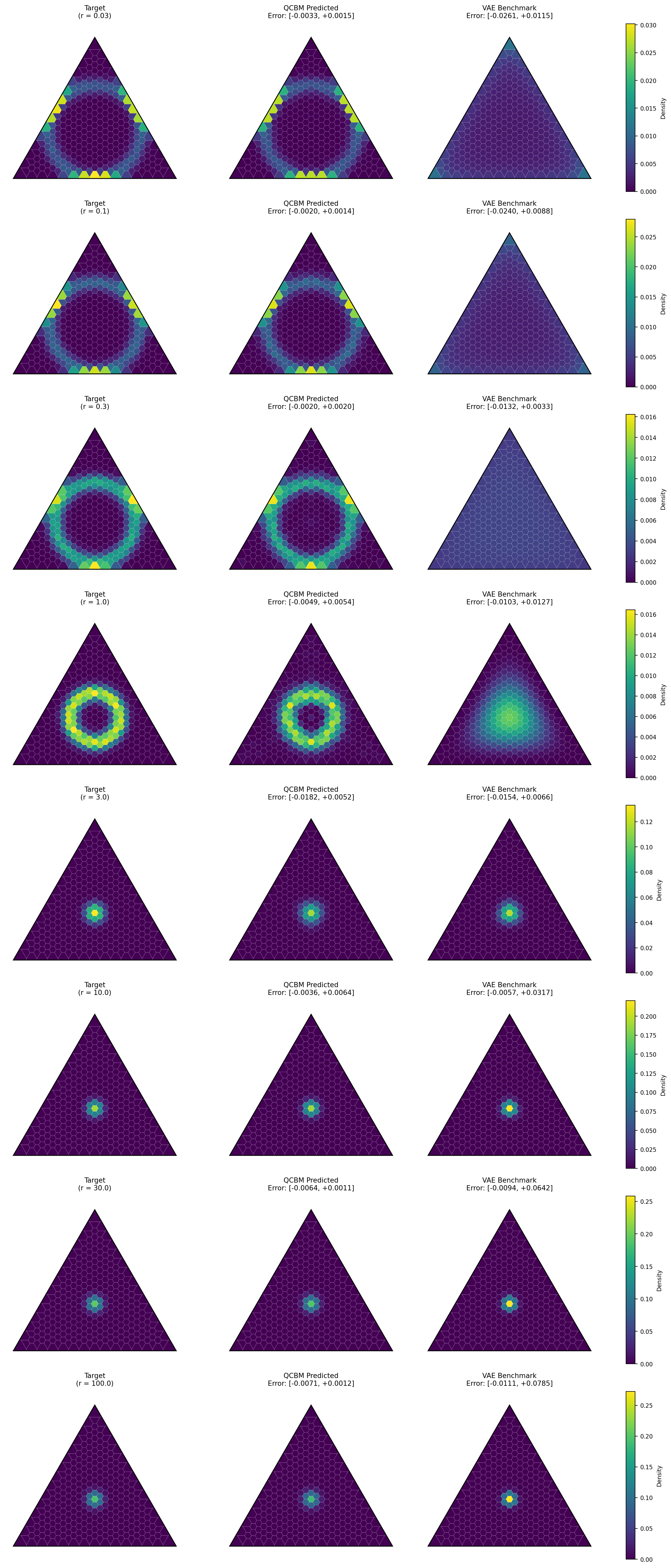}
        \caption{Ring}
        \label{fig:3e}
    \end{subfigure}
    \caption{Qualitative comparison of reconstructed dependence structures.
Reconstructed angular densities on the adaptive partition of the
2-simplex. While both methods accurately recover the Dirichlet
distribution, the VAE progressively smooths multimodal and curved-support
targets towards simpler distributions. Q-Edge faithfully reconstructs
multiple modes and ring-shaped dependence by learning directly in the
exchangeable orbit space.
}
    \label{fig:density-panels}
\end{figure}

\begin{figure}[htbp]
\centering
\includegraphics[width=\textwidth]{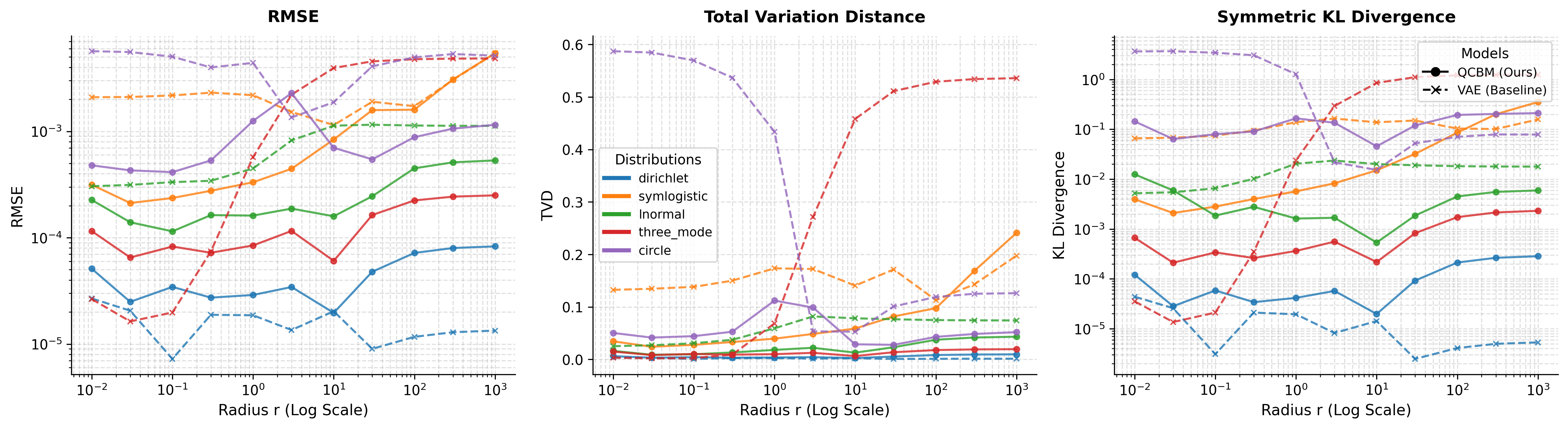}
\caption{Radius-resolved accuracy of the orbit-space QCBM (solid) and the conditional VAE (dashed) on the five $S_3$-symmetric targets ($D=3$, $N=14$). Q-Edge maintains consistently low errors across the full range of
dependence geometries, whereas the VAE deteriorates as the target
distribution departs from the Dirichlet family, particularly for
multimodal and curved-support distributions. Mean errors are summarized
in Table~\ref{tab:orbit-bench}}
\label{fig:metrics-vs-r}
\end{figure}

\subsection{Does Orbit-Space Compression Enable Scalable Quantum Learning?}
\label{sec:exp30d}

Having demonstrated the statistical advantages of orbit-space learning
in low-dimensional settings, we finally investigate whether these
advantages persist at a practically relevant scale. The theoretical
analysis predicts that orbit-space compression transforms an otherwise
combinatorial representation into a polynomially growing one. We test
this prediction on a 30-dimensional dependence distribution whose full
partition is far beyond the capacity of a direct quantum encoding.

The target is a radius-conditioned Dirichlet distribution on
$\Delta^{29}$ with three concentration parameters assigned to three
blocks of ten exchangeable variables, giving symmetry group
$G=S_{10}^3$. At partition resolution $N=6$, the augmented grid
contains $1\,628\,616$ cells, requiring approximately 21 qubits under a
direct encoding. Exploiting permutation symmetry reduces the
representation to
\[
\bO_G(\mathcal{Q}_6)=221+22+9+3+1=256=2^8
\]
orbit states, corresponding to a $6\,362\times$ compression of the
state space. Consequently, the entire learning problem can be solved
using only eight qubits, exactly as predicted by the orbit-space
construction.

The orbit-space QCBM is trained using symmetric KL divergence with
3\,456 trainable parameters (35 circuit layers and radial order
$T=3$), comfortably satisfying the necessary condition
$N_p\ge \bO_G-1=255$. As a classical baseline, we employ a
parameter-matched conditional VAE whose Dirichlet decoder is explicitly
constrained to satisfy the same $S_{10}^3$ symmetry. Both models
therefore learn within the identical exchangeable probability space,
allowing the comparison to isolate differences in learning capability
rather than differences in expressive capacity.

Unlike the three-dimensional examples, a 30-dimensional dependence
distribution cannot be visualized directly on the simplex. We therefore
compare the reconstructed per-point densities over the 256 orbit states.
Each histogram reports the probability assigned to every
symmetry-equivalent orbit, providing a direct visualization of how
accurately each model reconstructs the high-dimensional probability
geometry while preserving exchangeability.

The comparison is intentionally stringent because the target
distribution belongs to the native decoder family of the VAE
(Table~\ref{tab:scale30}). Consequently, the VAE achieves a lower
symmetric KL divergence, consistent with the three-dimensional
Dirichlet benchmark. This advantage is expected, as the decoder is
parameterized by the correct distribution family.

In contrast, the orbit-space QCBM makes no assumption about the
functional form of the target distribution. Nevertheless, it reduces the
per-point RMSE by approximately $60\%$ and the total variation distance
by $27\%$ relative to the VAE, indicating a more accurate reconstruction
of the local probability geometry while remaining competitive in overall
distributional fidelity.

Both models accurately recover the target distribution across all eleven held-out radii, while the orbit-space QCBM more faithfully captures fine-scale probability variations among neighbouring orbit states. These qualitative differences are consistent with the lower RMSE and TVD achieved by Q-Edge, whereas the VAE retains a slight advantage in symmetric KL because the target lies exactly within its assumed Dirichlet decoder family.

More importantly, these results demonstrate that orbit-space compression
preserves statistical accuracy while fundamentally changing the
computational scale of the learning problem. A direct encoding would
require representing more than $1.6$ million partition cells
(approximately $21$ qubits), whereas the symmetry-reduced
representation requires only $256$ orbit states (eight qubits), a
$6\,362\times$ reduction in state-space size. Orbit-space compression
therefore transforms an otherwise intractable high-dimensional quantum
learning problem into one that is compatible with near-term quantum
hardware without sacrificing reconstruction quality.

\begin{table}[t]
\centering
\caption{\textbf{Scaling to a 30-dimensional dependence distribution.}
Comparison between the orbit-space QCBM (8 qubits) and a parameter-matched
$S_{10}^3$-symmetric VAE. The original partition contains
$1\,628\,616$ cells (approximately 21 qubits), which are compressed into
256 orbit states, corresponding to a $6\,362\times$ reduction. RMSE
measures pointwise reconstruction accuracy, symmetric KL divergence
measures distributional fidelity, and TVD measures the overall
probability mass discrepancy. Entries are reported as
\textbf{Q-Edge / VAE}, followed by the relative difference
($\downarrow$, lower is better).}
\label{tab:scale30}

\small
\begin{tabular}{lccc}
\toprule
&
\textbf{RMSE $\times 10^{-5}$}
&
\textbf{Sym. KL $\times 10^{-1}$}
&
\textbf{TVD $\times 10^{-1}$}
\\
\midrule

Dirichlet &
$\mathbf{0.101}/0.253$ ($\downarrow 60\%$) &
$0.386/\mathbf{0.272}$ ($\uparrow 42\%$) &
$\mathbf{0.567}/0.775$ ($\downarrow 27\%$)
\\

Logistic normal &
$\mathbf{0.0289}/7.56$ ($\downarrow 99\%$) &
$\mathbf{0.00433}/3.44$ ($\downarrow 99\%$) &
$\mathbf{0.0744}/3.10$ ($\downarrow 98\%$)
\\

Three-mode mixture &
$\mathbf{0.0243}/0.574$ ($\downarrow 96\%$) &
$\mathbf{0.00914}/9.23$ ($\downarrow 99\%$) &
$\mathbf{0.133}/3.35$ ($\downarrow 96\%$)
\\

\bottomrule
\end{tabular}
\end{table}

\section{Discussion}\label{sec:discussion}

The central contribution of this work is conceptual rather than
algorithmic. We show that the principal obstacle to quantum generative
modelling of multivariate extremes is not limited circuit expressivity,
but an inappropriate statistical representation. Classical and quantum
models alike are typically trained in the full probability simplex,
whereas the target distributions occupy a much smaller
symmetry-constrained manifold. By learning directly in the orbit space,
Q-Edge aligns the model representation with the statistical structure of
exchangeable extremes, transforming an exponentially large learning
problem into an equivalent low-dimensional one.

This perspective establishes a different route towards scalable quantum
machine learning. Rather than overcoming expressivity limitations by
increasing circuit size or depth, orbit-space compression exploits
representation-theoretic structure before optimization begins. The
resulting reduction preserves the target distribution exactly while
substantially decreasing the effective dimension of the learning
problem. Our theoretical analysis identifies the conditions under which
such compression is lossless, and the numerical experiments demonstrate
that it leads to improved statistical accuracy despite requiring fewer
qubits and fewer effective degrees of freedom.

The present work also clarifies the role of quantum computing in
learning extreme dependence. The quantum advantage does not arise from
representing arbitrary high-dimensional distributions, but from
efficiently modelling structured probability spaces whose symmetry makes
them intrinsically compressible. Although demonstrated here for
exchangeable multivariate extremes, the underlying principle extends
naturally to other generative modelling problems possessing group
symmetries.

Several limitations remain. The experiments are performed using exact
statevector simulation, and the robustness of the approach under
hardware noise and finite sampling remains to be established. In
addition, the largest experiments remain within the classically
simulable regime, demonstrating the correctness and scalability of the
orbit-space construction rather than a computational quantum advantage.
Finally, the current framework assumes that the underlying symmetry is
known exactly; extending orbit-space learning to approximately symmetric
or data-driven group structures is an important direction for future
research.

More broadly, Q-Edge suggests that exploiting statistical symmetry may
be as important as improving quantum hardware for scalable quantum
machine learning. Orbit-space compression provides a general mechanism
for matching quantum representations to structured probability spaces,
opening opportunities well beyond extreme-value modelling, including
scientific machine learning, probabilistic digital twins, and
high-dimensional stochastic simulation.

\bibliographystyle{apalike}
\bibliography{ref}

\clearpage
\appendix
\section*{Appendix}

\section{Implementation}\label{sec:methods}
\subsection{Experiment Setup}\label{sec:exe_setup}
\paragraph{Training and the orbit-space encoding.}
The ansatz and its parameter count $N_p=3n(L+1)$ are specified under the rational radial parameterisation of
Eq.~\eqref{eq:unitary_layers}. The orbit-space route places no symmetry
constraint on the entangler, and the simulations use the circular-CNOT layer;
for the hardware runs we use the equivalent cyclic-CZ layer, whose gates
pairwise commute and therefore compile to half the depth on devices whose
native two-qubit gate is CZ. Appendix~\ref{sec:strategyA} recalls a full-CZ reference construction for
qubit-permutation symmetry; equivariance on the encoded grid is
defined using the induced representation of Definition~\ref{def:induced_rep}. The circuit is trained by minimizing the symmetric KL divergence (Jeffreys divergence) between the Born
distribution and the target cell masses,
\begin{equation}\label{eq:jeffreys}
D_{\mathrm{KL}}^{\mathrm{sym}}(p,\pi) \;=\; \tfrac12 \sum_{k}\Bigl[
p(k)\log\frac{p(k)}{\pi(k)} + \pi(k)\log\frac{\pi(k)}{p(k)}\Bigr].
\end{equation}
Training minimises the summation of $D_{\mathrm{KL}}^{\mathrm{sym}}$ over the representative radii
$r_1,\dots,r_{K_r}$ of the adaptive bins of Eq.~\eqref{eq:bin_formula}
($K_r = 41$),
\begin{equation}\label{eq:training-loss}
\mathcal{L}(\boldsymbol{\theta})
\;=\; \sum_{j=1}^{K_r}
D_{\mathrm{KL}}^{\mathrm{sym}}\bigl(p_{\boldsymbol{\theta}(r_j)},\,\pi_{r_j}\bigr),
\end{equation}
both computed in closed form from the statevector. The choice of divergence is principled: for the bounded Gaussian-mixture kernels standard in Born-machine training~\citep{gretton2012kernel,liu2018differentiable}, with
$\sup_x k(x,x)\leq 1$, one has
$\mathrm{MMD}(p_{\boldsymbol{\theta}},\pi)\leq
2\,d_{\mathrm{TV}}(p_{\boldsymbol{\theta}},\pi)\leq
\sqrt{2\,\mathcal{L}(\boldsymbol{\theta})}$ by the reproducing property and
Pinsker's inequality, so driving the symmetric KL to zero forces the MMD (and
the total variation) to zero. The converse fails: a bounded kernel is
insensitive to likelihood ratios, so the MMD can be arbitrarily small while the
KL divergence remains large---or infinite, when the model assigns vanishing
mass where the target does not. Exact statevector simulation makes the KL
computable without sampling, removing the estimation advantage that usually
motivates the MMD. For the orbit-space model, $K=\bO_G(\mathcal{Q}_N)$ is
chosen to be a power of two, $n'=\log_2 K$ qubits are used, and $\psi$ is fixed
by barycentre anchoring, $(L_1,L_2,m,\mathrm{partition})$ sorting and Gray
coding (Section~\ref{sec:orbit}). The Dirichlet specialisation of Burnside's
lemma reads
$\bO_G(\mathcal{X}_N)=\sum_{j_1+\cdots+j_b=N}\prod_{k=1}^{b}p_{n_k}(j_k)$,
where $p_n(j)$ is the number of partitions of $j$ into at most $n$ parts
(Appendix, Theorem~\ref{thm:dirichlet_orbits}).

\paragraph{Target distributions (2-simplex).}
All five radius-conditioned targets are $S_3$-invariant on $\Delta^2$. The
\emph{Dirichlet} target is symmetric $\mathrm{Dir}(\alpha(r)\mathbf{1})$ with
$\alpha(r)=3/2-(1/2)\tanh(\log r-2)$, which decays through the uniform law
($\alpha=1$ at $r=2$) toward vertex-favouring concentration. The
\emph{logistic spectral measure} has density
$h(\boldsymbol\theta)\propto(\prod_i\theta_i^{-1/\alpha})
(\sum_i\theta_i^{-1/\alpha})^{\alpha-D}$ with $\alpha(r)=0.1+0.8/(1+r)$, moving
from vertex-concentrated near-independence at $r\to 0$ to
barycentre-concentrated near-complete dependence as $r\to\infty$. The
\emph{logistic-normal} is an Aitchison density with exchangeable log-ratio
covariance $\sigma(r)^2\bigl[\begin{smallmatrix}2&1\\1&2\end{smallmatrix}\bigr]$
and $\sigma(r)=2.5-2.2\,e^{-0.2r}$. The \emph{three-mode mixture} is
$\tfrac13[\mathrm{Dir}(M/2,M/2,1)+\mathrm{Dir}(1,M/2,M/2)
+\mathrm{Dir}(M/2,1,M/2)]$ with $M(r)=10\tanh(0.3r)+2$, ranging from
near-uniform at $r\to 0$ to three sharp peaks at the edge midpoints as
$r\to\infty$. The \emph{circle} is an annulus of width $0.02$ centred at
$(0.5,\sqrt3/6)$ with radius
$R_{\mathrm{circ}}(r)=(\sqrt3/6)\cdot 2/(1+e^{r})$, projected to barycentric
coordinates.

\paragraph{Implementation and metrics.}
The orbit-space QCBM ($n'=6$ qubits, $L=12$, $T=2$) is implemented in
PennyLane (\texttt{default.qubit}, backpropagation). Training uses Adam with a
cosine-annealed learning rate and gradient clipping, minimising the exact
symmetric-KL objective; targets are cached on the adaptive radial boundaries of
Eq.~\eqref{eq:bin_formula} ($K_r=41$, $a_{K_r}=200$). For the $D=30$ experiment
($n'=8$, $L=35$, $T=2$) the same protocol is used.
The classical baseline on the 2-simplex is a conditional VAE with a Dirichlet
decoder (latent dimension $2$, parameter count matched to the quantum model),
trained by standard ELBO on samples drawn at the same radial boundaries and
evaluated by Monte-Carlo marginalisation of its decoder at the grid points,
normalised to the point-density semantics used for the QCBM and the analytic
target. In thirty dimensions the baseline is itself constrained to be $S_{10}^3$ invariant, at $6{,}596$
parameters against $3{,}456$ for the quantum model; the baseline is therefore
favoured on both counts, searching the $\Delta_{\mathrm{sym}}$ and doing so with more parameters. Its latent variable is active on every target (final $D_{\mathrm{KL}}^{\mathrm{sym}}$ between $0.96$ and $2.21$ nats over a three-dimensional latent), so the comparison is not confounded by
posterior collapse.
Reported metrics---RMSE, symmetric KL divergence, and total-variation distance.

\subsection{Execution on quantum hardware}
\label{sec:hardware}

The orbit-space construction compresses the $D=3$, $N=6$ angular grid from
$|\mathcal{Q}_6| = 64$ points to $16$ exchangeability orbits under the full
symmetric group $S_3$, so the Born machine requires $n = \log_2 16 = 4$ qubits. To ensure the model possesses adequate capacity to learn arbitrary distributions, we conservatively set the quantum circuit depth to $L=4$.
We verify that this compression makes the model executable on present-day
hardware by running the trained Dirichlet model on two IBM processors of
different coupling topology: \texttt{ibm\_miami}, a $120$-qubit Nighthawk
device with a square lattice, and \texttt{ibm\_pittsburgh}, a $156$-qubit Heron
device with a heavy-hexagonal lattice. Training is performed classically; the
devices are used only for inference, i.e.\ to sample
$p_{\boldsymbol{\theta}(r)}$ at eleven held-out values of $r$. Both runs
execute \emph{the same logical circuit at the same angles}: the angle matrix
$\boldsymbol{\theta}(r)$ obtained from the classical optimization is stored and
replayed on the second device, so the two executions differ only in how that
circuit is compiled onto the hardware. The noiseless reference is recomputed
from those same angles and is therefore identical for both, which is what makes
the comparison controlled.

The two topologies compile this circuit differently, and that difference is the
object of the comparison. The Nighthawk lattice is a square grid and so
contains four-cycles; the cyclic entangler
$\mathrm{CZ}(q, q\!+\!1 \bmod 4)$ therefore maps onto physical couplings with
no routing overhead, and the compiled circuit contains exactly $4L = 16$
two-qubit gates at depth $56$. The heavy-hexagonal lattice contains no
four-cycles, so the same ring requires inserted \textsc{swap} operations and
compiles to $27$ two-qubit gates at depth $71$---a $1.7\times$ overhead in gate
count. This is a property of the coupling graph rather than of the
construction, and we report it so that the gate counts can be attributed
correctly, not as an advantage of the method.

\begin{table}[t]
\centering
\caption{\textbf{Compiled circuit and execution settings.}
Physical qubits were selected by enumerating the four-cycles of the coupling
graph, where they exist, and scoring candidates by two-qubit and readout
fidelities, reported static $ZZ$ crosstalk, and decoherence over the circuit
duration; on the heavy-hexagonal device, which admits no four-cycle, the layout
was chosen by seeded \textsc{sabre} routing under the same score.}
\label{tab:hw-config}
\small
\begin{tabular}{@{}lcc@{}}
\toprule
& \texttt{ibm\_miami} & \texttt{ibm\_pittsburgh}\\
& (Nighthawk, square) & (Heron, heavy-hex)\\
\midrule
Qubits / orbits / grid points & \multicolumn{2}{c}{$4$ / $16$ / $64$ (compression $4.00\times$)}\\
Layers $L$ / polynomial order $T$ & \multicolumn{2}{c}{$4$ / $2$}\\
Effective parameters $(L{+}1)n$ & \multicolumn{2}{c}{$20 \geq \bO_{S_3}(\mathcal{Q}_6) - 1 = 15$}\\
Angle matrix & \multicolumn{2}{c}{identical (replayed from the first run)}\\
\midrule
Two-qubit gates (logical / compiled) & $16$ / $16$ & $16$ / $27$\\
Inserted \textsc{swap} & none & yes\\
ISA depth & $56$ & $71$\\
\midrule
Shots per $r$ & \multicolumn{2}{c}{$144{,}153$}\\
Mirror / readout calibration & \multicolumn{2}{c}{$8{,}192$ / $16{,}384$ shots $\times\,16$ circuits}\\
Dynamical decoupling / twirling & \multicolumn{2}{c}{\texttt{XpXm} / gates and measurement}\\
\bottomrule
\end{tabular}
\end{table}

\paragraph{Result.}
Comparing a finitely sampled hardware distribution against an analytically
evaluated simulator distribution is not like-for-like: at
$N_{\mathrm{shots}} = 144{,}153$ and $K = 16$ outcomes, multinomial sampling
alone contributes $\sqrt{K/2\pi N_{\mathrm{shots}}} = 4.2\times10^{-3}$ in
total variation, comparable to the model error itself. We therefore quote all
ratios against a \emph{shot-matched} baseline obtained by drawing
$N_{\mathrm{shots}}$ samples from the noiseless model distribution.

Table~\ref{tab:hw-main} gives the outcome. Two standard corrections are applied on each device, in sequence and in the reverse order of the channels they undo: readout first, then depolarizing.
Neither uses the target, and both are calibrated from separate circuits. The readout step is confusion-matrix unfolding~\citep{maciejewski2020mitigation,
nation2021scalable}; with only four qubits we measure the full $16\times16$ matrix from sixteen basis-state preparations rather than the tensor-product approximation used at larger registers, and solve for the corrected distribution by constrained least squares, which returns a valid probability vector where direct inversion would not. The depolarizing step follows the noise-estimation-circuit method of~\citet{urbanek2021mitigating}, with the mirror circuit $U(\boldsymbol{\theta})U^{\dagger}(\boldsymbol{\theta})$ as the estimator; because the mirror carries twice the entangling gates, the factor for a single execution is taken as the square root of the measured one, and it is estimated separately at each radius rather than once for the run. Gate and measurement twirling are enabled throughout, which is what makes a single-parameter depolarizing model a reasonable ansatz. After both, the learned angular measure is reproduced to within $3.9\times$ the shot-matched baseline in total variation on the square lattice and $4.7\times$ on the heavy-hexagonal one, with absolute total variation distances of $1.8\times10^{-2}$ and $2.2\times10^{-2}$ over the $64$-point grid.

\begin{table}[t]
\centering
\caption{\textbf{Distributional fidelity on hardware.}
Means over the eleven held-out values of $r$, computed on the reconstructed
$64$-point grid distribution against the target angular density. Both devices
execute the same angle matrix, so the noiseless rows are common to them.}
\label{tab:hw-main}
\small
\begin{tabular}{@{}lcccc@{}}
\toprule
& RMSE & $D_{\mathrm{KL}}^{\mathrm{sym}}$ & TVD & $W_1$ \\
\midrule
Noiseless simulator, exact        & $1.12\times10^{-4}$ & $3.8\times10^{-5}$ & $0.00226$ & $4.60\times10^{-4}$\\
Noiseless simulator, shot-matched & $2.05\times10^{-4}$ & $9.0\times10^{-5}$ & $0.00477$ & $8.7\times10^{-4}$\\
\midrule
\multicolumn{5}{@{}l}{\emph{\texttt{ibm\_miami} --- square lattice, $16$ two-qubit gates}}\\
\quad raw                         & $1.15\times10^{-3}$ & $2.82\times10^{-3}$ & $0.02737$ & $5.42\times10^{-3}$\\
\quad $+$ readout                 & $7.86\times10^{-4}$ & $1.31\times10^{-3}$ & $0.01974$ & $3.71\times10^{-3}$\\
\qquad $+$ depolarizing            & $\mathbf{6.71\times10^{-4}}$ & $\mathbf{9.4\times10^{-4}}$ & $\mathbf{0.01837}$ & $\mathbf{2.98\times10^{-3}}$\\
\quad \textit{ratio to shot-matched} & $3.27\times$ & $10.4\times$ & $3.85\times$ & $3.43\times$\\
\midrule
\multicolumn{5}{@{}l}{\emph{\texttt{ibm\_pittsburgh} --- heavy-hex lattice, $27$ two-qubit gates}}\\
\quad raw                         & $1.31\times10^{-3}$ & $3.35\times10^{-3}$ & $0.03106$ & $6.14\times10^{-3}$\\
\quad $+$ readout                 & $1.09\times10^{-3}$ & $2.19\times10^{-3}$ & $0.02704$ & $5.02\times10^{-3}$\\
\qquad $+$ depolarizing            & $\mathbf{8.42\times10^{-4}}$ & $\mathbf{1.32\times10^{-3}}$ & $\mathbf{0.02242}$ & $\mathbf{3.86\times10^{-3}}$\\
\quad \textit{ratio to shot-matched} & $4.11\times$ & $14.7\times$ & $4.70\times$ & $4.44\times$\\
\bottomrule
\end{tabular}
\end{table}

\paragraph{Error budget and the effect of topology.}
Because the target is $S_3$-invariant it is constant within each orbit, and the
decoder $\Phi$ spreads orbit mass uniformly; grid-space and orbit-space total
variation distances therefore coincide, allowing the three error sources to be
separated. Referencing hardware output to the noiseless simulator rather than
to the target isolates the device contribution, giving post-mitigation device
residuals of $1.81\times10^{-2}$ and $2.27\times10^{-2}$ against a common model
approximation error of $2.26\times10^{-3}$ and a common sampling floor of
$4.2\times10^{-3}$. Combining these in quadrature---an approximation, since
total variation is not an $L^2$ norm, but one that reproduces the measured
residual to within $1\%$---the device accounts for $93\%$ and $96\%$ of the
total error budget, sampling for $5\%$ and $3\%$, and the model for $1\%$ on
both: the fidelity reported above is limited by the processors, not by the
ansatz or the shot budget.

The topological difference enters as predicted, but is partly offset. The
mirror-circuit depolarizing factor falls from $\lambda = 0.911$ on the square
lattice to $0.832$ on the heavy-hexagonal one, consistent with the $1.7\times$
larger gate count; solving $\lambda \approx (1-\epsilon)^{n_{2q}}$ for
$\epsilon$ gives per-gate depolarizing rates of $5.8\times10^{-3}$ and
$6.8\times10^{-3}$, close enough that the difference in $\lambda$ is
attributable to gate count rather than gate quality. The final fidelities,
however, differ by only $1.22\times$ in total variation, because the
heavy-hexagonal device has markedly better readout: readout mitigation removes
$29\%$ of the raw residual there against $14\%$ on the square lattice,
offsetting much of the routing overhead.

A second difference is structural rather than aggregate. On the square lattice
the post-mitigation residual varies by a factor of $2.3$ across $r$ and is
strongly correlated with the angle-schedule variable $t = 1/(1+|r|)$
(Pearson $-0.94$), with the implied error strength
$1 - \lambda_{\mathrm{res}}$ spanning $198\%$; on the heavy-hexagonal device
the same residual varies by only $1.1\times$ (correlation $-0.48$, span
$30\%$). Since the gate count is identical for every $r$ and only the rotation
angles differ, a stochastic channel of fixed strength cannot produce such a
dependence: an angle-dependent, coherent component remains identifiable on the
shallower circuit, whereas on the deeper one the larger stochastic error masks
it. Removing that component would require learning the device's
Pauli--Lindblad generator, which we do not attempt here; the purpose of this
section is to establish executability rather than to optimise hardware
performance.

We note two limitations. The two executions are separated in time and therefore
differ in calibration state as well as in topology, so drift cannot be excluded
as a contributor to the differences reported above; and each configuration was
run once, so the residual structure---in particular the correlation of
$-0.94$---would benefit from replication.

\begin{figure}[htbp]
\centering
\includegraphics[width=\linewidth]{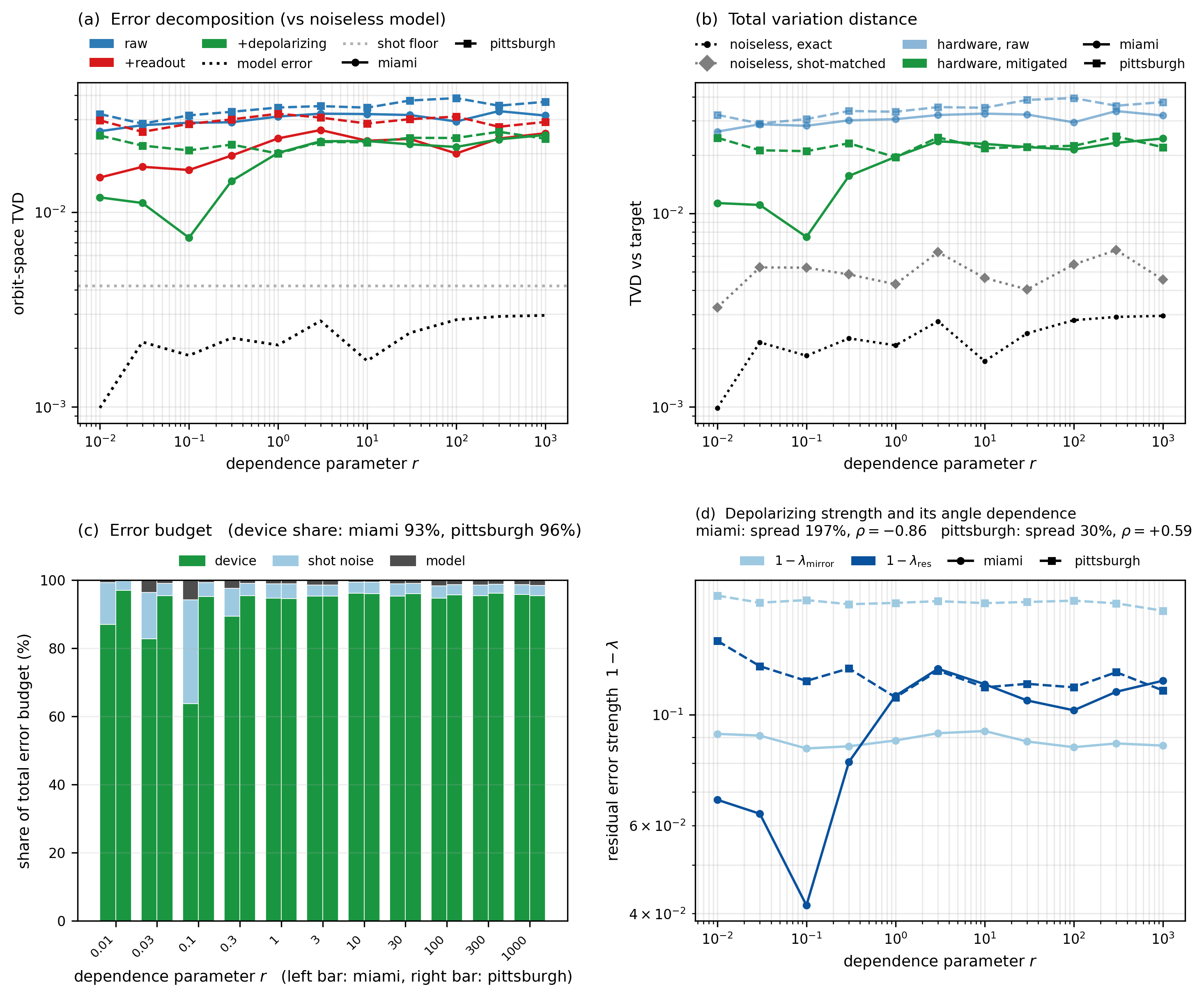}
\caption{\textbf{Execution of the orbit-space Born machine on two coupling
topologies.} Both devices execute the same angle matrix, so the noiseless
references are common to them. Throughout, line style identifies the device
(solid with circles, \texttt{ibm\_miami}; dashed with squares,
\texttt{ibm\_pittsburgh}), while colour identifies the quantity, as given in
each panel's legend.
\textbf{a}, Orbit-space error decomposition against the dependence parameter
$r$, referenced to the noiseless simulator, at the three mitigation stages. The
dotted lines are the model approximation error and the multinomial sampling
floor $\sqrt{K/2\pi N_{\mathrm{shots}}}$.
\textbf{b}, Total variation distance against the target, before and after
mitigation, together with the exact and shot-matched noiseless baselines.
\textbf{c}, Decomposition of the post-mitigation error budget into model,
sampling and device contributions at each $r$ (left bar, square lattice; right
bar, heavy-hex); the device dominates throughout on both.
\textbf{d}, Residual error strength $1-\lambda$, as estimated from mirror
circuits and as implied by the post-mitigation residual. The much larger spread
of the latter on the square lattice is the signature of an angle-dependent
component that survives the isotropic correction.}
\label{fig:hw}
\end{figure}

\clearpage

\section{The Parameterized Quantum Circuit}

\subsection{Setup and Notation}

Consider an $n$-qubit parameterized quantum circuit (PQC) with $L$
entangling layers.

\begin{definition}[Circuit architecture]\label{def:circuit}
The circuit $U(\bm{\theta})$ acts on $n$ qubits with $L$ entangling
layers, defined by:
\begin{enumerate}
    \item \emph{Initial layer:} Hadamard gates on all qubits, $H^{\otimes n}$.
    \item \emph{Parameterized blocks} ($l = 0, 1, \ldots, L$): For each qubit
    $q \in \{0,\ldots,n{-}1\}$, apply the $YZY$ decomposition:
    \[
        W_q^{(l)}(\bm{\theta})
        = R_y\!\bigl(\theta_{3q+2}^{(l)}\bigr)\,
          R_z\!\bigl(\theta_{3q+1}^{(l)}\bigr)\,
          R_y\!\bigl(\theta_{3q}^{(l)}\bigr)
    \]
    where $R_y(\alpha) = e^{-i\alpha Y/2}$ and
    $R_z(\alpha) = e^{-i\alpha Z/2}$.
    \item \emph{Entangling layers} ($l = 0,1,\ldots,L{-}1$): After each
    parameterized block (except the last), apply a circular CNOT ladder $E$:
    \[
        E = \mathrm{CNOT}_{n-1\to 0}\prod_{q=0}^{n-2}\mathrm{CNOT}_{q\to q+1}.
    \]
\end{enumerate}
The total number of variational parameters is $N_p = 3n(L+1)$.
\end{definition}

In the \emph{layer-level} representation, the full unitary reads:
\begin{equation}\label{eq:unitary_layers}
    U(\bm{\theta})
    = \biggl[\bigotimes_{q} W_q^{(L)}\biggr]
      \prod_{l=0}^{L-1}\biggl(E\cdot\bigotimes_{q} W_q^{(l)}\biggr)
      \cdot H^{\otimes n}.
\end{equation}

Since gates acting on different qubits commute, each tensor product
$\bigotimes_q W_q^{(l)}$ can be expanded into a sequential product in
any order. Assigning a global index $j=1,2,\ldots,N_p$ to the
parameterized rotation gates---ordered so that $j=1$ is applied first
(immediately after $H^{\otimes n}$) and $j=N_p$ is applied last (closest
to the measurement end)---and absorbing all fixed gates (CNOT layers and
identity operators between commuting single-qubit blocks) into fixed
(parameter-independent) unitaries~$C_j$, we obtain the \emph{gate-level} representation:
\begin{equation}\label{eq:unitary_gates}
    U(\bm{\theta})
    = \prod_{j=N_p}^{1}\bigl[C_j\cdot e^{-i\theta_j G_j/2}\bigr]
      \cdot H^{\otimes n}
\end{equation}
where the product is ordered so that $j=1$ acts first and $j=N_p$ acts
last. Each $G_j\in\{Y^{(q)},Z^{(q)}\}$ is a single-qubit Pauli generator
(tensored with identity on the remaining qubits), and each $C_j$ is a
fixed (parameter-independent) unitary: $C_j = E$ when gate~$j$
is the first parameterized gate of a new layer, and $C_j = I$ otherwise.

For notational convenience, we write~\eqref{eq:unitary_gates} as:
\begin{equation}\label{eq:shorthand}
    U(\bm{\theta})
    = G_{N_p}(\theta_{N_p})\cdots G_2(\theta_2)\cdot G_1(\theta_1)
      \cdot H^{\otimes n},
    \qquad
    G_j(\theta_j)\coloneqq C_j\cdot e^{-i\theta_j G_j/2}.
\end{equation}

\subsection{Polynomial Radius Parameterization}

\begin{definition}[Polynomial asymptotic parameterization]
\label{def:param}
Given $T+1$ learned parameter vectors
$\bm{\theta}^{(0)}, \bm{\theta}^{(1)}, \ldots, \bm{\theta}^{(T)}
\in \Real^{N_p}$,
the radius-dependent parameters are:
\begin{equation}\label{eq:theta_r}
    \bm{\theta}(r)
    = \bm{\theta}^{(0)}
      + \sum_{i=1}^{T} \bm{\theta}^{(i)} \cdot g_i(r),
    \qquad
    g_i(r) = \frac{1}{(1+r)^i},
    \qquad r \geq 0.
\end{equation}
Equivalently, with the substitution $t = 1/(1+r) \in (0,1]$:
\begin{equation}\label{eq:theta_t}
    \bm{\theta}(t) = \bm{\theta}^{(0)}
    + \sum_{i=1}^{T} \bm{\theta}^{(i)} \cdot t^i.
\end{equation}
The total number of trainable real parameters is $(T+1) \cdot N_p$.
\end{definition}

\begin{remark}[Asymptotic behavior]
As $r \to \infty$ (i.e., $t \to 0$):
$\bm{\theta}(r) \to \bm{\theta}^{(0)}$, with
convergence rate $O(1/(1+r))$ dominated by the $i = 1$ term.
As $r \to 0$ (i.e., $t \to 1$):
$\bm{\theta}(0) = \bm{\theta}^{(0)} + \sum_{i=1}^{T}\bm{\theta}^{(i)}$.
\end{remark}

\subsection{Quantum States and Born Distribution}

The output state is
$\ket{\psi(r)}=U(\bm{\theta}(r))\ket{0}^{\otimes n}$,
corresponding to the density operator
$\rho(r)=\ket{\psi(r)}\!\bra{\psi(r)}$.
The Born distribution over computational basis states is:
\begin{equation}\label{eq:born}
    P_r(x) = \tr(\ket{x}\bra{x}\rho(r))=\abs{\braket{x}{\psi(r)}}^2,
    \qquad x\in\{0,1\}^n.
\end{equation}

\subsection{Distance Measures}

We work with two levels of distance: at the quantum state level and at
the classical distribution level.

\begin{definition}[Quantum distances]\label{def:quantum_dist}
For two density operators $\rho,\sigma$ on
$\mathcal{H}=(\Complex^2)^{\otimes n}$:
\begin{itemize}
    \item \emph{Trace distance:}
    $T(\rho,\sigma)=\frac{1}{2}\tr\abs{\rho-\sigma}$,
    where $\abs{A}=\sqrt{A^\dagger A}$.
    \item \emph{Quantum fidelity:}
    $F(\rho,\sigma)
    =\bigl(\tr\sqrt{\sqrt{\rho}\,\sigma\,\sqrt{\rho}}\,\bigr)^2$.
\end{itemize}
For pure states $\rho=\ket{\psi_1}\!\bra{\psi_1}$ and
$\sigma=\ket{\psi_2}\!\bra{\psi_2}$, these simplify to:
\begin{equation}\label{eq:pure_state_dist}
    T(\rho,\sigma)=\sqrt{1-\abs{\braket{\psi_1}{\psi_2}}^2},
    \qquad
    F(\rho,\sigma)=\abs{\braket{\psi_1}{\psi_2}}^2.
\end{equation}
They are related by the Fuchs--van de Graaf inequalities \cite{fuchs1999cryptographic}:
\begin{equation}\label{eq:fvdg_quantum}
    1-\sqrt{F(\rho,\sigma)}
    \leq T(\rho,\sigma)
    \leq\sqrt{1-F(\rho,\sigma)}.
\end{equation}
\end{definition}

\begin{definition}[Classical distance]\label{def:classical_dist}
For two probability measures $P,Q$ on a measurable space
$(\mathcal{X},\mathcal{F})$, the \emph{total variation distance} is:
\begin{equation}\label{eq:tv_def}
    \TV(P,Q) = \sup_{A\in\mathcal{F}}\abs{P(A)-Q(A)}
    = \frac{1}{2}\int_{\mathcal{X}}\abs{p(x)-q(x)}\,d\mu(x)
\end{equation}
where $p,q$ are densities of $P,Q$ with respect to a common dominating
measure~$\mu$. When $\mathcal{X}$ is finite and $\mu$ is counting
measure, the integral reduces to $\frac{1}{2}\sum_x\abs{P(x)-Q(x)}$.
\end{definition}

The two levels are connected by the following fundamental fact.

\begin{proposition}[Measurement contraction]\label{prop:contraction}
Let $P_j$ denote the probability distribution obtained by measuring the
state~$\rho_j$ in any fixed orthonormal basis. Then:
\begin{equation}\label{eq:contraction}
    \TV(P_1,P_2)\leq T(\rho_1,\rho_2).
\end{equation}
\end{proposition}

\begin{proof}
This follows from the monotonicity of the trace distance under quantum channels (CPTP maps). A projective measurement $\mathcal{M}$ is a quantum channel mapping $\rho$ to a diagonal density matrix $\mathcal{M}(\rho) = \sum_x P(x)\ket{x}\bra{x}$. The trace distance between two such diagonal states is exactly the total variation distance between their diagonal entries:
\[
    T(\mathcal{M}(\rho_1), \mathcal{M}(\rho_2))
    = \frac{1}{2}\sum_x \abs{P_1(x) - P_2(x)}
    = \TV(P_1, P_2).
\]
Since the trace distance does not increase under quantum channels (data processing inequality \cite{nielsen2010quantum}), we have $T(\mathcal{M}(\rho_1), \mathcal{M}(\rho_2)) \leq T(\rho_1, \rho_2)$, which implies~\eqref{eq:contraction}.
\end{proof}

\subsection{Main Result}

\begin{definition}[Weighted dynamic norm]\label{def:Lambda}
The \emph{weighted dynamic norm} of order~$T$ is:
\begin{equation}\label{eq:Lambda_T}
    \Lambda_T
    \coloneqq \sum_{i=1}^{T} i \cdot
    \norm{\bm{\theta}^{(i)}}_1.
\end{equation}
\end{definition}

\begin{theorem}[Lipschitz continuity: uniform bound]
\label{thm:main}
Let $P_r$ be the Born distribution generated by the circuit
$U(\bm{\theta}(r))$ of Definition~\ref{def:circuit} with the
polynomial parameterization~\eqref{eq:theta_r}. Then for all
$r_1, r_2 \geq 0$:
\begin{equation}\label{eq:main_bound}
    \TV(P_{r_1}, P_{r_2})
    \leq \frac{1}{2}\,\Lambda_T
         \cdot \abs{r_1 - r_2}.
\end{equation}
\end{theorem}

A tighter, radius-dependent refinement is available:

\begin{definition}[Radius-dependent weighted norm]\label{def:Lambda_r}
For $r_1, r_2 \geq 0$, define:
\begin{equation}\label{eq:Lambda_T_r}
    \Lambda_T(r_1, r_2)
    \coloneqq \sum_{i=1}^{T} \norm{\bm{\theta}^{(i)}}_1
    \cdot \sum_{k=0}^{i-1}
    \frac{1}{(1+r_1)^k\,(1+r_2)^{i-1-k}}.
\end{equation}
\end{definition}

\begin{theorem}[Lipschitz continuity: refined bound]
\label{thm:refined}
Under the same hypotheses as Theorem~\ref{thm:main}:
\begin{equation}\label{eq:refined_bound}
    \TV(P_{r_1}, P_{r_2})
    \leq \frac{1}{2}\,\Lambda_T(r_1, r_2)
         \cdot \frac{\abs{r_1 - r_2}}{(1+r_1)(1+r_2)}.
\end{equation}
This is tighter than~\eqref{eq:main_bound} since $\Lambda_T(r_1,r_2) \le \Lambda_T$ for all
$r_1, r_2 \ge 0$. If $\bm{\theta}^{(i)} \ne \bm{0}$ for some $i \ge 2$, the inequality
is strict unless $r_1 = r_2 = 0$; if $\bm{\theta}^{(i)} = \bm{0}$ for all $i \ge 2$
(in particular when $T = 1$), then $\Lambda_T(r_1,r_2) = \Lambda_T$
identically and the two bounds coincide.
\end{theorem}

\subsection{Proof}

The proof proceeds through three lemmas, corresponding to the chain
$r\to\bm{\theta}(r)\to\ket{\psi(\bm{\theta})}\to P_r$.

\begin{lemma}[Parameter distance: exact factorization]
\label{lem:param}
For all $r_1, r_2 \geq 0$, with $t_k = 1/(1+r_k)$:
\begin{equation}\label{eq:param_exact}
    \bm{\theta}(r_1) - \bm{\theta}(r_2)
    = (t_1 - t_2) \sum_{i=1}^{T} \bm{\theta}^{(i)}
      \cdot S_i(t_1, t_2)
\end{equation}
where
\begin{equation}\label{eq:S_i}
    S_i(t_1, t_2) = \sum_{k=0}^{i-1} t_1^k\, t_2^{i-1-k}
\end{equation}
is the standard factorization polynomial satisfying
$t_1^i - t_2^i = (t_1 - t_2)\,S_i(t_1, t_2)$.

Consequently:
\begin{equation}\label{eq:param_bound}
    \norm{\bm{\theta}(r_1) - \bm{\theta}(r_2)}_1
    \leq \Lambda_T(r_1, r_2)
         \cdot \frac{\abs{r_1 - r_2}}{(1+r_1)(1+r_2)}.
\end{equation}
\end{lemma}

\begin{proof}
From~\eqref{eq:theta_t}:
\[
    \bm{\theta}(r_1) - \bm{\theta}(r_2)
    = \sum_{i=1}^{T} \bm{\theta}^{(i)}(t_1^i - t_2^i).
\]
The algebraic identity
$t_1^i - t_2^i = (t_1 - t_2)\sum_{k=0}^{i-1}t_1^k\,t_2^{i-1-k}$
yields~\eqref{eq:param_exact}.

Taking norms and applying the triangle inequality:
\begin{align}
    \norm{\bm{\theta}(r_1) - \bm{\theta}(r_2)}_1
    &= \abs{t_1 - t_2}\cdot
       \Bigl\|\sum_{i=1}^{T}\bm{\theta}^{(i)}
       \cdot S_i(t_1,t_2)\Bigr\|_1
    \notag\\
    &\leq \abs{t_1 - t_2}\cdot
          \sum_{i=1}^{T}\norm{\bm{\theta}^{(i)}}_1
          \cdot S_i(t_1,t_2).
    \label{eq:triangle_step}
\end{align}
Since $\abs{t_1 - t_2}
= \abs{r_1 - r_2}/\bigl((1+r_1)(1+r_2)\bigr)$
and $S_i(t_1,t_2) = \sum_{k=0}^{i-1}t_1^k\,t_2^{i-1-k}$,
substituting into~\eqref{eq:triangle_step}
gives~\eqref{eq:param_bound} with $\Lambda_T(r_1,r_2)$ as
defined in~\eqref{eq:Lambda_T_r}.
\end{proof}

\begin{corollary}[Uniform parameter bound]\label{cor:param_uniform}
Since $r_1, r_2 \in [0,+\infty)$, each term in $S_i$ satisfies
$t_1^k\,t_2^{i-1-k} \leq 1$, giving $S_i(t_1,t_2) \leq i$.
Therefore:
\begin{equation}\label{eq:param_uniform}
    \norm{\bm{\theta}(r_1) - \bm{\theta}(r_2)}_1
    \leq \Lambda_T
         \cdot \frac{\abs{r_1 - r_2}}{(1+r_1)(1+r_2)} \leq \Lambda_T \cdot \abs{r_1 - r_2}.
\end{equation}
\end{corollary}

\begin{lemma}\label{lem:state}
For any two parameter vectors
$\bm{\theta}_1,\bm{\theta}_2\in\Real^{N_p}$:
\begin{equation}\label{eq:state_lip}
    \norm{\ket{\psi(\bm{\theta}_1)}-\ket{\psi(\bm{\theta}_2)}}
    \leq\frac{1}{2}\,
         \norm{\bm{\theta}_1-\bm{\theta}_2}_1
\end{equation}
where $\norm{\cdot}$ denotes the Hilbert space $2$-norm.
\end{lemma}

\begin{proof}
We use the gate-level representation~\eqref{eq:shorthand}. The argument
proceeds in three stages.

\medskip
\noindent\textbf{Stage 1: Hybrid unitaries.}
For $j=1,2,\ldots,N_p+1$, define the \emph{hybrid unitary}
$\widetilde{U}^{(j)}$ that uses $\bm{\theta}_1$ for gates $1$ through
$j-1$ and $\bm{\theta}_2$ for gates $j$ through~$N_p$:
\begin{equation}\label{eq:hybrid}
    \widetilde{U}^{(j)}
    =\underbrace{G_{N_p}(\theta_{2,N_p})\cdots
                 G_{j}(\theta_{2,j})}_{\text{outer: }\bm{\theta}_2}
     \;\cdot\;
     \underbrace{G_{j-1}(\theta_{1,j-1})\cdots
                 G_{1}(\theta_{1,1})}_{\text{inner: }\bm{\theta}_1}
     \;\cdot\;H^{\otimes n}
\end{equation}
where $G_k(\theta)=C_k\cdot e^{-i\theta G_k/2}$ as
in~\eqref{eq:shorthand}. Every hybrid unitary contains all $N_p$
parameterized gates and the initial $H^{\otimes n}$; only the parameter
values assigned to each gate differ. At the two extremes:
\begin{align}
    \widetilde{U}^{(1)}
    &=G_{N_p}(\theta_{2,N_p})\cdots G_{1}(\theta_{2,1})
      \cdot H^{\otimes n}
    =U(\bm{\theta}_2),
    \label{eq:hybrid_left}\\
    \widetilde{U}^{(N_p+1)}
    &=G_{N_p}(\theta_{1,N_p})\cdots G_{1}(\theta_{1,1})
      \cdot H^{\otimes n}
    =U(\bm{\theta}_1).
    \label{eq:hybrid_right}
\end{align}

\medskip
\noindent\textbf{Stage 2: Telescoping and single-gate bound.}
The consecutive hybrids $\widetilde{U}^{(j)}$ and
$\widetilde{U}^{(j+1)}$ differ only at gate~$j$, which switches from
$\theta_{2,j}$ to~$\theta_{1,j}$. Telescoping:
\begin{equation}\label{eq:telescope}
    U(\bm{\theta}_1)-U(\bm{\theta}_2)
    =\widetilde{U}^{(N_p+1)}-\widetilde{U}^{(1)}
    =\sum_{j=1}^{N_p}
     \bigl[\widetilde{U}^{(j+1)}-\widetilde{U}^{(j)}\bigr].
\end{equation}
The difference at position~$j$ factors as:
\begin{equation}\label{eq:single_diff}
    \widetilde{U}^{(j+1)}-\widetilde{U}^{(j)}
    =A_j\cdot\Delta_j\cdot B_j
\end{equation}
where
\begin{equation}\label{eq:Delta_j}
    \Delta_j
    \coloneqq e^{-i\theta_{1,j}G_j/2}-e^{-i\theta_{2,j}G_j/2}
\end{equation}
and the outer and inner blocks are:
\begin{align}
    A_j &= G_{N_p}(\theta_{2,N_p})\cdots
           G_{j+1}(\theta_{2,j+1})\cdot C_j,
    \label{eq:Aj}\\
    B_j &= G_{j-1}(\theta_{1,j-1})\cdots
           G_{1}(\theta_{1,1})\cdot H^{\otimes n}.
    \label{eq:Bj}
\end{align}
Both $A_j$ and $B_j$ are unitary and independent of
$\theta_{1,j}$ and~$\theta_{2,j}$. The fixed unitary $C_j$ has been
absorbed into~$A_j$.

To verify~\eqref{eq:single_diff}: $\widetilde{U}^{(j+1)}$ and
$\widetilde{U}^{(j)}$ share the same $A_j$ and $B_j$; they differ
only in the rotation sandwiched between them.
In~$\widetilde{U}^{(j+1)}$ this rotation uses~$\theta_{1,j}$, while
in~$\widetilde{U}^{(j)}$ it uses~$\theta_{2,j}$. Subtracting
gives~\eqref{eq:single_diff}.

\medskip
\noindent\textbf{Bounding $\norm{\Delta_j}_{\op}$.}
Recall that the generator $G_j$ is a tensor product of a single-qubit Pauli matrix ($Y$ or $Z$) on one qubit and identity operators on the others. Consequently, $G_j$ has eigenvalues $\pm 1$, satisfies $G_j^2=I$, and has unit operator norm $\norm{G_j}_{\op}=1$.
Expanding the exponential via $e^{-i\alpha G/2}=\cos(\alpha/2)\,I-i\sin(\alpha/2)\,G$, we have:
\begin{equation}\label{eq:expand_diff}
    e^{-i\alpha G/2}-e^{-i\beta G/2}
    =\Bigl(\cos\tfrac{\alpha}{2}-\cos\tfrac{\beta}{2}\Bigr)I
     -i\Bigl(\sin\tfrac{\alpha}{2}-\sin\tfrac{\beta}{2}\Bigr)G.
\end{equation}
Since $G$ is Hermitian with $G^2 = I$, it is unitarily diagonalizable with
eigenvalues $\pm 1$. Writing $c = \cos\frac{\alpha}{2}-\cos\frac{\beta}{2}$
and $s = \sin\frac{\alpha}{2}-\sin\frac{\beta}{2}$, the operator $cI - isG$
is diagonal in an eigenbasis of $G$ with diagonal entries $c \mp is$, each
of modulus $\sqrt{c^2+s^2}$. Hence
\begin{equation}\label{eq:single_gate_bound}
\| e^{-i\alpha G/2} - e^{-i\beta G/2} \|_{\op}
= \sqrt{c^2+s^2}
= \sqrt{2-2\cos\tfrac{\alpha-\beta}{2}}
= 2\Big|\sin\tfrac{\alpha-\beta}{4}\Big|
\le \frac{|\alpha-\beta|}{2},
\end{equation}
where the last step uses $\abs{\sin x}\leq\abs{x}$.

\medskip
\noindent\textbf{Eliminating $A_j$ and $B_j$ from the norm.}
Applying both sides of~\eqref{eq:single_diff} to
$\ket{\phi}=\ket{0}^{\otimes n}$:
\begin{align}
    \norm{\bigl[\widetilde{U}^{(j+1)}-\widetilde{U}^{(j)}\bigr]
          \ket{\phi}}
    &=\norm{A_j\cdot\Delta_j\cdot B_j\ket{\phi}}
    \notag\\
    &=\norm{\Delta_j\cdot B_j\ket{\phi}}
    \label{eq:drop_Aj}\\
    &\leq\norm{\Delta_j}_{\op}\cdot
          \underbrace{\norm{B_j\ket{\phi}}}_{=\,1}
    \label{eq:drop_Bj}\\
    &\leq\frac{\abs{\theta_{1,j}-\theta_{2,j}}}{2}.
    \label{eq:term_bound}
\end{align}
Step~\eqref{eq:drop_Aj} uses the fact that $A_j$ is unitary:
$\norm{A_j\ket{w}}^2=\bra{w}A_j^\dagger A_j\ket{w}=\norm{\ket{w}}^2$,
so $A_j$ can be removed as an exact equality.
Step~\eqref{eq:drop_Bj} applies the definition
$\norm{\Delta_j}_{\op}
=\sup_{\norm{\ket{u}}=1}\norm{\Delta_j\ket{u}}$
to the unit vector $B_j\ket{\phi}$ (unit norm because $B_j$ is unitary).
Step~\eqref{eq:term_bound} substitutes the single-gate
bound~\eqref{eq:single_gate_bound}.

\medskip
\noindent\textbf{Stage 3: Summation.}
Applying the triangle inequality to~\eqref{eq:telescope}:
\begin{equation}\label{eq:l1bound}
    \norm{\ket{\psi(\bm{\theta}_1)}-\ket{\psi(\bm{\theta}_2)}}
    \leq\sum_{j=1}^{N_p}\frac{\abs{\theta_{1,j}-\theta_{2,j}}}{2}
    =\frac{1}{2}\norm{\bm{\theta}_1-\bm{\theta}_2}_1.
    \qedhere
\end{equation}
\end{proof}

\begin{corollary}[$\ell_2$ version]\label{cor:state_l2}
By the Cauchy--Schwarz inequality
$\norm{\mathbf{a}}_1\leq\sqrt{N_p}\,\norm{\mathbf{a}}_2$,
Lemma~\ref{lem:state} implies
\begin{equation}\label{eq:l2bound}
    \norm{\ket{\psi(\bm{\theta}_1)}-\ket{\psi(\bm{\theta}_2)}}
    \leq\frac{\sqrt{N_p}}{2}\,
         \norm{\bm{\theta}_1-\bm{\theta}_2}_2.
\end{equation}
This $\ell_2$ form is the one used in the covering-number
arguments of Section~\ref{sec:generic}, where the parameter
cube $[0,2\pi)^{N_p}$ is covered by Euclidean balls.
\end{corollary}

\begin{lemma}\label{lem:born}
For any two pure states
$\ket{\psi_1},\ket{\psi_2}\in(\Complex^2)^{\otimes n}$, let
$P_j(x)=\abs{\braket{x}{\psi_j}}^2$ be the Born distributions
obtained by measurement in the computational basis. Then:
\begin{equation}\label{eq:born_lip}
    \TV(P_1,P_2)\leq\norm{\ket{\psi_1}-\ket{\psi_2}}.
\end{equation}
\end{lemma}

\begin{proof}

Let $\rho_j=\ket{\psi_j}\!\bra{\psi_j}$. By
Proposition~\ref{prop:contraction},
$\TV(P_1,P_2)\leq T(\rho_1,\rho_2)$. For pure states, the trace
distance satisfies the standard bound:
\begin{equation}\label{eq:trace_hilbert}
    T(\rho_1,\rho_2)\leq\norm{\ket{\psi_1}-\ket{\psi_2}}.
\end{equation}
To verify~\eqref{eq:trace_hilbert}: set
$\ket{\delta}=\ket{\psi_1}-\ket{\psi_2}$. Then
$\rho_1-\rho_2
=\ket{\psi_1}\!\bra{\psi_1}-\ket{\psi_2}\!\bra{\psi_2}
=\ket{\delta}\!\bra{\psi_1}+\ket{\psi_2}\!\bra{\delta}$,
so $\tr\abs{\rho_1-\rho_2}\leq
\norm{\ket{\delta}}\norm{\ket{\psi_1}}
+\norm{\ket{\psi_2}}\norm{\ket{\delta}}
=2\norm{\ket{\delta}}$, giving
$T=\frac{1}{2}\tr\abs{\rho_1-\rho_2}\leq\norm{\ket{\delta}}$.
Combining: $\TV(P_1,P_2)\leq T(\rho_1,\rho_2)
\leq\norm{\ket{\psi_1}-\ket{\psi_2}}$.

\end{proof}

\begin{proof}[Proof of Theorem~\ref{thm:refined}]
Composing the three lemmas:
\begin{align}
    \TV(P_{r_1},P_{r_2})
    &\leq \norm{\ket{\psi(r_1)}-\ket{\psi(r_2)}}
        &&\text{(Lemma~\ref{lem:born})}
    \label{eq:chain1}\\
    &\leq \frac{1}{2}\,
          \norm{\bm{\theta}(r_1)-\bm{\theta}(r_2)}_1
        &&\text{(Lemma~\ref{lem:state})}
    \label{eq:chain2}\\
    &\leq \frac{1}{2}\,\Lambda_T(r_1,r_2)
          \cdot\frac{\abs{r_1-r_2}}{(1+r_1)(1+r_2)}
        &&\text{(Lemma~\ref{lem:param})}
    \label{eq:chain3}
\end{align}
which is~\eqref{eq:refined_bound}.
\end{proof}

\begin{proof}[Proof of Theorem~\ref{thm:main}]
Apply Corollary~\ref{cor:param_uniform} to
replace~\eqref{eq:chain3} with
$\Lambda_T(r_1,r_2) \leq \Lambda_T$.
\end{proof}

\subsection{Consequences}

\begin{corollary}[Convergence to the base distribution]
\label{cor:asymptotic}
Let $P_\infty(x)
= \abs{\bra{x}U(\bm{\theta}^{(0)})\ket{0}^{\otimes n}}^2$
be the distribution at $r = \infty$. Then:
\begin{equation}\label{eq:asymptotic}
    \TV(P_r, P_\infty)
    \leq \frac{1}{2}
         \sum_{i=1}^{T}
         \frac{\norm{\bm{\theta}^{(i)}}_1}{(1+r)^i}
    = \frac{1}{2}\,
      \frac{\norm{\bm{\theta}^{(1)}}_1}{1+r}
      + O\!\left(\frac{1}{(1+r)^2}\right).
\end{equation}
\end{corollary}

\begin{proof}
By Lemmas~\ref{lem:born} and~\ref{lem:state},
$\TV(P_r, P_\infty) \le \frac{1}{2}\,
\|\bm{\theta}(r) - \bm{\theta}^{(0)}\|_1$. From \eqref{eq:theta_t}, with $t = 1/(1+r)$,
\[
\bm{\theta}(r) - \bm{\theta}^{(0)} = \sum_{i=1}^{T} \bm{\theta}^{(i)}\, t^i,
\qquad\text{so}\qquad
\|\bm{\theta}(r) - \bm{\theta}^{(0)}\|_1 \le \sum_{i=1}^{T}
\frac{\|\bm{\theta}^{(i)}\|_1}{(1+r)^i}
\]
by the triangle inequality, which gives~\eqref{eq:asymptotic} directly.
\end{proof}

\begin{corollary}[Fidelity bound]\label{cor:fidelity}
Let $F\big(\rho(r_1),\rho(r_2)\big) = |\langle\psi(r_1)|\psi(r_2)\rangle|^2$
and define
\[
b \;\coloneqq\; 1 - \frac{1}{8}\,\Lambda_T(r_1,r_2)^2 \cdot
\frac{(r_1-r_2)^2}{(1+r_1)^2 (1+r_2)^2}.
\]
Then
\[
F\big(\rho(r_1),\rho(r_2)\big) \;\ge\; \big(\max\{b,\,0\}\big)^2 .
\]
The same holds with $\Lambda_T(r_1,r_2)$ replaced by $\Lambda_T$ (uniform
version, weaker).
\end{corollary}

\begin{proof}
Write $|\delta\rangle = |\psi(r_1)\rangle - |\psi(r_2)\rangle$. For unit
vectors, $\||\delta\rangle\|^2 = 2 - 2\text{Re}\langle\psi(r_1)|\psi(r_2)\rangle$,
hence
\[
\sqrt{F} \;=\; |\langle\psi(r_1)|\psi(r_2)\rangle|
\;\ge\; \text{Re}\langle\psi(r_1)|\psi(r_2)\rangle
\;=\; 1 - \tfrac{1}{2}\||\delta\rangle\|^2 .
\]
Since $\sqrt{F} \ge 0$ as well, $\sqrt{F} \ge
\max\{\,1 - \tfrac12\||\delta\rangle\|^2,\; 0\,\}$. By Lemmas~\ref{lem:state} and~\ref{lem:param},
$\||\delta\rangle\|^2 \le \frac{1}{4}\Lambda_T(r_1,r_2)^2\,
\frac{(r_1-r_2)^2}{(1+r_1)^2(1+r_2)^2}$, so
$1 - \tfrac12\||\delta\rangle\|^2 \ge b$ and therefore
$\sqrt{F} \ge \max\{b, 0\}$. Since both sides are non-negative, $F \;\ge\; \big(\max\{b,\,0\}\big)^2$.
\end{proof}

\begin{remark}[Unconditional linear bound]\label{rem:linear-fid}
The linear bound
\[
F \;\ge\; 1 - \frac{1}{4}\,\Lambda_T(r_1,r_2)^2\cdot
\frac{(r_1-r_2)^2}{(1+r_1)^2(1+r_2)^2}
\]
holds \emph{unconditionally}: from
$F \ge (\text{Re}\langle\psi_1|\psi_2\rangle)^2 =
\big(1-\tfrac12\||\delta\rangle\|^2\big)^2 \ge 1 - \||\delta\rangle\|^2$,
substitute the bound on $\||\delta\rangle\|^2$; when the right-hand side is
negative the inequality is trivially true. Unlike the quadratic bound of
Corollary~\ref{cor:fidelity}, no truncation is needed here.
\end{remark}

\begin{corollary}[Binning approximation error---relaxed bound]
\label{cor:binning_relaxed}
With bins $[a_k, a_{k+1})$ and representative $r_k^* \in [a_k,a_{k+1})$:
\begin{equation}\label{eq:bin_relaxed}
    \sup_{r \in [a_k,a_{k+1})}\TV(P_r, P_{r_k^*})
    \leq \frac{1}{2}\,\Lambda_T
         \cdot\frac{a_{k+1}-a_k}{(1+a_k)^2}.
\end{equation}
\end{corollary}

\begin{proof}
Use $\abs{r-r_k^*} \leq a_{k+1}-a_k$,
$(1+r)(1+r_k^*) \geq (1+a_k)^2$, and
$\Lambda_T(r,r_k^*) \leq \Lambda_T$ in
Theorem~\ref{thm:refined}.
\end{proof}

\begin{remark}
The relaxation $(1+r)(1+r_k^*) \geq (1+a_k)^2$ replaces
the exact product of two distinct terms by the square of the
smaller term. This is standard but discards the algebraic
structure of the exact denominator. The next subsection
shows that retaining the exact denominator leads to a
closed-form optimal binning strategy.
\end{remark}

The uniform bound from Theorem~\ref{thm:main} preserves a
telescoping structure that yields an explicit formula for the
optimal bin boundaries. This structure holds for \emph{arbitrary}
polynomial order~$T$.

\begin{definition}[Per-bin error]\label{def:exact_bin_error}
Let $C = \frac{1}{2}\Lambda_T$ (for any $T \geq 1$).
The worst-case error for bin $[a_k, a_{k+1})$ with
representative at the left endpoint is:
\begin{equation}\label{eq:exact_bin_error}
    E_k = C \cdot \frac{a_{k+1} - a_k}{(1+a_k)(1+a_{k+1})}.
\end{equation}
The per-bin error admits the decomposition:
\begin{equation}\label{eq:telescope_identity}
    E_k = C \cdot \left(\frac{1}{1+a_k}
          - \frac{1}{1+a_{k+1}}\right).
\end{equation}
This identity is purely algebraic and holds for all~$T$.
\end{definition}

\begin{theorem}[Closed-form adaptive binning]\label{thm:binning}
Suppose the radial domain $[0, R_{\max}]$ is partitioned into $K$
bins of equal worst-case error $\varepsilon$, with $a_0 = 0$ and
$a_K = R_{\max}$. Then:

\begin{enumerate}
\item \textbf{Per-bin error:}
\begin{equation}\label{eq:eps_from_K}
    \varepsilon = \frac{C}{K}\cdot\frac{R_{\max}}{1+R_{\max}},
    \qquad C = \frac{1}{2}\Lambda_T.
\end{equation}

\item \textbf{Bin boundaries:}
\begin{equation}\label{eq:bin_formula}
    \boxed{a_k = \frac{1}
    {1 - \dfrac{k}{K}\cdot\dfrac{R_{\max}}{1+R_{\max}}} - 1,
    \qquad k = 0, 1, \ldots, K.}
\end{equation}

\item \textbf{Bin widths:}
\begin{equation}\label{eq:bin_widths}
    \Delta_k = a_{k+1} - a_k
    = \frac{\varepsilon}{C}\cdot(1+a_k)(1+a_{k+1}).
\end{equation}
\end{enumerate}
\end{theorem}

\begin{proof}
\textbf{Step 1: Equal-error condition linearizes in $w$-space.}
Define $w_k = 1/(1+a_k)$, so $w_k \in (0,1]$ with
$w_0 = 1$ and $w_K = 1/(1+R_{\max})$.
By Definition~\ref{def:exact_bin_error}, the equal-error condition
$E_k = \varepsilon$ for all~$k$ becomes:
\begin{equation}\label{eq:arith_prog}
    w_k - w_{k+1} = \frac{\varepsilon}{C}
    \qquad\text{(constant)}.
\end{equation}
Therefore $\{w_k\}_{k=0}^{K}$ is an \emph{arithmetic progression}
with common difference $d = -\varepsilon/C$:
\begin{equation}\label{eq:w_k}
    w_k = 1 - \frac{k\varepsilon}{C}.
\end{equation}

\textbf{Step 2: Determine $\varepsilon$.}
The boundary condition $w_K = 1/(1+R_{\max})$ gives:
\[
    1 - \frac{K\varepsilon}{C} = \frac{1}{1+R_{\max}},
    \qquad
    \frac{K\varepsilon}{C} = \frac{R_{\max}}{1+R_{\max}},
\]
yielding~\eqref{eq:eps_from_K}.

Equivalently, this follows from the telescoping sum:
\begin{equation}\label{eq:telescope_sum}
    K\varepsilon = \sum_{k=0}^{K-1} E_k
    = C \sum_{k=0}^{K-1}\left(\frac{1}{1+a_k}
      - \frac{1}{1+a_{k+1}}\right)
    = C\left(1 - \frac{1}{1+R_{\max}}\right).
\end{equation}

\textbf{Step 3: Invert to obtain $a_k$.}
Substituting~\eqref{eq:eps_from_K} into~\eqref{eq:w_k}:
\[
    w_k = 1 - \frac{k}{K}\cdot\frac{R_{\max}}{1+R_{\max}}.
\]
Note that $\varepsilon/C$ has been replaced by
$R_{\max}/(K(1+R_{\max}))$, which is independent of $C$
(and hence independent of $\Lambda_T$).
Inverting $a_k = 1/w_k - 1$ yields~\eqref{eq:bin_formula}.
The width formula~\eqref{eq:bin_widths} follows from
$\Delta_k = 1/w_{k+1} - 1/w_k
= (w_k - w_{k+1})/(w_k w_{k+1})
= (\varepsilon/C)\cdot(1+a_k)(1+a_{k+1})$.
\end{proof}

\begin{remark}[Independence from $\Lambda_T$]
\label{rem:Lambda_cancellation}
The bin boundary formula~\eqref{eq:bin_formula} depends only
on $k$, $K$, and $R_{\max}$. The constant
$C = \frac{1}{2}\Lambda_T$ enters the per-bin
error~\eqref{eq:eps_from_K} and the bin
widths~\eqref{eq:bin_widths}, but cancels exactly from the
boundary positions. Concretely, the ratio
$\varepsilon/C = R_{\max}/(K(1+R_{\max}))$ that appears in
the arithmetic progression~\eqref{eq:w_k} contains no reference
to~$\Lambda_T$.

This means that the \emph{geometric structure} of the optimal
grid---which fraction of the domain each bin occupies---is
\emph{universal}: it depends only on the domain $[0,R_{\max}]$
and the number of bins~$K$, not on the circuit parameters.
The circuit parameters affect only the \emph{error level}
$\varepsilon$ achieved by a given $(K, R_{\max})$ pair.
\end{remark}

\begin{corollary}[Design procedure]\label{cor:design}
Given a tolerance $\varepsilon > 0$ and circuit parameters:

\emph{Step 1 (tail region).} Compute $R_{\max}$ solving
$\frac{1}{2} \sum_{i=1}^{T}
\|\bm{\theta}^{(i)}\|_1 (1+R_{\max})^{-i} = \varepsilon$, and assign to every
$r \ge R_{\max}$ the representative distribution $P_\infty$ (parameters
$\bm{\theta}^{(0)}$). By Corollary~\ref{cor:asymptotic}, the error on $[R_{\max}, \infty)$ is at
most $\varepsilon$.

\emph{Step 2 (bounded region).} Set
$K = \big\lceil \frac{C}{\varepsilon}\cdot
\frac{R_{\max}}{1+R_{\max}} \big\rceil$ with
$C = \frac{1}{2}\Lambda_T$, and compute the bin boundaries on
$[0, R_{\max})$ from \eqref{eq:bin_formula}. By Theorem~\ref{thm:binning}, each bin incurs error at most
$\varepsilon$.

Since the two regions are disjoint, the overall worst-case error is
$\sup_{r \ge 0} \TV\big(P_r, P_{\mathrm{rep}(r)}\big) \le \varepsilon$
(a maximum over regions, not a sum).
\end{corollary}
\section{Reachable Distribution of the PQC}
\label{sec:reachable}

\subsection{Circuit Architecture}

We consider the PQC architecture defined in our paper:
$n$ qubits, $L$ entangling layers with YZY single-qubit blocks
and circular CNOT entangling layers, giving
$N_p = 3n(L+1)$ variational parameters. With $L = 2n$
(the empirical scaling rule), $N_p = 3n(2n+1) = 6n^2 + 3n$.

The circuit generates a pure state
$\ket{\psi(\bm{\theta})} = U(\bm{\theta})\ket{0}^{\otimes n}$
and the Born distribution
$P_{\bm{\theta}}(x) = |\braket{x}{\psi(\bm{\theta})}|^2$
for $x \in \{0,1\}^n$.

\begin{definition}[Reachable set]\label{def:reachable}
The \emph{reachable set} of the PQC is
\begin{equation}\label{eq:reachable}
    \mathcal{B} = \bigl\{P_{\bm{\theta}}^{\mathrm{Born}}
    : \bm{\theta} \in [0,2\pi)^{N_p}\bigr\}
    \subset \Delta^{M-1}
\end{equation}
where $M = 2^n$ and
$\Delta^{M-1} = \{P \in \Real^M : P_i \geq 0,\;\sum_i P_i = 1\}$
is the probability simplex.
\end{definition}

\subsection{Target Distributions and Grid Bijection}

The target is a probability distribution on the standard
$(D{-}1)$-simplex $\Delta^{D-1}$, discretized via an
augmented Voronoi partition.

\begin{definition}[Grid encoding constraints]\label{def:encoding-constraints}
Let $\mathcal{Q}_N$ be an augmented Voronoi grid
(Definition~\ref{def:augmented_grid}), with symmetry group $G$ acting on
it (in our setting $G = S_D$; Section~\ref{sec:equivariance}). Two
distinct cardinality constraints arise depending on the encoding:
\emph{(a) direct (full-space) QCBM:} the augmentation is chosen so that
$|\mathcal{Q}_N| = 2^n$, and the bijection
$\phi: \mathcal{Q}_N \leftrightarrow \{0,1\}^n$ maps
grid points to basis states;
\emph{(b) orbit-space QCBM (Strategy B):} the augmentation is chosen so
that the number of $G$-orbits satisfies
$\mathcal{O}_G(\mathcal{Q}_N) = 2^{n'}$, and the bijection
$\psi$ maps orbits to basis states of $n'$ qubits. The two constraints are
generally incompatible and are never imposed simultaneously.
\end{definition}

The Born distribution $P_{\bm{\theta}}^{\mathrm{Born}}$ and
the target distribution $P^{\mathrm{target}}$ are both
probability vectors of the same dimension $M = 2^n$, with
the $x$-th component representing the probability mass
assigned to the Voronoi cell $\mathrm{cell}(\phi^{-1}(x))$ of the grid
point $\phi^{-1}(x)$, where $\phi$ is the encoding bijection of
Definition~\ref{def:encoding-constraints} and $\mathrm{cell}(\cdot)$
denotes the Voronoi cell of a grid point:
\begin{equation}\label{eq:target_born_correspondence}
    P^{\mathrm{target}}(x) = \int_{\mathrm{cell}(\phi^{-1}(x))} f(\mathbf{s})\,
    d\lambda_{D-1}(\mathbf{s}), \qquad
    P_{\bm{\theta}}^{\mathrm{Born}}(x)
    = |\braket{x}{\psi(\bm{\theta})}|^2.
\end{equation}
The measure-theoretic consistency of this discretization---monotone
inner--outer approximation of Borel sets, $L^1$/total-variation
convergence of the induced histogram as $N\to\infty$, and
identifiability---is established in Section~\ref{sec:consistency}.

We prove that when $n \to \infty$, most distributions in
$\Delta^{M-1}$ cannot be approximated by $\mathcal{B}$.

\subsection{Metric Entropy of the Reachable Set}

\begin{definition}[Covering number]\label{def:covering}
For a subset $S$ of a metric space $(X,d)$ and $\varepsilon > 0$,
the \emph{$\varepsilon$-covering number} $\mathcal{N}(S,\varepsilon,d)$
is the minimum number of closed $d$-balls of radius $\varepsilon$
needed to cover~$S$. When the metric is clear from context,
we write $\mathcal{N}(S,\varepsilon)$.
\end{definition}

\begin{proposition}[Covering number of $\mathcal{B}$]
\label{prop:covering_B}
The reachable set $\mathcal{B}$ satisfies:
\begin{equation}\label{eq:covering_B}
    \log\mathcal{N}(\mathcal{B},\varepsilon,\TV)
    \leq N_p \cdot \log\left(\frac{\pi N_p}{\varepsilon}+1\right)
\end{equation}
\end{proposition}

\begin{proof}
The Born distribution map
$F: \bm{\theta} \mapsto P_{\bm{\theta}}^{\mathrm{Born}}$
satisfies
$\TV(P_{\bm{\theta}_1}, P_{\bm{\theta}_2})
\leq L_F\norm{\bm{\theta}_1 - \bm{\theta}_2}_2$
with $L_F = \sqrt{N_p}/2$
(Lemma~\ref{lem:born} combined with Corollary~\ref{cor:state_l2}:
$\TV(P_{\bm{\theta}_1}, P_{\bm{\theta}_2}) \leq \norm{\ket{\psi_1}-\ket{\psi_2}}
\leq \frac{\sqrt{N_p}}{2}\norm{\bm{\theta}_1-\bm{\theta}_2}_2$).

For every $\varepsilon > 0$ there exist $K \in \mathbb{Z}^+$ and points
$\{\bm{\theta}_1^*, \ldots, \bm{\theta}_K^*\}$ forming an
$(\varepsilon/L_F)$-cover of $[0,2\pi)^{N_p}$ in the
Euclidean metric, i.e., for every
$\bm{\theta} \in [0,2\pi)^{N_p}$ there exists some
$\bm{\theta}_j^*$ with
$\norm{\bm{\theta} - \bm{\theta}_j^*}_2 \leq \varepsilon/L_F$.
The number of such $\norm{\cdot}_2$-balls is
$K = \mathcal{N}\bigl([0,2\pi)^{N_p},\,
\varepsilon/L_F,\,\norm{\cdot}_2\bigr)$.

We claim that the image set
$\{F(\bm{\theta}_1^*), \ldots, F(\bm{\theta}_K^*)\}$
is an $\varepsilon$-cover of $\mathcal{B}$ in the TV metric.
Indeed, for any $P_{\bm{\theta}} \in \mathcal{B}$, the parameter
$\bm{\theta}$ is covered by some $\bm{\theta}_j^*$, and
the Lipschitz condition gives:
\[
    \TV\bigl(P_{\bm{\theta}},\, P_{\bm{\theta}_j^*}\bigr)
    \leq L_F \cdot
    \norm{\bm{\theta} - \bm{\theta}_j^*}_2
    \leq L_F \cdot \frac{\varepsilon}{L_F}
    = \varepsilon.
\]
Therefore $P_{\bm{\theta}_j^*}$ is within TV distance
$\varepsilon$ of $P_{\bm{\theta}}$, so the $K$ image points
cover $\mathcal{B}$. Since the minimum covering number cannot exceed
the size of any valid cover:
\begin{equation}\label{eq:lip_transfer}
    \mathcal{N}(\mathcal{B},\,\varepsilon,\,\TV)
    \leq K
    = \mathcal{N}\bigl([0,2\pi)^{N_p},\,
    \varepsilon/L_F,\,\norm{\cdot}_2\bigr).
\end{equation}

The parameter domain $[0,2\pi)^{N_p}$ is contained in a
Euclidean ball of radius $\pi\sqrt{N_p}$ (since the diameter of the $N_p$-dimensional
hypercube $[0,2\pi)^{N_p}$ is $2\pi\sqrt{N_p}$).
By the standard volumetric bound for covering a bounded set
in $\Real^{N_p}$:
\[
    \mathcal{N}\bigl([0,2\pi)^{N_p},\,
    \varepsilon/L_F,\,\norm{\cdot}_2\bigr)
    \leq \left(\frac{2\pi\sqrt{N_p}\cdot L_F}{\varepsilon}
    + 1\right)^{N_p}.
\]
Substituting $L_F = \sqrt{N_p}/2$:
\[
    \mathcal{N}(\mathcal{B},\,\varepsilon,\,\TV)
    \leq \left(\frac{\pi N_p}{\varepsilon}+1\right)^{N_p}
\]
Taking logarithms yields~\eqref{eq:covering_B}.
\end{proof}

\begin{remark}[Architecture independence]\label{rem:arch-indep}
Lemma~\ref{lem:state} and Proposition~\ref{prop:covering_B} use only two structural facts: (i) each
parameterized gate is generated by a Hermitian involution
($G_j^2 = I$, $\|G_j\|_{\op} = 1$), and (ii) all non-parameterized gates
are fixed unitaries absorbed into the $C_j$. Both hold verbatim for the
circular-CNOT ansatz (Definition~\ref{def:circuit}) and the full-CZ ansatz (Definition~\ref{def:full_entangle})
with independent parameters. For the shared-parameter qubit-permutation reference ansatz of Theorem~\ref{thm:full_equiv_lie}, a tied parameter $\theta$ multiplying the
commuting orbit sum $\sum_{q \in O} G_q$ contributes
$|O| \cdot |\theta_{1} - \theta_{2}|/2$ to~\eqref{eq:l1bound} (factor the exponential into
$|O|$ commuting single-qubit rotations), so the Lipschitz constant acquires
orbit-size weights while all qualitative conclusions are unchanged.
\end{remark}

\begin{proposition}[Covering number of the simplex]
\label{prop:covering_simplex}
The probability simplex satisfies: for all
$\varepsilon \in (0,1/4)$,
\begin{equation}\label{eq:covering_simplex}
    \log\mathcal{N}(\Delta^{M-1},\varepsilon,\TV)
    \geq (M-1)\log\frac{1}{4\varepsilon}.
\end{equation}
\end{proposition}

\begin{proof}
We use a volume comparison argument. The probability simplex
$\Delta^{M-1}$ has $(M{-}1)$-dimensional volume
$Vol_{M-1}(\Delta^{M-1}) = \sqrt{M}/(M-1)!$
with respect to the Lebesgue measure on the affine hyperplane
$\sum_i P_i = 1$.

A TV-ball of radius $\varepsilon$ centered at $P^0 \in \Delta^{M-1}$
is $B_\varepsilon(P^0) = \{P \in \Delta^{M-1}
: \TV(P,P^0) \leq \varepsilon\}
= \{P \in \Delta^{M-1}
: \norm{P-P^0}_1 \leq 2\varepsilon\}$.
The intersection of the $\ell^1$-ball with the simplex has
volume at most that of the full $\ell^1$-ball intersected with
the hyperplane, which satisfies:
\[
    Vol_{M-1}(B_\varepsilon(P^0))
    \leq \frac{(4\varepsilon)^{M-1}}{(M-1)!}\cdot\sqrt{M}
\]
(this is the volume of the $(M{-}1)$-dimensional cross-polytope
of radius $2\varepsilon$).

Since the covering must satisfy
$\mathcal{N}(\Delta^{M-1},\varepsilon,\TV)
\cdot Vol_{M-1}(B_\varepsilon)
\geq Vol_{M-1}(\Delta^{M-1})$,
we obtain:
\[
    \mathcal{N}(\Delta^{M-1},\varepsilon,\TV)
    \geq \frac{Vol_{M-1}(\Delta^{M-1})}
    {Vol_{M-1}(B_\varepsilon)}
    \geq \left(\frac{1}{4\varepsilon}\right)^{M-1}.
\]
Taking logarithms yields~\eqref{eq:covering_simplex}.
\end{proof}

\subsection{Main Hardness Theorem}\label{sec:generic}

\begin{theorem}[Generic inapproximability]
\label{thm:generic}
Let $\mathcal{B} \subset \Delta^{M-1}$ be the reachable set of a
PQC with $N_p$ parameters. Define $\alpha = N_p/(M-1)$ and
$C_{\mathrm{Lip}} = \pi N_p + 1$.
If $\alpha < 1$ (i.e., $N_p < M - 1$), then there exist
distributions in $\Delta^{M-1}$ whose TV distance to $\mathcal{B}$
is at least:
\begin{equation}\label{eq:eps_star}
    \varepsilon^*
    = \left(\frac{1}{8}\right)^{1/(1-\alpha)}
    \cdot\, C_{\mathrm{Lip}}^{-\alpha/(1-\alpha)}.
\end{equation}
\end{theorem}

\begin{proof}
We first verify that $\varepsilon^* < 1/8$. Since $\alpha < 1$,
the exponent $1/(1-\alpha) \geq 1$, so
$(1/8)^{1/(1-\alpha)} \leq 1/8$. Since $C_{\mathrm{Lip}} > 1$
and $\alpha/(1-\alpha) > 0$, we have
$C_{\mathrm{Lip}}^{-\alpha/(1-\alpha)} < 1$. Therefore
$\varepsilon^* < 1/8$.

Now suppose for contradiction that there exists
$\varepsilon < \varepsilon^*$ such that $\mathcal{B}$
$\varepsilon$-approximates every distribution in
$\Delta^{M-1}$, i.e.,
$\sup_{P \in \Delta^{M-1}}\inf_{\bm{\theta}}
\TV(P, P_{\bm{\theta}}) \leq \varepsilon$.
Since $\varepsilon < \varepsilon^* < 1$, for any
$P \in \Delta^{M-1}$ there exists $Q \in \mathcal{B}$ with
$\TV(P,Q) \leq \varepsilon$, and for any $\varepsilon$-cover
$\{b_1,\ldots,b_K\}$ of $\mathcal{B}$, there exists $b_j$ with
$\TV(Q,b_j) \leq \varepsilon$. By the triangle inequality:
\[
    \TV(P, b_j) \leq \TV(P,Q) + \TV(Q,b_j)
    \leq 2\varepsilon.
\]
Therefore every $\varepsilon$-cover of $\mathcal{B}$ is a
$2\varepsilon$-cover of $\Delta^{M-1}$:
\begin{equation}\label{eq:cover_chain}
    \mathcal{N}(\Delta^{M-1}, 2\varepsilon, \TV)
    \leq \mathcal{N}(\mathcal{B}, \varepsilon, \TV).
\end{equation}
Substituting Proposition~\ref{prop:covering_simplex}
(with argument $2\varepsilon$, noting $2\varepsilon
< 2\varepsilon^* < 1/4$)
and Proposition~\ref{prop:covering_B}:
\begin{equation}\label{eq:cover_ineq}
    \left(\frac{1}{8\varepsilon}\right)^{M-1}
    \leq \left(\frac{\pi N_p}{\varepsilon}+1\right)^{N_p}
    \leq \left(\frac{C_{\mathrm{Lip}}}{\varepsilon}
    \right)^{N_p}
\end{equation}
where the second inequality uses
$\pi N_p/\varepsilon + 1
\leq (\pi N_p + 1)/\varepsilon
= C_{\mathrm{Lip}}/\varepsilon$
(valid since $\varepsilon < 1$).
Taking logarithms of the outer inequality:
\[
    (M-1)\log\frac{1}{8\varepsilon}
    \leq N_p \log\frac{C_{\mathrm{Lip}}}{\varepsilon}.
\]
Setting $u = \log(1/\varepsilon) > 0$, $\alpha = N_p/(M-1)$:
\[
    (M-1)(u - \log 8) \leq N_p(\log C_{\mathrm{Lip}} + u)
\]
\[
    (1-\alpha)\,u \leq \alpha\log C_{\mathrm{Lip}} + \log 8
\]
\[
    u \leq \frac{\alpha\log C_{\mathrm{Lip}} + \log 8}
    {1-\alpha}
\]
Therefore:
\[
    \varepsilon = e^{-u}
    \geq \exp\!\left(
    -\frac{\alpha\log C_{\mathrm{Lip}} + \log 8}{1-\alpha}
    \right)
    = \left(\frac{1}{8}\right)^{1/(1-\alpha)}
    \cdot\,C_{\mathrm{Lip}}^{-\alpha/(1-\alpha)}
    = \varepsilon^*.
\]
This contradicts the assumption $\varepsilon < \varepsilon^*$.
Therefore no such $\varepsilon$ exists, i.e.,
$\sup_{P \in \Delta^{M-1}} \inf_{\bm{\theta}}
\TV(P, P_{\bm{\theta}}) \ge \varepsilon^*$.
The map $f(P) = \inf_{\bm{\theta}} \TV(P, P_{\bm{\theta}})$ is
$1$-Lipschitz with respect to $\TV$ (triangle inequality), hence
continuous, and $\Delta^{M-1}$ is compact; the supremum is therefore
attained at some $P^\star$ with $f(P^\star) \ge \varepsilon^*$.
\end{proof}

\begin{corollary}[Explicit lower bound for polynomially many parameters]
\label{cor:explicit}
Suppose the parameter count grows at most polynomially in the number of
qubits, $N_p = O(\mathrm{poly}(n))$. Then $\alpha = N_p/(M-1) < 1$ for all
sufficiently large $n$, and for every such $n$ there exists a distribution in
$\Delta^{M-1}$ whose TV distance to every Born distribution in $\mathcal{B}$
is at least
\[
\varepsilon^* \;=\; \Big(\tfrac{1}{8}\Big)^{1/(1-\alpha)} \,
(\pi N_p + 1)^{-\alpha/(1-\alpha)}
\;\xrightarrow{\;n\to\infty\;}\; \tfrac{1}{8}.
\]
\end{corollary}

\begin{proof}
Since $M-1 = 2^{n}-1$ grows exponentially while $N_p$ grows polynomially,
$\alpha = N_p/(2^{n}-1) \to 0$; in particular $\alpha < 1$ for all
sufficiently large $n$, so Theorem~\ref{thm:generic} applies. For the limit,
$1/(1-\alpha) \to 1$ gives $(1/8)^{1/(1-\alpha)} \to 1/8$, while
$\tfrac{\alpha}{1-\alpha}\log(\pi N_p+1) \to 0$---the logarithm grows
logarithmically in $n$ and $\alpha$ decays exponentially---gives
$(\pi N_p+1)^{-\alpha/(1-\alpha)} \to 1$. Hence $\varepsilon^* \to 1/8$.
\end{proof}

\begin{lemma}[TV-ball volume ratio]\label{lem:ball_volume}
For any $Q \in \Delta^{M-1}$ and $\delta > 0$:
\begin{equation}\label{eq:ball_volume}
    \frac{Vol_{M-1}\bigl(
    B_\delta^{\TV}(Q)\cap\Delta^{M-1}\bigr)}
    {Vol_{M-1}(\Delta^{M-1})}
    \leq (4\delta)^{M-1}.
\end{equation}
When $4\delta \ge 1$ the bound is trivial (the ratio never exceeds one), so no restriction on $\delta$ is needed.
\end{lemma}

\begin{proof}
We work in the reduced coordinate system
$y = (P_1,\ldots,P_{M-1}) \in \Real^{M-1}$ with
$P_M = 1 - \sum_{i=1}^{M-1} P_i$.
In these coordinates, the simplex is
$\Delta = \{y \in \Real^{M-1} : y_i \geq 0,\;
\sum_i y_i \leq 1\}$, with volume
$Vol(\Delta) = 1/(M-1)!$.

The $\ell^1$ distance between two distributions
$P, Q \in \Delta^{M-1}$ decomposes as:
\begin{equation}\label{eq:l1_decomp}
    \norm{P-Q}_1
    = \sum_{i=1}^{M-1} |y_i - q_i|
    + \Bigl|\sum_{i=1}^{M-1}(q_i - y_i)\Bigr|.
\end{equation}
Since $|\sum_i(q_i - y_i)| \leq \sum_i |y_i - q_i|$
(triangle inequality), the two terms satisfy:
\begin{equation}\label{eq:l1_bounds}
    \sum_{i=1}^{M-1}|y_i - q_i|
    \;\leq\; \norm{P-Q}_1
    \;\leq\; 2\sum_{i=1}^{M-1}|y_i - q_i|.
\end{equation}
The left inequality shows that
$\{P : \norm{P-Q}_1 \leq 2\delta\}
\subseteq \{y : \sum_i |y_i - q_i| \leq 2\delta\}$,
i.e., every TV-ball is contained in a cross-polytope
in the reduced coordinates:
\begin{equation}\label{eq:tv_in_cross}
    B_\delta^{\TV}(Q) \cap \Delta
    \;\subseteq\;
    \bigl\{y \in \Real^{M-1} :
    \norm{y - q}_1 \leq 2\delta\bigr\}.
\end{equation}

The cross-polytope
$\{y \in \Real^{M-1} : \norm{y-q}_1 \leq r\}$
has $(M{-}1)$-dimensional volume $(2r)^{M-1}/(M-1)!$
(the standard formula for the $\ell^1$-ball in
$\Real^{M-1}$). Substituting $r = 2\delta$:
\[
    Vol\bigl(\{y : \norm{y-q}_1 \leq 2\delta\}\bigr)
    = \frac{(4\delta)^{M-1}}{(M-1)!}.
\]
Since $B_\delta^{\TV}(Q) \cap \Delta$ is contained in
this cross-polytope by~\eqref{eq:tv_in_cross}:
\[
    \frac{Vol(B_\delta^{\TV}(Q) \cap \Delta)}
    {Vol(\Delta)}
    \leq \frac{(4\delta)^{M-1}/(M-1)!}{1/(M-1)!}
    = (4\delta)^{M-1}.\qedhere
\]
\end{proof}

\begin{theorem}[Most distributions are hard]
\label{thm:levy}
Let $P$ be drawn uniformly from $\Delta^{M-1}$
(i.e., $P \sim \mathrm{Dir}(\mathbf{1}_M)$). Then for any
$\varepsilon > 0$:
\begin{equation}\label{eq:levy_bound}
    \mathbb{P}\Bigl[\inf_{\bm{\theta}}
    \TV(P,P_{\bm{\theta}}^{\mathrm{Born}})
    \leq \varepsilon\Bigr]
    \leq \left(\frac{\pi N_p}{\varepsilon}
    + 1\right)^{N_p}
    \cdot\,(8\varepsilon)^{M-1}.
\end{equation}
In particular, the right-hand side decays exponentially
in $M$ for $\varepsilon < \varepsilon^*$ where
$\varepsilon^*$ is from Theorem~\ref{thm:generic}.
\end{theorem}

\begin{proof}
Define $f: \Delta^{M-1} \to [0,1]$ by
$f(P) = \inf_{\bm{\theta}}\TV(P,P_{\bm{\theta}}^{\mathrm{Born}})$
and the nonempty subset
$\mathcal{B}^\varepsilon = \{P \in \Delta^{M-1} :
f(P) \leq \varepsilon\}$
(the $\varepsilon$-neighborhood of $\mathcal{B}$ in the TV metric).
Since $P \sim \mathrm{Dir}(\mathbf{1}_M)$ is the uniform
distribution on $\Delta^{M-1}$, the probability equals
the volume ratio:
\begin{equation}\label{eq:prob_vol}
    \mathbb{P}[f(P) \leq \varepsilon]
    = \frac{Vol_{M-1}(\mathcal{B}^\varepsilon)}
    {Vol_{M-1}(\Delta^{M-1})}.
\end{equation}
Let $\{b_1,\ldots,b_K\}$ be an $\varepsilon$-cover of $\mathcal{B}$
in the TV metric, with
$K = \mathcal{N}(\mathcal{B},\varepsilon,\TV)$.
For any $P \in \mathcal{B}^\varepsilon$, there exists $Q \in \mathcal{B}$
with $\TV(P,Q) \leq \varepsilon$, and there exists $b_j$
with $\TV(Q,b_j) \leq \varepsilon$. By the triangle inequality:
$\TV(P,b_j) \leq 2\varepsilon$.
Therefore:
\begin{equation}\label{eq:neighborhood_cover}
    \mathcal{B}^\varepsilon \subseteq
    \bigcup_{j=1}^{K} B_{2\varepsilon}^{\TV}(b_j).
\end{equation}
By the subadditivity of volume
and Lemma~\ref{lem:ball_volume} with $\delta = 2\varepsilon$:
\begin{align}
    Vol_{M-1}(\mathcal{B}^\varepsilon)
    &\leq \sum_{j=1}^{K}
    Vol_{M-1}\bigl(B_{2\varepsilon}^{\TV}(b_j)
    \cap \Delta^{M-1}\bigr)
    \notag\\
    &\leq K \cdot (4 \cdot 2\varepsilon)^{M-1}
    \cdot Vol_{M-1}(\Delta^{M-1})
    \notag\\
    &= K \cdot (8\varepsilon)^{M-1}
    \cdot Vol_{M-1}(\Delta^{M-1}).
    \label{eq:vol_ub}
\end{align}
Dividing~\eqref{eq:vol_ub} by $Vol_{M-1}(\Delta^{M-1})$
and substituting Proposition~\ref{prop:covering_B}:
\begin{equation}\label{eq:prob_bound}
    \mathbb{P}[f(P) \leq \varepsilon]
    \leq K \cdot (8\varepsilon)^{M-1}
    \leq \left(\frac{\pi N_p}{\varepsilon}
    + 1\right)^{N_p}
    \cdot\,(8\varepsilon)^{M-1}.
\end{equation}
Taking logarithms:
\begin{equation}\label{eq:log_prob}
    \log\mathbb{P}[f(P) \leq \varepsilon]
    \leq N_p\log\!\left(\frac{\pi N_p}{\varepsilon}
    + 1\right)
    - (M{-}1)\log\frac{1}{8\varepsilon}.
\end{equation}
The first term grows as $N_p\log(1/\varepsilon)$ and the
second decreases as $(M{-}1)\log(1/\varepsilon)$. Since
$N_p = O(n^2)$ and $M - 1 = 2^n - 1$, the second term
dominates for all sufficiently small $\varepsilon$, and the
exponent is negative whenever $\varepsilon < \varepsilon^*$
from Theorem~\ref{thm:generic}.
\end{proof}

\section{Measure-Theoretic Consistency of the Discretization}
\label{sec:consistency}

\providecommand{\R}{\mathbb{R}}
\providecommand{\Z}{\mathbb{Z}}
\providecommand{\cB}{\mathcal{B}}
\providecommand{\cP}{\mathcal{P}}
\providecommand{\cI}{\mathcal{I}}
\providecommand{\cO}{\mathcal{O}}
\providecommand{\cW}{\mathcal{W}}
\providecommand{\diam}{\operatorname{diam}}
\providecommand{\dist}{\operatorname{dist}}
\providecommand{\ind}{\mathbf{1}}

\subsection{Setup}\label{sec:measure-setup}

Let $D\geq 2$ be the number of simplex components and $\Delta^{D-1}$
the standard $(D{-}1)$-simplex,
\begin{equation}\label{eq:simplex}
  \Delta^{D-1} = \Bigl\{ x\in\R^{D} : x_i \geq 0,\;
  \textstyle\sum_{i=1}^{D} x_i = 1 \Bigr\},
\end{equation}
equipped with the $(D{-}1)$-dimensional Lebesgue measure
$\lambda_{D-1}$ (the surface measure on the affine hyperplane
$\sum_i x_i = 1$) and the Borel $\sigma$-algebra
$\cB(\Delta^{D-1})$.

Let $\mu$ be a probability measure on
$(\Delta^{D-1},\cB(\Delta^{D-1}))$ with $\mu\ll\lambda_{D-1}$.
By the Radon--Nikodym theorem there is a density
$f = d\mu/d\lambda_{D-1} \geq 0$ with $f\in L^1(\Delta^{D-1})$ and
$\int f\,d\lambda_{D-1} = 1$, so that $\mu(A)=\int_A f\,d\lambda_{D-1}$
for every $A\in\cB(\Delta^{D-1})$. Here $f$ is the target density
of~\eqref{eq:target_born_correspondence}, so that
$P^{\mathrm{target}}(x) = \mu\bigl(\mathrm{cell}(\phi^{-1}(x))\bigr)$.

\subsection{Admissible Convex Partitions}\label{sec:partition}

The convergence theory of Section~\ref{sec:convergence} applies to any
partition satisfying the following axioms.

\subsubsection{Abstract framework}\label{subsec:abstract}

\begin{definition}[Admissible convex partition]\label{def:admissible}
A family $\cP = \{W_1,\ldots,W_M\}$ of subsets of $\Delta^{D-1}$ is an
\emph{admissible convex partition} at resolution $h>0$ if:
\begin{enumerate}[label=(A\arabic*)]
  \item\label{ax:convex}
    Each $W_j$ is a convex compact set with $\lambda_{D-1}(W_j)>0$.
  \item\label{ax:cover}
    $\Delta^{D-1} = \bigcup_{j=1}^M W_j$ and
    $\lambda_{D-1}(W_i\cap W_j)=0$ for $i\neq j$.
  \item\label{ax:diam}
    There is a constant $C>0$, independent of $h$, with
    $\diam(W_j)\leq Ch$ for all $j$.
  \item\label{ax:ecc}
    There is a constant $\alpha>0$, independent of $h$, such that every
    $W_j$ contains a ball of radius $\alpha\,\diam(W_j)$
    (uniform non-degeneracy / bounded eccentricity; balls are taken
    within the affine hull of $\Delta^{D-1}$).
\end{enumerate}
\end{definition}

\begin{remark}
Axiom~\ref{ax:ecc} is the Vitali regularity needed for the Lebesgue
differentiation theorem, i.e.\ for the \emph{pointwise} statement
Theorem~\ref{thm:L1}(i); it rules out cells that degenerate into arbitrarily
thin slivers as $h\to0$. Everything else---the inner--outer
approximation, the $L^1$/TV convergence Theorem~\ref{thm:L1}(ii)--(iii), the
identifiability theorem, and the error decomposition---uses only
\ref{ax:convex}--\ref{ax:diam}. Moreover, convexity in
\ref{ax:convex} enters only through connectedness (in the proof of
Theorem~\ref{thm:inner-outer}(b)), so the framework applies verbatim to
partitions with connected, possibly non-convex cells.
\end{remark}

All results in Section~\ref{sec:convergence} hold for any sequence of
admissible convex partitions $\cP_{h_n}$ with $h_n\downarrow 0$.
Neither the lattice origin of the cells nor any tessellation rule is
used beyond axioms \ref{ax:convex}--\ref{ax:ecc}.

\subsubsection{Concrete realisation: level-set polytope partition}
\label{subsec:concrete}

The augmented grid $\mathcal{Q}_N$ of
Definition~\ref{def:augmented_grid} induces partitions realising
these axioms. A natural first choice is the lattice Voronoi partition:
place the rescaled lattice points
$\tfrac{1}{N}\mathcal{X}_N\subset\Delta^{D-1}$ and assign to each
its nearest-neighbour (Euclidean) cell. Near $\partial\Delta^{D-1}$,
however, the Euclidean Voronoi cells are truncated irregularly by the
simplex boundary, complicating the uniform \emph{eccentricity}
control (the diameter bound survives; cf.\
Remark~\ref{rem:voronoi-applicable} below). We therefore first
develop the \emph{level-set polytope partition}, which satisfies all
four axioms with clean constants.

Each unit hypercube $\bm{z}+[0,1]^{D-1}$
($\bm{z}\in\Z_{\geq 0}^{D-1}$, in the affine chart of $\Delta^{D-1}$)
is sliced by the $D-2$ level-set hyperplanes
$\{\ell(x)=\ell(\bm{z})+m\}_{m=1}^{D-2}$, where $\ell(x)=\sum_i x_i$.
This produces $D-1$ convex \emph{type-$m$ polytopes}
\begin{equation}\label{eq:type-m}
  P_m(\bm{z}) = \bigl\{x\in \bm{z}+[0,1]^{D-1} :
  \ell(\bm{z})+m-1 \leq \ell(x) \leq \ell(\bm{z})+m\bigr\},
  \quad m=1,\ldots,D-1,
\end{equation}
which tile the hypercube,
$\bm{z}+[0,1]^{D-1} = \bigcup_{m=1}^{D-1} P_m(\bm{z})$, with adjacent
slabs sharing a $(D{-}2)$-dimensional face of measure zero. The
distinguished representative of $P_m(\bm{z})$ is its vertex centroid
$G_m(\bm{z}) = \bm{z}+\frac{m}{D}\mathbf{1}$ (the geometry and the
$S_D$-equivariance of $G_m$ are developed in
Lemma~\ref{lem:intrinsic} below; only convexity, the diameter bound,
and the eccentricity bound below are used here).

\begin{definition}[Augmented partition]\label{def:augmented}
At resolution $h=1/N$, the \emph{augmented cell partition} is
\begin{equation}\label{eq:Qpartition}
  \cW_N = \bigl\{P_m(\bm{z})\cap\Delta^{D-1} :
  \bm{z}\in\Z_{\geq 0}^{D-1},\;
  m=1,\ldots,D-1,\;
  \lambda_{D-1}\bigl(P_m(\bm{z})\cap\Delta^{D-1}\bigr)>0\bigr\},
\end{equation}
rescaled by $1/N$ onto $\Delta^{D-1}$.
\end{definition}

\begin{proposition}\label{prop:augmented-admissible}
The augmented partition $\cW_N$ is an admissible convex partition
(Definition~\ref{def:admissible}) with resolution $h=1/N$.
\end{proposition}

\begin{proof}
\ref{ax:convex}: Each $P_m(\bm{z})$ is a hypercube (convex) intersected
with two half-spaces (convex), hence convex; intersection with the
convex set $\Delta^{D-1}$ keeps it convex and compact. The condition
$\lambda_{D-1}(P_m(\bm{z})\cap\Delta^{D-1})>0$ in
\eqref{eq:Qpartition} discards degenerate boundary pieces, so
\ref{ax:convex} holds.

\ref{ax:cover}: The hypercubes tile the orthant and, within each, the
type-$m$ slabs tile the hypercube, so
$\bigcup_{m}(P_m(\bm{z})\cap\Delta^{D-1})=\Delta^{D-1}$; adjacent cells
overlap only in $(D{-}2)$-faces, which are $\lambda_{D-1}$-null.

\ref{ax:diam}: $P_m(\bm{z})\subseteq\bm{z}+[0,1]^{D-1}$ gives
$\diam(P_m(\bm{z}))\leq\sqrt{D-1}$; after the $1/N$ rescaling,
$\diam\leq\sqrt{D-1}/N$, i.e.\ \ref{ax:diam} with $C=\sqrt{D-1}$.

\ref{ax:ecc}: This is Lemma~\ref{lem:eccentricity} below.
\end{proof}

\begin{lemma}[Uniform eccentricity]\label{lem:eccentricity}
There is a constant $\alpha_D>0$, depending only on $D$, such that every
cell $W\in\cW_N$ contains a ball of radius $\alpha_D\,\diam(W)$, for all
$N$.
\end{lemma}

\begin{proof}
There are only finitely many prototype shapes: the type-$m$ polytopes
$P_m(\mathbf{0})$ for $m=1,\ldots,D-1$. Each is a full-dimensional
convex polytope, so it contains a ball of some radius $\rho_m>0$ and has
diameter $\delta_m>0$; set
$\beta = \min_{1\leq m\leq D-1}\rho_m/\delta_m>0$. Every interior cell
of $\cW_N$ is an integer translate of some $P_m(\mathbf{0})$ rescaled by
$1/N$; translation and uniform scaling preserve the ratio
$(\text{inradius})/\diam$, so each such cell contains a ball of radius
$\beta\,\diam$.

In fact no cell of $\cW_N$ is truncated by $\partial\Delta^{D-1}$.
In the affine chart, the facet constraints $y_i \geq 0$ hold
automatically on every cube $\bm{z}+[0,1]^{D-1}$ with
$\bm{z}\in\Z_{\geq 0}^{D-1}$, and the remaining facet
$\{\ell(y)=N\}$ coincides with one of the slicing hyperplanes, since
$N$ is an integer. Consequently each cell of positive measure in
\eqref{eq:Qpartition} is a \emph{complete} type-$m$ polytope---an
integer translate of one of the $D-1$ prototypes rescaled by
$1/N$---and cells beyond the level $\ell=N$ intersect the simplex in a
$\lambda_{D-1}$-null face and are discarded. Hence $\alpha_D=\beta$
works for all $N$.
\end{proof}

\begin{lemma}[Intrinsic description and $S_D$-equivariance]
\label{lem:intrinsic}
In the full coordinates $u = Nx$ (so $\sum_{i=1}^{D} u_i = N$,
$u_i \geq 0$), the cells of $\cW_N$ are exactly
\begin{equation}\label{eq:intrinsic-cell}
  C(\bm{k}) \;=\; \bigl(\bm{k}+[0,1]^{D}\bigr)
  \cap \Bigl\{u : \textstyle\sum_{i=1}^{D} u_i = N\Bigr\},
  \qquad \bm{k}\in\Z_{\geq 0}^{D},\;
  m \coloneqq N - \textstyle\sum_i k_i \in \{1,\ldots,D-1\},
\end{equation}
and the chart cell $P_m(\bm{z})\cap\Delta^{D-1}$ with
$\bm{z}=(k_1,\ldots,k_{D-1})$ is the image of $C(\bm{k})$ under the
chart projection. Consequently, for every $\sigma\in S_D$,
\[
  \sigma\cdot C(\bm{k}) = C(\sigma\cdot\bm{k}),
\]
i.e.\ the partition is $S_D$-invariant with
$\mathrm{cell}(\sigma\cdot x)=\sigma\cdot\mathrm{cell}(x)$ up to
$\lambda_{D-1}$-null sets, and the representative centroids transform
equivariantly: $G_m(\sigma\cdot\bm{k}) = \sigma\cdot G_m(\bm{k})$.
\end{lemma}

\begin{proof}
For $u$ with $\sum u_i = N$, set $k_i=\lfloor u_i\rfloor$ and
$t=u-\bm{k}\in[0,1)^D$. Then $\sum_i t_i = N-\sum_i k_i$ is an
integer in $\{0,\ldots,D-1\}$, and the value $0$ occurs only on the
$\lambda_{D-1}$-null set $u\in\Z^D$; so a.e.\ point determines a
unique pair $(\bm{k},m)$ with $m\in\{1,\ldots,D-1\}$ and lies in
$C(\bm{k})$. Each $C(\bm{k})$ is convex and compact, and $k_i\geq0$
ensures $C(\bm{k})\subseteq N\Delta^{D-1}$. For the chart
correspondence, dropping the $D$-th coordinate maps $C(\bm{k})$
affinely onto
$\{y\in\bm{z}+[0,1]^{D-1} : \ell(\bm{z})+m-1\leq\ell(y)\leq
\ell(\bm{z})+m\}$, since $u_D\in[k_D,k_D+1]$ is equivalent to
$\ell(y)=N-u_D\in[N-k_D-1,\,N-k_D]$ and
$\ell(\bm{z})+m = N-k_D$. Equivariance is immediate from
$\lfloor(\sigma u)_i\rfloor = k_{\sigma^{-1}(i)}$, and the centroid
identity is the $S_D$-equivariance of
$G_m(\bm{k})=\bm{k}+\frac{m}{D}\mathbf{1}$.
\end{proof}

\begin{remark}[Chart invariance]\label{rem:chart}
The chart projection is an affine bijection whose condition number
depends only on $D$; hence axioms
\ref{ax:convex}--\ref{ax:ecc} hold in the chart iff they hold on the
hyperplane $\{\sum_i x_i=1\}$ (with $D$-dependent constants), and the
chart Lebesgue measure differs from the surface measure
$\lambda_{D-1}$ by the constant Jacobian $\sqrt{D}$. All statements
below are therefore chart-invariant.
\end{remark}

\begin{remark}[Achieving $2^n$ cells]\label{rem:2n}
The cell count of the full partition is fixed by $(N,D)$, whereas the
encoding requires exactly $2^n$ cells. Three mechanisms reconcile the
two, and the theory covers each. \emph{(a) Merging:} unselected cells
are merged with adjacent selected ones; merged cells are connected but
non-convex, with diameter inflated by at most a $D$-dependent factor
and inradius no smaller than that of any constituent, so
\ref{ax:diam}--\ref{ax:ecc} persist and convexity is not needed
(cf.\ the remark after Definition~\ref{def:admissible}).
\emph{(b) Voronoi of the selected representatives:} since the
$L^1$/TV theory of Theorem~\ref{thm:L1}(ii)--(iii) requires neither convexity
nor the regularity \ref{ax:ecc}, boundary-truncated Voronoi cells are
admissible for all distributional conclusions. \emph{(c) Masking:}
retain all cells and accept
$n'=\lceil\log_2 |\cW_N|\rceil$ with masked states. The arguments
below depend only on \ref{ax:cover}--\ref{ax:diam} (plus
\ref{ax:ecc} where stated), not on the particular mechanism.
\end{remark}

We now make the main text's point--cell correspondence explicit.

\begin{definition}[Voronoi partition of the simplex]\label{def:voronoi}
Let $\mathcal{Q}_N\subset\Delta^{D-1}$ be the augmented grid, whose
elements we call \emph{sites}, and equip $\Delta^{D-1}$ with the
Euclidean metric inherited from $\mathbb{R}^{D}$. For a site
$\mathbf{x}\in\mathcal{Q}_N$, the associated \emph{Voronoi cell} is the
set of points of the simplex no farther from $\mathbf{x}$ than from any
other site,
\begin{equation}\label{eq:voronoi_cell}
    V(\mathbf{x})
    \;=\;
    \bigl\{\, \mathbf{s}\in\Delta^{D-1}\;:\;
    \lVert \mathbf{s}-\mathbf{x}\rVert_2
    \le
    \lVert \mathbf{s}-\mathbf{y}\rVert_2
    \quad\text{for all } \mathbf{y}\in\mathcal{Q}_N \,\bigr\}.
\end{equation}
\end{definition}

Expanding the squared norms in \eqref{eq:voronoi_cell}, the condition
$\lVert\mathbf{s}-\mathbf{x}\rVert_2\le\lVert\mathbf{s}-\mathbf{y}\rVert_2$
is equivalent to the affine inequality
$\langle \mathbf{s},\,\mathbf{y}-\mathbf{x}\rangle
\le \tfrac12\bigl(\lVert\mathbf{y}\rVert_2^{2}
-\lVert\mathbf{x}\rVert_2^{2}\bigr)$,
so each cell is the intersection of $\Delta^{D-1}$ with finitely many
closed half-spaces.

\begin{lemma}[Admissibility of the Voronoi partition]
\label{lem:voronoi_props}
The cells \eqref{eq:voronoi_cell} satisfy:
\begin{enumerate}\itemsep2pt
\item[(i)] \emph{(Convexity and compactness.)} Each $V(\mathbf{x})$ is
a convex polytope, being the intersection of the compact convex set
$\Delta^{D-1}$ with $|\mathcal{Q}_N|-1$ closed half-spaces.
\item[(ii)] \emph{(Covering.)}
$\bigcup_{\mathbf{x}\in\mathcal{Q}_N}V(\mathbf{x})=\Delta^{D-1}$, since
every $\mathbf{s}\in\Delta^{D-1}$ lies in $V(\mathbf{x})$ for any
nearest site $\mathbf{x}$, which exists because $\mathcal{Q}_N$ is
finite and non-empty.
\item[(iii)] \emph{(Essential disjointness and positive measure.)} For
$\mathbf{x}\neq\mathbf{y}$ the intersection
$V(\mathbf{x})\cap V(\mathbf{y})$ lies in the bisector hyperplane
$\{\mathbf{s}:\langle\mathbf{s},\mathbf{y}-\mathbf{x}\rangle
=\tfrac12(\lVert\mathbf{y}\rVert_2^2-\lVert\mathbf{x}\rVert_2^2)\}$.
Since $\mathbf{x},\mathbf{y}\in\Delta^{D-1}$ implies
$\langle\mathbf{1},\mathbf{y}-\mathbf{x}\rangle=0$ with
$\mathbf{y}-\mathbf{x}\neq 0$, the bisector is transversal to the
affine hull $\{\sum_i s_i=1\}$ and meets it in a $(D-2)$-dimensional
affine subspace, which is $\lambda_{D-1}$-null. Moreover
$\lambda_{D-1}\bigl(V(\mathbf{x})\bigr)>0$, since
$B(\mathbf{x},\rho)\cap\Delta^{D-1}\subseteq V(\mathbf{x})$ for
$\rho<\tfrac12\min_{\mathbf{y}\neq\mathbf{x}}
\lVert\mathbf{x}-\mathbf{y}\rVert_2$.
\item[(iv)] \emph{(Mesh size.)}
$\operatorname{diam}V(\mathbf{x})\le 2\sqrt{D}/N$ for every site, so
all cells have diameter $O(1/N)$ at fixed $D$.
\end{enumerate}
Consequently the family $\{V(\mathbf{x})\}$ satisfies
\ref{ax:convex}--\ref{ax:diam} of Definition~\ref{def:admissible} with
$C=2\sqrt{D}$ and $h=1/N$---all that the $L^1$/TV theory of
Theorem~\ref{thm:L1}(ii)--(iii) requires.
\end{lemma}

\begin{proof}[Proof of (iv)]
Since $\mathcal{Q}_N\supseteq\tfrac1N\mathcal{X}_N$, it suffices to
bound the covering radius of the natural lattice. Given
$\mathbf{s}\in\Delta^{D-1}$, set $k_i=\lfloor N s_i\rfloor$ and
distribute the deficit $N-\sum_i k_i\in\{0,\dots,D-1\}$ among the
coordinates with the largest fractional parts, one unit each; the
resulting $\mathbf{k}$ lies in $\mathcal{X}_N$ and satisfies
$|s_i-k_i/N|\le 1/N$ coordinatewise, hence
$\lVert\mathbf{s}-\mathbf{k}/N\rVert_2\le\sqrt{D}/N$. Every point of
the simplex is therefore within $\sqrt{D}/N$ of a site; a Voronoi cell
is contained in the ball of the covering radius about its site, giving
the stated diameter bound. Augmenting the site set with centroids
(Definition~\ref{def:augmented_grid}) can only decrease the covering
radius, and the containment applies to every site of the augmented
set.
\end{proof}

By (ii) and (iii) the cells cover the simplex and overlap only on a
$\lambda_{D-1}$-null set, so they form a partition up to null sets; a
genuine partition is obtained by assigning each tie to a single site
under any fixed measurable rule, and all statements below are
insensitive to this choice. We write
$\mathrm{cell}(\mathbf{x})=V(\mathbf{x})$, matching the correspondence
of Eq.~\eqref{eq:target_born_correspondence}. The partition also
inherits the symmetry of the grid: coordinate permutations are
Euclidean isometries preserving $\Delta^{D-1}$, and $\mathcal{Q}_N$ is
closed under them (Definition~\ref{def:augmented_grid}), so
$\mathrm{cell}(\sigma\cdot\mathbf{x})=\sigma\cdot\mathrm{cell}(\mathbf{x})$
up to null sets (Lemma~\ref{lem:voronoi-equiv}); since Lebesgue measure
is isometry-invariant, all cells within a common orbit have equal
volume. This is what makes the uniform spreading in the orbit map the
measure-theoretically correct lift: a density that is constant on an
orbit induces equal mass on each of its cells, so distributing an orbit
weight $q_j$ evenly over its $s_j$ members reproduces the underlying
measure exactly rather than merely approximately
(Theorem~\ref{thm:orbit_symmetric}). The same argument applies
verbatim to any subgroup $G\le S_D$.

\begin{remark}[The Voronoi partition of the augmented grid]
\label{rem:voronoi-applicable}
The main-text discretization is the Euclidean Voronoi partition of the
augmented grid $\mathcal{Q}_N$ (Definition~\ref{def:augmented_grid};
cf.\ Lemma~\ref{lem:voronoi-equiv}), whose cell count equals
$|\mathcal{Q}_N| = 2^n$ by construction
(Definition~\ref{def:encoding-constraints}). These Voronoi cells
satisfy \ref{ax:convex}--\ref{ax:diam} with $C=2\sqrt{D}$
(Definition~\ref{def:voronoi}, Lemma~\ref{lem:voronoi_props}); cells
truncated by $\partial\Delta^{D-1}$ may, however, fail the uniform
eccentricity \ref{ax:ecc}. Since
Theorem~\ref{thm:L1}(ii)--(iii), Theorem~\ref{thm:uniqueness}, and
Theorem~\ref{thm:qcbm} use only \ref{ax:convex}--\ref{ax:diam}, all
distributional conclusions of this section apply verbatim to the
Voronoi discretization of the main text; the level-set partition
$\cW_N$ additionally satisfies \ref{ax:ecc} and therefore also yields
the pointwise statement of Theorem~\ref{thm:L1}(i).
\end{remark}

\subsection{QCBM Histogram and the Bijection}\label{sec:qcbm}

\subsubsection{Bijection from qubits to partition cells}
\label{subsec:bijection}

A QCBM on $n$ qubits produces
\begin{equation}\label{eq:qcbm-output}
  q_k(\bm\theta) = |\langle k|U(\bm\theta)|0\rangle|^2,
  \quad k=0,\ldots,2^n-1,
  \quad \textstyle\sum_{k} q_k = 1.
\end{equation}
When the cell count is $2^n$ (Remark~\ref{rem:2n}) we relabel the
cells by basis indices, writing the partition as
$\{W_0,\ldots,W_{2^n-1}\}$, so that basis state $|k\rangle$
corresponds to exactly one cell $W_k$---no surplus states, no penalty
term. For the Voronoi realisation of
Remark~\ref{rem:voronoi-applicable}, this labelling is the one induced
by the encoding bijection $\phi$ of
Definition~\ref{def:encoding-constraints}, namely
$W_k = \mathrm{cell}(\phi^{-1}(k))$, matching the correspondence
in~\eqref{eq:target_born_correspondence}. The objective becomes
\begin{equation}\label{eq:new-loss}
  \mathcal{L}(\bm\theta)
  = \mathsf{D}\!\bigl(\{q_k(\bm\theta)\},\{\mu(W_k)\}\bigr),
\end{equation}
with $\mathsf{D}$ any divergence (MMD, KL, total variation).

\subsubsection{Histograms as simple functions}\label{subsec:simple}

The QCBM output and the target induce the piecewise-constant densities
\begin{equation}\label{eq:qcbm-density}
  f_h^{\mathrm{QCBM}}(x)
  = \sum_{k}
    \frac{q_k(\bm\theta)}{\lambda_{D-1}(W_k)}\,
    \ind_{W_k}(x),
  \qquad
  f_h(x)
  = \sum_{k}
    \frac{\mu(W_k)}{\lambda_{D-1}(W_k)}\,
    \ind_{W_k}(x),
\end{equation}
where $f_h=\mathbb{E}[f\mid\cW_N]$ is the conditional expectation of $f$
on the partition $\sigma$-algebra, with associated measures
$\tilde\mu_h = f_h\,d\lambda_{D-1}$ and
$\tilde\mu_h^{\mathrm{QCBM}} = f_h^{\mathrm{QCBM}}\,d\lambda_{D-1}$.
Both densities are simple functions; when the
QCBM matches the histogram ($q_k=\mu(W_k)$ for all $k$),
$f_h^{\mathrm{QCBM}}=f_h$.

\subsection{Core Convergence Theorems}\label{sec:convergence}

Throughout, $\{W_j^{(h)}\}_{j=1}^{M(h)}$ is a sequence of admissible
convex partitions (Definition~\ref{def:admissible}) with $h=h_n\downarrow 0$; in
particular $\cW_N$ with $h=1/N$ qualifies.

\subsubsection{Inner--outer measure approximation}\label{subsec:inner-outer}

\begin{definition}\label{def:three-classes}
For a Borel set $A\subseteq\Delta^{D-1}$ and resolution $h$, split
$\{1,\ldots,M(h)\}$ into interior
$\cI_h(A)=\{j:W_j^{(h)}\subseteq A\}$, exterior
$\cO_h(A)=\{j:W_j^{(h)}\cap A=\emptyset\}$, and boundary
$\cB_h(A)=\{j:W_j^{(h)}\cap A\neq\emptyset,\ W_j^{(h)}\not\subseteq A\}$
cells, with inner/outer approximations
$\mu_h^{-}(A)=\sum_{j\in\cI_h}\mu(W_j^{(h)})$ and
$\mu_h^{+}(A)=\sum_{j\in\cI_h\cup\cB_h}\mu(W_j^{(h)})$.
\end{definition}

\begin{theorem}[Inner--outer convergence]\label{thm:inner-outer}
\leavevmode
\begin{enumerate}[label=(\alph*)]
  \item $\mu_h^{-}(A)\leq\mu(A)\leq\mu_h^{+}(A)$ for every Borel $A$.
  \item If $\mu(\partial A)=0$, then $\mu_h^{\pm}(A)\to\mu(A)$ as
    $h\to0$.
\end{enumerate}
(Uses only \ref{ax:convex}--\ref{ax:diam}.)
\end{theorem}

\begin{proof}
(a) $A_h^{-}=\bigcup_{j\in\cI_h}W_j^{(h)}\subseteq A$ gives
$\mu_h^{-}(A)=\mu(A_h^{-})\leq\mu(A)$ (cell overlaps are
$\lambda_{D-1}$-null, hence $\mu$-null by $\mu\ll\lambda_{D-1}$, so
the sum equals the measure of the union); every $x\in A$ lies in some cell
meeting $A$, so $A\subseteq A_h^{+}:=\bigcup_{j\in\cI_h\cup\cB_h}W_j^{(h)}$
and $\mu(A)\leq\mu_h^{+}(A)$.

(b) Each boundary cell meets $A$ and $A^c$; being convex (connected) it
meets $\partial A$, so it lies in the tube
$(\partial A)_{Ch}=\{x:\dist(x,\partial A)\leq Ch\}$ (by \ref{ax:diam}).
Hence
$0\leq\mu_h^{+}(A)-\mu_h^{-}(A)\leq\mu((\partial A)_{Ch})$. As
$h\to0$, $(\partial A)_{Ch}\downarrow\overline{\partial A}=\partial A$,
so by continuity of measure from above the right side tends to
$\mu(\partial A)=0$. The squeeze theorem gives (b).
\end{proof}

\subsubsection{Simple-function approximation of densities}\label{subsec:L1}

\begin{theorem}[$L^1$ convergence]\label{thm:L1}
Let $\mu=f\,d\lambda_{D-1}$, $f\geq 0$, $\int f=1$, and let
$h_n\downarrow0$. With $f_{h_n}=\mathbb{E}[f\mid\cP_{h_n}]$ as in
\eqref{eq:qcbm-density},
\begin{enumerate}[label=(\roman*)]
  \item under \ref{ax:convex}--\ref{ax:ecc},
    $f_{h_n}\to f$ $\lambda_{D-1}$-a.e.;
  \item under \ref{ax:convex}--\ref{ax:diam} alone,
    $\|f_{h_n}-f\|_{L^1}\to0$;
  \item $d_{\TV}(\tilde\mu_{h_n},\mu)=\tfrac12\|f_{h_n}-f\|_{L^1}\to0$,
    where $\tilde\mu_{h_n}=f_{h_n}\,d\lambda_{D-1}$.
\end{enumerate}
Statement (ii) holds for every $g\in L^1(\Delta^{D-1})$ in place of
$f$, with no positivity or normalisation assumption.
\end{theorem}

\begin{proof}
\emph{Step 1 (integral preservation).} By construction
$\int_{W_j}f_{h_n}=\int_{W_j}f$, so $\int f_{h_n}=1$.

\emph{Step 2 ($L^1$ convergence; only
\ref{ax:convex}--\ref{ax:diam}).}
The map $g\mapsto\mathbb{E}[g\mid\cP_h]$ is linear and an
$L^1$-contraction: on each cell,
$|\mathbb{E}[g\mid\cP_h]|\leq\mathbb{E}[|g|\mid\cP_h]$ by the
triangle inequality for integrals, and integrating over the cell and
summing gives
$\|\mathbb{E}[g\mid\cP_h]\|_{L^1}\leq\|g\|_{L^1}$.
If $g$ is continuous (hence uniformly continuous on the compact
$\Delta^{D-1}$, with modulus $\omega_g$), then for a.e.\ $x$ in a
cell $W$,
\[
  \bigl|\mathbb{E}[g\mid\cP_h](x)-g(x)\bigr|
  = \Bigl|\frac{1}{\lambda_{D-1}(W)}\int_W (g-g(x))\Bigr|
  \leq \omega_g(\diam W) \leq \omega_g(Ch)
  \;\xrightarrow[h\to0]{}\;0
\]
uniformly, hence in $L^1$ ($\lambda_{D-1}(\Delta^{D-1})<\infty$).
For general $g\in L^1$ and $\varepsilon>0$, choose continuous
$g_\varepsilon$ with $\|g-g_\varepsilon\|_{L^1}<\varepsilon$
(density of $C(\Delta^{D-1})$ in $L^1$); then
\[
  \|\mathbb{E}[g\mid\cP_{h_n}]-g\|_{L^1}
  \leq 2\|g-g_\varepsilon\|_{L^1}
  + \|\mathbb{E}[g_\varepsilon\mid\cP_{h_n}]-g_\varepsilon\|_{L^1}
  \leq 2\varepsilon + o(1).
\]
Since $\varepsilon$ is arbitrary, (ii) follows---for signed $g$ as
well, proving the final claim of the statement. Note that neither
convexity nor \ref{ax:ecc} is used: only the null-overlap covering
\ref{ax:cover} and the diameter bound \ref{ax:diam}.

\emph{Step 3 (a.e.\ convergence; uses \ref{ax:ecc}).} Fix
$x\in\operatorname{int}(\Delta^{D-1})$ a Lebesgue point of $f$ (a.e.\
$x$ qualifies). Let $W^{(n)}(x)$ be the cell of $\cP_{h_n}$ containing
$x$. For large $n$, $W^{(n)}(x)$ is an interior cell, and
\[
  |f_{h_n}(x)-f(x)|
  = \Bigl|\frac{1}{\lambda_{D-1}(W^{(n)})}
    \int_{W^{(n)}}(f-f(x))\,d\lambda_{D-1}\Bigr|
  \leq \frac{1}{\lambda_{D-1}(W^{(n)})}
    \int_{W^{(n)}}|f-f(x)|\,d\lambda_{D-1}.
\]
By \ref{ax:ecc}, $W^{(n)}(x)$ contains a ball
$B(x',\alpha_D\diam W^{(n)})$ and is contained in $B(x,\diam W^{(n)})$,
so
$\lambda_{D-1}(W^{(n)})\geq\alpha_D^{\,D-1}\,\omega_{D-1}
(\diam W^{(n)})^{D-1}\geq c_D\,\lambda_{D-1}(B(x,\diam W^{(n)}))$
for a constant $c_D>0$, where $\omega_{D-1}$ denotes the
$\lambda_{D-1}$-volume of the unit ball of the affine hull. Hence
\[
  |f_{h_n}(x)-f(x)|
  \leq \frac{1}{c_D}\cdot
    \frac{1}{\lambda_{D-1}(B(x,\diam W^{(n)}))}
    \int_{B(x,\diam W^{(n)})}|f-f(x)|\,d\lambda_{D-1}.
\]
Since $\diam W^{(n)}\leq Ch_n\to0$, the right side $\to0$ at every
Lebesgue point of $f$ (Lebesgue differentiation theorem). Thus
$f_{h_n}\to f$ a.e.\ on the interior, and
$\lambda_{D-1}(\partial\Delta^{D-1})=0$.

\emph{Step 4 (total variation).} For absolutely continuous measures
$d_{\TV}(\tilde\mu_{h_n},\mu)=\tfrac12\|f_{h_n}-f\|_{L^1}$, giving
(iii).
\end{proof}

\subsubsection{Uniqueness via cell histograms}\label{subsec:uniqueness}

\begin{theorem}[Identifiability]\label{thm:uniqueness}
Let $\mu,\nu$ be Borel probability measures on $\Delta^{D-1}$ with
$\mu,\nu\ll\lambda_{D-1}$. If there is $h_n\downarrow0$ with
$\mu(W_j^{(h_n)})=\nu(W_j^{(h_n)})$ for all $n$ and all $j$, then
$\mu=\nu$.
\end{theorem}

\begin{proof}
Let $g=\frac{d\mu}{d\lambda_{D-1}}-\frac{d\nu}{d\lambda_{D-1}}\in
L^1(\Delta^{D-1})$. The hypothesis says $\int_{W}g\,d\lambda_{D-1}=0$
for every cell $W$ of every $\cP_{h_n}$, hence
$\mathbb{E}[g\mid\cP_{h_n}]=0$ identically. By the signed-$L^1$ form
of Theorem~\ref{thm:L1}(ii) (Step~2 of its proof, which uses neither
positivity nor normalisation),
$\mathbb{E}[g\mid\cP_{h_n}]\to g$ in $L^1$. The left side is $0$ for
all $n$, so $\|g\|_{L^1}=0$, i.e.\ $g=0$ a.e.\ and $\mu=\nu$.
\end{proof}

\begin{remark}\label{rem:ac-needed}
Absolute continuity is used: the conditional-expectation argument
converges to $g$ only in $L^1(\lambda_{D-1})$. For mutually singular
parts the cell histograms cannot separate $\mu$ from $\nu$ in general,
so $\mu,\nu\ll\lambda_{D-1}$ cannot be dropped (without it, the cell
masses may even depend on the null-boundary convention, since singular
parts can charge cell interfaces). The partitions need not be nested;
only \ref{ax:cover} and $\max_j\diam(W_j^{(h_n)})\to0$ are
required---\ref{ax:ecc} is \emph{not} needed here, by Step~2 of
Theorem~\ref{thm:L1}.
\end{remark}

\subsection{Main Error Decomposition}\label{sec:decomp}

\begin{theorem}[QCBM convergence]\label{thm:qcbm}
Let $N_n\to\infty$ with $|\cW_{N_n}|=2^{n_n}$ (the matching qubit
counts $n_n\to\infty$). Then
(writing $\tilde\mu_{N}:=\tilde\mu_{1/N}$ and
$\tilde\mu_{N}^{\mathrm{QCBM}}:=\tilde\mu_{1/N}^{\mathrm{QCBM}}$)
\begin{equation}\label{eq:error-split}
  d_{\TV}\bigl(\tilde\mu_{N_n}^{\mathrm{QCBM}},\mu\bigr)
  \;\leq\;
  \underbrace{d_{\TV}\bigl(\tilde\mu_{N_n}^{\mathrm{QCBM}},
    \tilde\mu_{N_n}\bigr)}_{\text{expressivity error}}
  \;+\;
  \underbrace{d_{\TV}\bigl(\tilde\mu_{N_n},\mu\bigr)}_{
    \text{discretization error}},
\end{equation}
and the discretization error tends to $0$ as $N_n\to\infty$. If, in
addition, the QCBM matches the histogram at each $N_n$
(expressivity error $=0$), then
$d_{\TV}(\tilde\mu_{N_n}^{\mathrm{QCBM}},\mu)\to0$.
\end{theorem}

\begin{proof}
\eqref{eq:error-split} is the triangle inequality for $d_{\TV}$. The
discretization term is
$d_{\TV}(\tilde\mu_{N_n},\mu)=\tfrac12\|f_{1/N_n}-f\|_{L^1}\to0$ by
Theorem~\ref{thm:L1}(ii)--(iii), since $\cW_{N_n}$ is admissible
(Proposition~\ref{prop:augmented-admissible}) with $h=1/N_n\to0$; note that only
\ref{ax:cover}--\ref{ax:diam} are needed for this step. Under exact
matching $f^{\mathrm{QCBM}}_{1/N_n}=f_{1/N_n}$, so the expressivity term
vanishes and the bound gives the claim.
\end{proof}

\begin{remark}\label{rem:two-errors}
The two errors have different origins. The discretization error is an
approximation-theoretic quantity, controlled unconditionally by
resolution $N_n\to\infty$ (Theorem~\ref{thm:L1}). The expressivity error
reflects how well the $n_n$-qubit circuit represents the target histogram
$\{\mu(W_k)\}$; it vanishes under exact matching and is
otherwise governed by the reachable set of the ansatz. Total
convergence requires both $N_n\to\infty$ \emph{and} an expressive (or
exactly matching) circuit. The expressivity error is precisely the
quantity analysed in the main text: for generic targets it is bounded
away from zero by the covering-number argument
(Theorem~\ref{thm:generic}), while for $S_D$-symmetric targets, Theorem~\ref{thm:orbit_symmetric} identifies
$\Delta_{\mathrm{sym}}$ with the orbit probability simplex, and
unrestricted universal orbit-space circuits have reachable
distributions dense in $\Delta_{\mathrm{sym}}$ (Corollary~\ref{cor:dominance}).
\end{remark}

\section{The Equivariance Escape}
\label{sec:equivariance}

Section~\ref{sec:generic} shows that generic distributions
are hard. We now show that distributions with symmetry can
escape this hardness---provided the circuit respects the symmetry.

\subsection{Group Actions on the Grid}

\begin{definition}[Augmented Voronoi grid]\label{def:augmented_grid}
Let $D \geq 2$ denote the number of simplex components
(so the target simplex is $\Delta^{D-1}$) and $N \geq 1$ the
grid resolution. The \emph{natural lattice} is:
\[
    \mathcal{X}_N
    = \Bigl\{\mathbf{k} \in \mathbb{Z}^D_{\geq0} :
    \textstyle\sum_{i=1}^D k_i = N\Bigr\},
    \qquad
    |\mathcal{X}_N| = \tbinom{N+D-1}{D-1}.
\]
For each \emph{polytope type} $m \in \{1,\ldots,D{-}1\}$,
the Translation--Centroid Correspondence defines a centroid map
$G_m: \mathbb{Z}_{\geq 0}^D \to \Real^D$ by:
\begin{equation}\label{eq:centroid_map}
    G_m(\mathbf{k}) = \mathbf{k} + \frac{m}{D}\mathbf{1},
    \qquad \mathbf{1} = (1,\ldots,1).
\end{equation}
\end{definition}

\begin{definition}[Admissible base points]\label{def:base_points}
For each polytope type $m \in \{1, \ldots, D{-}1\}$,
the \emph{type-$m$ admissible base point set} is
\[
    \mathcal{A}_N^{(m)}
    = \bigl\{\mathbf{k} \in \mathbb{Z}_{\geq 0}^D :
    \textstyle\sum_{i=1}^D k_i = N - m\bigr\},
    \qquad
    |\mathcal{A}_N^{(m)}| = \tbinom{N - m + D - 1}{D - 1}.
\]
Each $\mathbf{k} \in \mathcal{A}_N^{(m)}$ is a lattice point of
the $(N{-}m)$-th level set within the solid simplex
$\{\mathbf{k} \in \mathbb{Z}_{\geq 0}^D : \sum k_i \leq N\}$.
The centroid map $G_m$ sends $\mathcal{A}_N^{(m)}$
bijectively onto the type-$m$ centroid set:
\[
    \mathcal{T}_N^{(m)}
    = G_m\bigl(\mathcal{A}_N^{(m)}\bigr)=\Bigl\{G_m(\mathbf{k}) :
    \mathbf{k} \in \mathbb{Z}_{\geq 0}^D,\;
    \textstyle\sum_{i=1}^D k_i = N - m\Bigr\},
    \quad
    |\mathcal{T}_N^{(m)}| = \tbinom{N-m+D-1}{D-1}.
\]
Given a subset of selected types
$\mathcal{M} \subseteq \{1,\ldots,D{-}1\}$ and for each
$m \in \mathcal{M}$, define a subset of centroids
$\mathcal{C}^{(m)} \subseteq \mathcal{T}_N^{(m)}$ closed
under the symmetric group $S_D$ action (acts on $\mathcal{X}_N$ by permuting components:
$\sigma \cdot (k_1,\ldots,k_D) = (k_{\sigma^{-1}(1)},\ldots,
k_{\sigma^{-1}(D)})$), the \emph{augmented grid} is:
\begin{equation}\label{eq:augmented_grid}
    \mathcal{Q}_N
    = \mathcal{X}_N
    \;\cup\; \bigcup_{m \in \mathcal{M}} \mathcal{C}^{(m)}.
\end{equation}
\end{definition}

\begin{definition}[Orbits and stabilizers]
\label{def:orbits}
The \emph{orbit} of a grid point
$\mathbf{x} \in \mathcal{Q}_N$ under $S_D$ is
$\text{Orb}(\mathbf{x}) = \{\sigma\cdot\mathbf{x} :
\sigma \in S_D\}$.
The \emph{stabilizer} is
$\text{Stab}(\mathbf{x}) = \{\sigma \in S_D :
\sigma\cdot\mathbf{x} = \mathbf{x}\}$.
\end{definition}

\begin{lemma}[$S_D$-invariance of each component]
\label{lem:components}
The $S_D$ action on $\mathcal{Q}_N$ preserves each of the
$1 + |\mathcal{M}|$ components. In particular, no $S_D$-orbit
contains points from two distinct components.
\end{lemma}

\begin{proof}
The conclusion follows from two independent facts:
the components are pairwise disjoint,
and each component is $S_D$-invariant.
Together, these imply the orbit separation: if
$A \cap B = \emptyset$ and both $\sigma(A) = A$,
$\sigma(B) = B$ for all $\sigma \in S_D$, then
every orbit $\text{Orb}(x) = \{\sigma \cdot x : \sigma \in S_D\}$
is contained entirely within a single component
(since $x \in A$ implies
$\text{Orb}(x) \subseteq \sigma(A) = A$).

Every component of a lattice point
$\mathbf{k} \in \mathcal{X}_N$ is a non-negative integer.
Every component of a type-$m$ centroid
$G_m(\mathbf{k}) = \mathbf{k} + \frac{m}{D}\mathbf{1}$ has the form
$k_i + \frac{m}{D}$ with $k_i \in \mathbb{Z}_{\geq 0}$
and $1 \leq m \leq D - 1$, so its fractional part
(with respect to $\mathbb{Z}$) is $\frac{m}{D} \in (0,1)$.
Therefore:
\begin{itemize}
    \item Lattice points have integer coordinates;
    centroids do not. Hence
    $\mathcal{X}_N \cap \mathcal{T}_N^{(m)} = \emptyset$
    for all~$m$.
    \item For distinct types $m \neq m'$ with
    $1 \leq m, m' \leq D{-}1$, the fractional parts
    $\frac{m}{D} \neq \frac{m'}{D}$. Hence
    $\mathcal{T}_N^{(m)} \cap \mathcal{T}_N^{(m')} = \emptyset$.
\end{itemize}

Lattice: for $\mathbf{k} \in \mathcal{X}_N$ and $\sigma \in S_D$,
$\sigma \cdot \mathbf{k}$ is a permutation of the components
of~$\mathbf{k}$, hence has the same non-negative integer
entries with the same sum~$N$.
Therefore $\sigma \cdot \mathbf{k} \in \mathcal{X}_N$.

Centroids: the centroid map $G_m$ is $S_D$-equivariant:
\[
    \sigma \cdot G_m(\mathbf{k})
    = \sigma \cdot \mathbf{k} + \frac{m}{D}\,
    \underbrace{(\sigma \cdot \mathbf{1})}_{=\,\mathbf{1}}
    = G_m(\sigma \cdot \mathbf{k}).
\]
If $\mathbf{k} \in \mathbb{Z}_{\geq 0}^D$ (i.e.,
$k_i \geq 0$ and $\sum k_i = N - m$), then
$\sigma \cdot \mathbf{k} \in \mathbb{Z}_{\geq 0}^D$ and $\sum k_i = N-m$.
Therefore
$\sigma \cdot G_m(\mathbf{k}) = G_m(\sigma \cdot \mathbf{k})
\in \mathcal{T}_N^{(m)}$, so $\mathcal{T}_N^{(m)}$ is
$S_D$-invariant.
Since $\mathcal{C}^{(m)} \subseteq \mathcal{T}_N^{(m)}$ is assumed
to be a union of complete $S_D$-orbits,
$\sigma(\mathcal{C}^{(m)}) = \mathcal{C}^{(m)}$.
\end{proof}

\begin{proposition}[Orbit correspondence under $G_m$]
\label{prop:orbit_correspondence}
For each type $m \in \{1, \ldots, D{-}1\}$, two centroids
$G_m(\mathbf{k}_1)$ and $G_m(\mathbf{k}_2)$ lie in the same
$S_D$-orbit if and only if their base points
$\mathbf{k}_1$ and $\mathbf{k}_2$ do:
\begin{equation}\label{eq:orbit_iff}
    G_m(\mathbf{k}_1) \sim_{S_D} G_m(\mathbf{k}_2)
    \quad\Longleftrightarrow\quad
    \mathbf{k}_1 \sim_{S_D} \mathbf{k}_2.
\end{equation}
Consequently, $G_m$ induces a bijection between
the $S_D$-orbits of $\mathcal{A}_N^{(m)}$ and the
$S_D$-orbits of $\mathcal{T}_N^{(m)}$, preserving orbit
sizes and stabilizer subgroups.
\end{proposition}

\begin{proof}\leavevmode

$(\Rightarrow)$
Suppose $\sigma \cdot G_m(\mathbf{k}_1) = G_m(\mathbf{k}_2)$
for some $\sigma \in S_D$.
By the $S_D$-equivariance of $G_m$
(Lemma~\ref{lem:components}):
\[
    G_m(\sigma \cdot \mathbf{k}_1) = G_m(\mathbf{k}_2).
\]
Since $G_m(\mathbf{k}) = \mathbf{k} + \frac{m}{D}\mathbf{1}$ is injective
(adding a constant vector preserves differences:
$G_m(\mathbf{k}) = G_m(\mathbf{k}')$ implies $\mathbf{k} = \mathbf{k}'$),
we conclude $\sigma \cdot \mathbf{k}_1 = \mathbf{k}_2$.

$(\Leftarrow)$
Suppose $\sigma \cdot \mathbf{k}_1 = \mathbf{k}_2$
for some $\sigma \in S_D$.
By equivariance:
\[
    \sigma \cdot G_m(\mathbf{k}_1)
    = G_m(\sigma \cdot \mathbf{k}_1)
    = G_m(\mathbf{k}_2).
\]

The two directions together show that $G_m$ maps
each orbit $\mathcal{O} \subseteq \mathcal{A}_N^{(m)}$
bijectively onto an orbit
$G_m(\mathcal{O}) \subseteq \mathcal{T}_N^{(m)}$,
with $|G_m(\mathcal{O})| = |\mathcal{O}|$
and $\text{Stab}(G_m(\mathbf{k})) = \text{Stab}(\mathbf{k})$
(since $\sigma \cdot G_m(\mathbf{k}) = G_m(\mathbf{k})
\Leftrightarrow \sigma \cdot \mathbf{k} = \mathbf{k}$).
\end{proof}

\begin{theorem}[Orbit count of the augmented grid]
\label{thm:augmented_orbits}
Let $\mathcal{O}_D(\mathcal{Q}_N)$ denote the number of
$S_D$-orbits on the augmented grid~\eqref{eq:augmented_grid}.
Then:
\begin{equation}\label{eq:orbit_bounds}
    p_D(N)
    \;\leq\;
    \mathcal{O}_D(\mathcal{Q}_N)
    \;\leq\;
    p_D(N)
    + \sum_{m \in \mathcal{M}} p_D(N{-}m),
\end{equation}
where $p_D(n)$ is the number of partitions of the integer $n$
into at most $D$ non-negative parts.

If all centroid types are selected and all centroids within
each type are included
($\mathcal{M} = \{1,\ldots,D{-}1\}$ and
$\mathcal{C}^{(m)} = \mathcal{T}_N^{(m)}$), then:
\begin{equation}\label{eq:full_augmented_orbits}
    \mathcal{O}_D(\mathcal{Q}_N)
    = \sum_{m=0}^{D-1} p_D(N{-}m).
\end{equation}
\end{theorem}

\begin{proof}
By Lemma~\ref{lem:components}, the $S_D$-orbits on
$\mathcal{Q}_N$ decompose as a disjoint union of orbits
within each component:
\begin{equation}\label{eq:orbit_decomp}
    \mathcal{O}_D(\mathcal{Q}_N)
    = \mathcal{O}_D(\mathcal{X}_N)
    + \sum_{m \in \mathcal{M}}
    \mathcal{O}_D(\mathcal{C}^{(m)}).
\end{equation}
We bound each term separately.

For Lattice orbits, two lattice points $\mathbf{k}$ and $\mathbf{k}'$ are in the same
$S_D$-orbit if and only if $\exists \sigma \in S_D, \mathbf{k}'=\sigma \cdot \mathbf{k}$ Therefore:
\begin{equation}\label{eq:lattice_orbits}
    \mathcal{O}_D(\mathcal{X}_N) = p_D(N).
\end{equation}

For centroid orbits for a fixed type $m \in \{1, \ldots, D-1\}$, by the commutation $G_m \circ \sigma = \sigma \circ G_m$ (Lemma~\ref{lem:components}), two centroids
$G_m(\mathbf{k}_1), G_m(\mathbf{k}_2) \in \mathcal{C}^{(m)}$ lie in the
same orbit if and only if $\mathbf{k}_1$ and $\mathbf{k}_2$ do. The type-$m$
centroid set is $\mathcal{T}_N^{(m)} = G_m(\mathcal{A}_N^{(m)})$.
By Proposition~\ref{prop:orbit_correspondence},
$G_m$ induces a bijection between the $S_D$-orbits of
$\mathcal{A}_N^{(m)}$ and the $S_D$-orbits of $\mathcal{T}_N^{(m)}$. The base
points of $\mathcal{A}_N^{(m)}$ satisfy
$\sum k_i = N - m$ with $k_i \geq 0$, whose
$S_D$-orbits correspond to partitions of $N{-}m$. Since
$\mathcal{C}^{(m)} \subseteq \mathcal{T}_N^{(m)}$ is
assumed to be a union of complete $S_D$-orbits:
\begin{equation}\label{eq:centroid_orbit_bound}
    \mathcal{O}_D(\mathcal{C}^{(m)})
    \leq \mathcal{O}_D(\mathcal{T}_N^{(m)})
    =\mathcal{O}_D(\mathcal{A}_N^{(m)})= p_D(N{-}m).
\end{equation}
Equality holds when $\mathcal{C}^{(m)} = \mathcal{T}_N^{(m)}$.

Substituting~\eqref{eq:lattice_orbits}
and~\eqref{eq:centroid_orbit_bound}
into~\eqref{eq:orbit_decomp}:
\[
    \mathcal{O}_D(\mathcal{Q}_N)
    = p_D(N) + \sum_{m \in \mathcal{M}}
    \mathcal{O}_D(\mathcal{C}^{(m)})
    \leq p_D(N) + \sum_{m \in \mathcal{M}} p_D(N{-}m).
\]
The lower bound $\mathcal{O}_D(\mathcal{Q}_N) \geq p_D(N)$
follows from $\mathcal{O}_D(\mathcal{C}^{(m)}) \geq 0$.
\end{proof}

\begin{example}

$D = 3$, $N = 7$:
There are $D - 1 = 2$ centroid types. The restricted
partition numbers $p_3(k)$ for $k = 5, 6, 7$ are computed
by enumerating sorted triples
$(k_{(1)} \geq k_{(2)} \geq k_{(3)} \geq 0)$:

\smallskip
\begin{center}
\begin{tabular}{@{}cl@{}}
\toprule
$p_3(7) = 8$ &
$(7,0,0),\,(6,1,0),\,(5,2,0),\,(5,1,1),\,(4,3,0),
\,(4,2,1),\,(3,3,1),\,(3,2,2)$ \\[3pt]
$p_3(6) = 7$ &
$(6,0,0),\,(5,1,0),\,(4,2,0),\,(4,1,1),\,(3,3,0),
\,(3,2,1),\,(2,2,2)$ \\[3pt]
$p_3(5) = 5$ &
$(5,0,0),\,(4,1,0),\,(3,2,0),\,(3,1,1),\,(2,2,1)$ \\
\bottomrule
\end{tabular}
\end{center}
\smallskip

Total: $\mathcal{O}_{3}(\mathcal{Q}_7) = 8 + 7 + 5 = 20$ orbits
over $36 + 28 + 21 = 85$ grid points.

\vspace{1em}

$D = 13$, $N = 7$:
\[
    \mathcal{O}_{13}(\mathcal{Q}_7)
    = \sum_{k=0}^{7} p(k)
    = 1 + 1 + 2 + 3 + 5 + 7 + 11 + 15 = 45.
\]
The full augmented grid has
$\sum_{m=0}^{7}\binom{19-m}{12}
= 77520$ points, giving a compression ratio of
$77520 / 45 = 1722.67$.
\end{example}

\subsection{Symmetric Distributions and Effective Dimension}

\begin{definition}[$S_D$-symmetric distribution]
\label{def:symmetric_dist}
A distribution $P$ on $\mathcal{Q}_N$ is \emph{$S_D$-symmetric} if
$P(\sigma\cdot\mathbf{x}) = P(\mathbf{x})$ for all
$\sigma \in S_D$ and $\mathbf{x} \in \mathcal{Q}_N$.
The set of $S_D$-symmetric distributions on $\mathcal{Q}_N$ is denoted
$\Delta_{\mathrm{sym}}$.
\end{definition}

\begin{proposition}[Dimension of the symmetric subspace]
\label{prop:sym_dim}
\begin{equation}\label{eq:sym_dim}
    \dim( \mathrm{aff} \ \Delta_{\mathrm{sym}}) = \mathcal{O}_D(\mathcal{Q}_N) - 1.
\end{equation}
\end{proposition}

\begin{proof}
An $S_D$-symmetric distribution is constant on each orbit.
With $\mathcal{O}_D(\mathcal{Q}_N)$ orbits, the distribution is specified by
$\mathcal{O}_D(\mathcal{Q}_N)$ non-negative numbers summing to~1, giving
$\mathcal{O}_D(\mathcal{Q}_N) - 1$ free parameters.
\end{proof}

\begin{lemma}[Equivariance of the Voronoi discretization]\label{lem:voronoi-equiv}
Let $Q_N \subset \Delta^{D-1}$ be $S_D$-invariant and let
$\{\mathrm{cell}(x)\}_{x \in Q_N}$ be the Voronoi partition of
$\Delta^{D-1}$ induced by $Q_N$ under the Euclidean metric. Then
$\sigma\big(\mathrm{cell}(x)\big) = \mathrm{cell}(\sigma \cdot x)$ up to
$\lambda_{D-1}$-null boundary sets, for every $\sigma \in S_D$.
Consequently, if the target density $f$ is $S_D$-invariant, the discretized
distribution $P^{\mathrm{target}}(x) = \int_{\mathrm{cell}(x)} f \,
d\lambda_{D-1}$ is $S_D$-symmetric on $Q_N$.
\end{lemma}

\begin{proof}
Coordinate permutations are Euclidean isometries of $\Delta^{D-1}$:
$\|\sigma \cdot s - \sigma \cdot y\| = \|s - y\|$. Hence for
$s \in \mathrm{int}\,\mathrm{cell}(x)$,
$\|s - x\| < \|s - y\|$ for all $y \in Q_N \setminus \{x\}$, which is
equivalent to $\|\sigma \cdot s - \sigma \cdot x\| <
\|\sigma \cdot s - \sigma \cdot y\|$ for all $\sigma \cdot y \in
Q_N \setminus \{\sigma \cdot x\}$ (using $\sigma(Q_N) = Q_N$), i.e.\
$\sigma \cdot s \in \mathrm{int}\,\mathrm{cell}(\sigma \cdot x)$. Then, $\mathrm{cell}(\sigma\cdot\mathbf{x})
    = \sigma\cdot\mathrm{cell}(\mathbf{x}),
    \lambda_{D-1}\bigl(\mathrm{cell}(\sigma\cdot\mathbf{x})\bigr)
    = \lambda_{D-1}\bigl(\mathrm{cell}(\mathbf{x})\bigr)$. Since cell
boundaries lie in finitely many hyperplanes and are $\lambda_{D-1}$-null, by invariance of $f$ and of the Lebesgue measure under permutations,
\[
P^{\mathrm{target}}(\sigma \cdot x)
= \int_{\mathrm{cell}(\sigma \cdot x)} f \, d\lambda
= \int_{\sigma(\mathrm{cell}(x))} f \, d\lambda
= \int_{\mathrm{cell}(x)} f \circ \sigma \, d\lambda
= P^{\mathrm{target}}(x). \qedhere
\]
\end{proof}

\begin{corollary}[Dimensional reduction for Dirichlet]
\label{cor:dirichlet_dim}
The symmetric Dirichlet distribution
$\mathrm{Dir}(\alpha\mathbf{1}_D)$ is $S_D$-symmetric, since its
density $f(\mathbf{x}) \propto \prod_{i=1}^D x_i^{\alpha-1}$
is invariant under coordinate permutation.
On the augmented grid $\mathcal{Q}_N$, it is fully determined by
$\mathcal{O}_D(\mathcal{Q}_N) - 1$ independent parameters (one
probability weight per orbit, minus normalization).
\end{corollary}

\subsection{Equivariant Circuits and the Symmetry Escape}
\label{subsec:equivariant}

The generic hardness results of Section~\ref{sec:generic}
show that approximating arbitrary distributions requires
$N_p \gtrsim M$. We now show that $S_D$-symmetric targets
can \emph{escape} this barrier---provided the circuit
respects the symmetry.

\begin{definition}[Induced representation]
\label{def:induced_rep}
Let $\phi: \mathcal{Q}_N \longleftrightarrow \{0,1\}^n$ be
the encoding bijection ($|\mathcal{Q}_N| = 2^n$). The $S_D$ action
on $\mathcal{Q}_N$ induces:
\begin{enumerate}
    \item A \emph{permutation representation} on the
    computational basis:
    \begin{equation}\label{eq:induced_perm}
        \pi: S_D \to S_{2^n}, \qquad
        \pi(\sigma)(x) = \phi(\sigma \cdot \phi^{-1}(x)).
    \end{equation}
    \item A \emph{unitary representation} on the $n$-qubit
    Hilbert space $\mathcal{H} = (\Complex^2)^{\otimes n}$:
    \begin{equation}\label{eq:unitary_rep}
        \Pi: S_D \to U(2^n), \qquad
        \Pi(\sigma)\ket{x} = \ket{\pi(\sigma)(x)}.
    \end{equation}
\end{enumerate}
Both depend on the choice of encoding~$\phi$.
\end{definition}

\begin{remark}[Exact grid construction in practice]
\label{rem:exact_grid}
The theoretical analysis requires $\mathcal{Q}_N$ to consist of
complete $S_D$-orbits. In implementation, one must also
ensure $|\mathcal{Q}_N| = 2^n$ to fully utilize the Hilbert space
(dummy states with forced zero probability introduce
numerical difficulties in training).

\end{remark}

\begin{definition}[$S_D$-equivariant PQC]
\label{def:equivariant}
A PQC $U(\bm{\theta})$ is \emph{$S_D$-equivariant}
(with respect to the encoding $\phi$) if:
\begin{equation}\label{eq:equivariant}
    U(\bm{\theta})\,\Pi(\sigma)
    = \Pi(\sigma)\,U(\bm{\theta})
    \qquad \forall\, \sigma \in S_D,\;
    \forall\, \bm{\theta}.
\end{equation}
\end{definition}

\begin{proposition}[Equivariance implies Born invariance]
\label{prop:equiv_implies_inv}
If $U(\bm{\theta})$ is $S_D$-equivariant and the initial
state $\ket{0}^{\otimes n}$ is $\Pi$-invariant
(i.e., $\Pi(\sigma)\ket{0}^{\otimes n}
= \ket{0}^{\otimes n}$ for all $\sigma \in S_D$),
then the Born distribution is $S_D$-symmetric:
\begin{equation}\label{eq:born_invariant}
    P_{\bm{\theta}}^{\mathrm{Born}}(\pi(\sigma)(x))
    = P_{\bm{\theta}}^{\mathrm{Born}}(x)
    \qquad \forall\, \sigma \in S_D,\;
    \forall\, \bm{\theta},\;
    \forall\, x \in \{0,1\}^n.
\end{equation}
\end{proposition}

\begin{proof}
$\forall \sigma \in S_D$:
\begin{align}
    P_{\bm{\theta}}(\pi(\sigma) (x))
    &= \bigl|\bra{\pi(\sigma) (x)}\,
    U(\bm{\theta})\ket{0}^{\otimes n}\bigr|^2
    \notag\\
    &= \bigl|\bra{x}\,\Pi(\sigma)^\dagger\,
    U(\bm{\theta})\ket{0}^{\otimes n}\bigr|^2
    \notag\\
    &= \bigl|\bra{x}\,\Pi(\sigma)^{-1}\,
    U(\bm{\theta})\ket{0}^{\otimes n}\bigr|^2
    \notag\\
    &= \bigl|\bra{x}\,\Pi(\sigma^{-1})\,
    U(\bm{\theta})\ket{0}^{\otimes n}\bigr|^2
    \notag\\
    &= \bigl|\bra{x}\,
    U(\bm{\theta})\,\Pi(\sigma^{-1})
    \ket{0}^{\otimes n}\bigr|^2
    \notag\\
    &= \bigl|\bra{x}\,
    U(\bm{\theta})\ket{0}^{\otimes n}\bigr|^2
    \notag\\
    &= P_{\bm{\theta}}(x). \notag
\end{align}
\end{proof}

\begin{remark}[Initial state invariance as an encoding constraint]
\label{rem:init-state}
The condition $\Pi(\sigma)|0\rangle^{\otimes n} = |0\rangle^{\otimes n}$
requires the grid point $x_0 \coloneqq \phi^{-1}(0\cdots 0)$ to be fixed by
every $\sigma \in S_D$, i.e.\ all coordinates of $x_0$ equal. Such a point exists precisely when the barycentre
$c_N=(N/D,\ldots,N/D)$ belongs to $Q_N$. If $D\mid N$, it belongs
to the mandatory lattice $X_N$. Otherwise, writing
$m\equiv N\pmod D$ with $1\leq m\leq D-1$, its inclusion requires
$m\in\mathcal M$ and $c_N\in C^{(m)}$; selecting the centroid type
alone does not guarantee inclusion of its singleton orbit.

When no fixed point exists, one may instead prepare the orbit-uniform
superposition
\[
|\psi_0\rangle \;=\; \frac{1}{\sqrt{|\text{Orb}(x_0)|}}
\sum_{x \in \text{Orb}(x_0)} |\phi(x)\rangle
\;\;\propto\;\; \text{Sym}\,|0\rangle^{\otimes n},
\qquad
\text{Sym} \coloneqq \frac{1}{|S_D|}\sum_{\sigma \in S_D}\Pi(\sigma),
\]
which lies in the trivial isotypic component of $\Pi$ and satisfies
$\Pi(\sigma)|\psi_0\rangle = |\psi_0\rangle$ by construction.
\end{remark}

\begin{theorem}[Equivariant escape from generic hardness]
\label{thm:escape}
If the PQC is $S_D$-equivariant (with $\Pi$-invariant
initial state), then for $S_D$-symmetric targets, the
approximation problem reduces from $\Delta^{M-1}$
(dimension $M - 1$) to $\Delta_{\mathrm{sym}}$
(dimension $\mathcal{O}_D(\mathcal{Q}_N) - 1$). In particular, a necessary condition
for the circuit to be able to cover the target space is:
\begin{equation}\label{eq:escape}
 N_p \geq \mathcal{O}_D(\mathcal{Q}_N) - 1
 \quad\text{(in place of } N_p \geq M - 1\text{)}.
\end{equation}
Since $\mathcal{O}_D(\mathcal{Q}_N) \le \sum_{m=0}^{D-1} p_D(N-m)$, two regimes arise:
(i) for fixed $D$, $p_D(N) = \Theta(N^{D-1})$ grows polynomially in the
resolution $N$; (ii) for $D > N$, $p_D(N) = p(N)$ is independent of $D$ but grows superpolynomially, $p(N) \sim \frac{1}{4N\sqrt{3}}
\exp\big(\pi\sqrt{2N/3}\big)$ (Hardy-Ramanujan formula). In both regimes $\mathcal{O}_D(\mathcal{Q}_N)$ is exponentially smaller than $|\mathcal{Q}_N| = 2^n$ for large $D$, so the minimal parameter
requirement $N_p \ge \mathcal{O}_D(\mathcal{Q}_N) - 1$ is satisfiable with $N_p = O(n^2)$ for
moderate $N$.
\end{theorem}

\begin{proof}
By Proposition~\ref{prop:equiv_implies_inv}, an
$S_D$-equivariant circuit with $\Pi$-invariant initial
state automatically produces $S_D$-symmetric Born
distributions. Its reachable set
$\mathcal{B}_{\mathrm{sym}} = \mathcal{B} \cap \Delta_{\mathrm{sym}}$
lies entirely within $\Delta_{\mathrm{sym}}$.

Applying the covering number framework of
Section~\ref{sec:generic} to the restricted problem
in $\Delta_{\mathrm{sym}}$: the simplex
$\Delta_{\mathrm{sym}}$ has dimension $\mathcal{O}_D(\mathcal{Q}_N) - 1$
(Proposition~\ref{prop:sym_dim}), so the covering number
lower bound (Proposition~\ref{prop:covering_simplex})
becomes:
\[
    \log\mathcal{N}(\Delta_{\mathrm{sym}},
    \varepsilon, \TV)
    \geq (\mathcal{O}_D(\mathcal{Q}_N) - 1)\log\frac{1}{4\varepsilon}.
\]
The covering number upper bound for $\mathcal{B}_{\mathrm{sym}}$
remains $\log\mathcal{N}(\mathcal{B}_{\mathrm{sym}},
\varepsilon, \TV)
\leq N_p\log(\pi N_p/\varepsilon + 1)$
(the circuit still has $N_p$ parameters).

The inapproximability threshold
(Theorem~\ref{thm:generic}) with $M - 1$ replaced by
$\mathcal{O}_D(\mathcal{Q}_N) - 1$ gives:
\[
    \alpha_{\mathrm{sym}}
    = \frac{N_p}{\mathcal{O}_D(\mathcal{Q}_N) - 1}.
\]
The condition $\alpha_{\mathrm{sym}} \geq 1$ (i.e.,
$N_p \geq \mathcal{O}_D(\mathcal{Q}_N) - 1$) is necessary for the reachable set
to cover $\Delta_{\mathrm{sym}}$.
\end{proof}

\subsection{Constructive Approaches to Symmetry-Respecting
QCBM}
\label{subsec:constructive}

Theorem~\ref{thm:escape} gives the necessary parameter condition. Strategy A consists of circuits equivariant under the encoding-induced representation $\Pi$ of Definition \ref{def:induced_rep}, with a fixed $\Pi$-invariant pure initial state. Strategy B uses the orbit-space representation. Proposition \ref{prop:dominance} compares these strategies at the level of abstract parameterized ansatz families. We first recall the full-CZ construction for qubit-permutation symmetry and distinguish its representation from the induced action on the simplex grid.

\subsubsection{Equivariant circuits and the qubit-permutation reference model}\label{sec:strategyA}

\begin{assumption}[Qubit-permutation reference setting]\label{qubitpermutation}
For the reference construction below, fix a subgroup $H\leq S_n$
acting on the $n$ qubits by permutation of tensor factors. Denote
this representation by $h\mapsto\Pi_h$, with
\begin{equation}
\Pi_h\lvert b_0\cdots b_{n-1}\rangle
=
\lvert b_{h^{-1}(0)}\cdots b_{h^{-1}(n-1)}\rangle,
\qquad h\in H.
\end{equation}
This reference representation is distinct from the encoding-induced
representation $\Pi$ of Definition~\ref{def:induced_rep}.
\end{assumption}

\begin{remark}[Encoding obstruction]
The induced $S_D$-action on the grid $Q_N$ cannot be realized by
qubit permutations through an encoding bijection
$Q_N\leftrightarrow\{0,1\}^n$. Indeed, every qubit permutation
fixes both $0^n$ and $1^n$. Their preimages under an equivariant
bijection would therefore be two distinct points fixed by every
element of $S_D$. However, every $x\in Q_N$ satisfies
$\sum_{i=1}^D x_i=N$, and the only possible globally fixed point
is $(N/D,\ldots,N/D)$. Since $D\geq2$ and $N\geq1$ imply
$n\geq1$, this is a contradiction.

The qubit-permutation construction below is consequently a separate
reference model. The grid-space comparison in Proposition~\ref{prop:dominance} uses
the general induced representation of Definition~\ref{def:induced_rep} and an explicitly
invariant initial state. The orbit-space construction of Theorem~\ref{thm:orbit_symmetric}
requires no identification of these two representations.
\end{remark}

\begin{definition}[Full entanglement layer]
\label{def:full_entangle}
The \emph{full CZ entanglement layer} on $n$ qubits applies
a controlled-$Z$ gate to every pair of qubits:
\begin{equation}\label{eq:full_cz}
   E_{\mathrm{full}} = \prod_{0 \leq i < j \leq n-1}
\mathrm{CZ}_{i,j}
\end{equation}
where the product runs over all $\binom{n}{2}$ unordered
pairs. The \emph{full-entanglement PQC} is:
\begin{equation}\label{eq:full_pqc}
    U(\bm{\theta}) = \prod_{l=0}^{L}
    \Bigl(E_{\mathrm{full}} \cdot
    \bigotimes_{q=0}^{n-1} W_q^{(l)}(\bm{\theta})\Bigr)
    \cdot H^{\otimes n},
\end{equation}
where $W_q^{(l)} = R_Y(\theta_1)R_Z(\theta_2)R_Y(\theta_3)$
is the YZY single-qubit block on qubit $q$ at layer $l$.
\end{definition}

\begin{definition}[Structural symmetry group]
\label{def:structural_sym}
A PQC $U(\bm{\theta})$ on $n$ qubits has \emph{structural
symmetry group} $G_{\mathrm{circuit}} \leq S_n$ if for
every qubit permutation $\tau \in G_{\mathrm{circuit}}$
there exists a parameter permutation
$\tau_*: \Real^{N_p} \to \Real^{N_p}$ such that:
\begin{equation}\label{eq:structural_sym_def}
    \Pi_\tau\, U(\bm{\theta})\,
    \Pi_\tau^\dagger
    = U(\tau_*(\bm{\theta}))
    \qquad \forall\, \bm{\theta},
\end{equation}
where $\Pi_\tau$ is the unitary that permutes qubits
according to $\tau$:
$\Pi_\tau\ket{b_0 \cdots b_{n-1}}
= \ket{b_{\tau^{-1}(0)} \cdots b_{\tau^{-1}(n-1)}}$.
\end{definition}

\begin{proposition}[Structural symmetry of full CZ]
\label{prop:full_cz_sym}
The full CZ entanglement layer is invariant under every
qubit permutation:
\begin{equation}\label{eq:full_cz_inv}
    \Pi_\tau\, E_{\mathrm{full}}\, \Pi_\tau^\dagger
    = E_{\mathrm{full}}
    \qquad \forall\, \tau \in S_n,
\end{equation}
Consequently, the structural symmetry group of the
full-entanglement PQC is $G_{\mathrm{circuit}} = S_n$.
\end{proposition}

\begin{proof}
The CZ gate satisfies $\mathrm{CZ}_{i,j} = \mathrm{CZ}_{j,i}$
(it is symmetric in its two qubits) and acts as the diagonal
unitary $\mathrm{diag}(1,1,1,-1)$ in the computational basis
of qubits $i$ and $j$.

For any qubit permutation $\tau \in S_n$:
\begin{equation}\label{eq:cz_conjugation}
    \Pi_\tau\, \mathrm{CZ}_{i,j}\, \Pi_\tau^\dagger
    = \mathrm{CZ}_{\tau(i),\tau(j)}.
\end{equation}
This holds because conjugation by $\Pi_\tau$ relabels qubit
$i$ as qubit $\tau(i)$.

Therefore:
\[
    \Pi_\tau\, E_{\mathrm{full}}\, \Pi_\tau^\dagger
    = \prod_{0 \leq i < j \leq n-1}
    \mathrm{CZ}_{\tau(i),\tau(j)}
    = \prod_{0 \leq i' < j' \leq n-1}
    \mathrm{CZ}_{i',j'}
    = E_{\mathrm{full}},
\]
where the second equality uses the fact that
$\tau$ is a bijection on $\{0,\ldots,n{-}1\}$,
so $\{\{\tau(i),\tau(j)\} : 0 \leq i < j \leq n-1\}
= \{\{i,j\} : 0 \leq i < j \leq n-1\}$
(the set of all unordered pairs is invariant under relabeling).

The YZY block on qubit $q$ at layer $l$ is
$W_q^{(l)}(\bm{\theta})
= R_Y(\theta_{q,1}^{(l)})R_Z(\theta_{q,2}^{(l)})
R_Y(\theta_{q,3}^{(l)})$.
Conjugation by $\Pi_\tau$ moves the block from
qubit $q$ to qubit $\tau(q)$:
\[
    \Pi_\tau\left(\bigotimes_{q=0}^{n-1}
    W_q^{(l)}(\bm{\theta})\right)
    \Pi_\tau^\dagger
    = \bigotimes_{q=0}^{n-1}
    W_{\tau^{-1}(q)}^{(l)}(\bm{\theta})
    = \bigotimes_{q=0}^{n-1}
    W_q^{(l)}(\tau_*(\bm{\theta})),
\]
where the parameter permutation $\tau_*$ is induced by
$(\tau_*(\bm{\theta}))_{q,k}^{(l)}
= \theta_{\tau^{-1}(q),k}^{(l)}$
(the parameters of qubit $q$ in the permuted circuit
are the original parameters of qubit $\tau^{-1}(q)$).

Therefore,
\[
    \Pi_\tau\, U(\bm{\theta})\,
    \Pi_\tau^\dagger
    = U(\tau_*(\bm{\theta}))
    \qquad \forall\, \tau \in S_n,
\]
confirming $G_{\mathrm{circuit}} = S_n$ in the sense
of Definition~\ref{def:structural_sym}.
\end{proof}
\begin{definition}[Equivariant Lie subalgebra]
\label{def:equiv_lie}
For a subgroup $H \le S_n$ acting on $n$ qubits via the unitary
representation $\tau \mapsto \Pi_\tau$, define the
\emph{$H$-equivariant Lie subalgebra} of the Hermitian operators
$\mathcal{H}_n := i\,\mathfrak{u}(2^n)$:
\begin{equation}\label{eq:equiv_lie_def}
    \mathcal{H}_n^{H}
    = \bigl\{X \in \mathcal{H}_n :
    \Pi_\tau\, X\, \Pi_\tau^\dagger = X
    \text{ for all } \tau \in H\bigr\}.
\end{equation}
A Hermitian operator $X$ generates an $H$-equivariant unitary
$e^{-i\theta X}$ (i.e., $\Pi_\tau e^{-i\theta X} \Pi_\tau^\dagger
= e^{-i\theta X}$ for all $\theta$ and all $\tau \in H$) if and
only if $X \in \mathcal{H}_n^{H}$.
\end{definition}

\begin{lemma}[Pauli conjugation under qubit permutation]
\label{lem:pauli_conjugation}
For any single-qubit Pauli operator $P \in \{X, Y, Z\}$ acting on
qubit $q \in \{0, \ldots, n{-}1\}$, and any qubit permutation
$\tau \in S_n$:
\begin{equation}\label{eq:pauli_conj}
    \Pi_\tau\, P_q\, \Pi_\tau^\dagger = P_{\tau(q)}.
\end{equation}
Consequently, $P_{q_1} = P_{q_2}$ as operators on $(\Complex^2)^{\otimes n}$
if and only if $q_1 = q_2$.
\end{lemma}

\begin{proof}
By definition,
$\Pi_\tau \ket{b_0 \cdots b_{n-1}}
= \ket{b_{\tau^{-1}(0)} \cdots b_{\tau^{-1}(n-1)}}$.
For any Pauli operator $P$ acting on qubit $q$, the operator
$\Pi_\tau P_q \Pi_\tau^\dagger$ acts on the qubit at position
$\tau(q)$ in the relabeled basis, with the same Pauli action $P$.
This is precisely $P_{\tau(q)}$.

For the second statement: $P_{q_1}$ and $P_{q_2}$ are distinct
Hermitian operators on $(\Complex^2)^{\otimes n}$ when $q_1 \neq q_2$,
since they act non-trivially on different qubits. (Formally,
$\mathrm{Tr}[(P_{q_1} - P_{q_2})^2] = 2 \cdot 2^n \neq 0$ when
$q_1 \neq q_2$.)
\end{proof}

\begin{theorem}[Equivariance of the shared-parameter
full-entanglement PQC]\label{fullentanglementPQC}
Let $H\leq S_n$ be the subgroup in the qubit-permutation reference
setting, and let $r_H$ be the number of its orbits on
$\{0,\ldots,n-1\}$. Restrict the parameters of the
full-entanglement PQC~\ref{eq:full_pqc} to
\[
\Theta_H
=
\left\{
\theta:
\theta_{q,k}^{(l)}=\theta_{h(q),k}^{(l)}
\text{ for all }h\in H,q,k,l
\right\},
\qquad
\dim\Theta_H=3r_H(L+1).
\]
Then
\[
\Pi_hU(\theta)\Pi_h^\dagger=U(\theta)
\qquad
\text{for all }h\in H\text{ and }\theta\in\Theta_H.
\]
Since $\lvert0\rangle^{\otimes n}$ is invariant under qubit
permutations, the Born distribution is constant on the $H$-orbits
of computational-basis bit strings.
\end{theorem}

\begin{proof}
Fix $h\in H$. By Proposition~9,
$\Pi_hE_{\mathrm{full}}\Pi_h^\dagger=E_{\mathrm{full}}$.
The Hadamard layer also satisfies
$\Pi_h\mathrm H^{\otimes n}\Pi_h^\dagger=\mathrm H^{\otimes n}$,
because it acts identically on every qubit.

Conjugating a single-qubit layer gives
\[
\Pi_h
\left(\bigotimes_{q=0}^{n-1}W_q^{(l)}(\theta)\right)
\Pi_h^\dagger
=
\bigotimes_{q=0}^{n-1}W_{h^{-1}(q)}^{(l)}(\theta).
\]
For $\theta\in\Theta_H$, the qubits $q$ and $h^{-1}(q)$ have
identical parameter triples, hence
\[
W_{h^{-1}(q)}^{(l)}(\theta)=W_q^{(l)}(\theta).
\]
Every layer therefore commutes with $\Pi_h$, and so does their
product $U(\theta)$.

Writing
$p_\theta(b)=|\langle b|U(\theta)|0^n\rangle|^2$,
equivariance and invariant initialization give
\[
p_\theta(h\cdot b)
=
\left|
\langle b|\Pi_h^\dagger U(\theta)|0^n\rangle
\right|^2
=
p_\theta(b).
\]
\end{proof}

\subsubsection{Strategy B (Route Proposed): Orbit-Space QCBM}\label{sec:strategyB}

\begin{definition}[Orbit-space encoding]
\label{def:orbit_encoding}
Let $\mathcal{O}_1, \ldots, \mathcal{O}_K$ be the $S_D$-orbits of
the augmented grid $\mathcal{Q}_N$, constructed by selecting
complete $S_D$-orbits from the lattice (mandatory) and
centroids (optional) as in
Theorem~\ref{thm:augmented_orbits}, with $K$ chosen to be a
power of~$2$. Let $s_j = |\mathcal{O}_j|$ denote the size of
orbit~$j$. Define the \emph{orbit-index function}
$j: \mathcal{Q}_N \to \{1, \ldots, K\}$ by assigning to each
$\mathbf{x} \in \mathcal{Q}_N$ the unique index
$j(\mathbf{x}) \in \{1, \ldots, K\}$ such that
$\mathbf{x} \in \mathcal{O}_{j(\mathbf{x})}$.

The \emph{orbit-space QCBM} uses $n' = \log_2 K$ qubits with a
bijection:
\[
    \psi: \{1, \ldots, K\}
    \longleftrightarrow \{0,1\}^{n'}.
\]
The Born distribution on orbit indices is:
\begin{equation}\label{eq:orbit_born}
    q_{\bm{\theta}}(j) =
    \bigl|\langle \psi(j) | U(\bm{\theta})
    |0\rangle^{\otimes n'}\bigr|^2,
    \qquad j = 1, \ldots, K,
\end{equation}
which automatically satisfies $\sum_{j=1}^K q_{\bm{\theta}}(j)
= 1$ (no inactive states, since $K = 2^{n'}$). The induced
distribution on the full grid is:
\begin{equation}\label{eq:orbit_to_grid}
    P_{\bm{\theta}}(\mathbf{x})
    = \frac{q_{\bm{\theta}}(j(\mathbf{x}))}{s_{j(\mathbf{x})}}
    \qquad \text{for } \mathbf{x} \in \mathcal{O}_{j(\mathbf{x})}.
\end{equation}
\end{definition}

\begin{theorem}[Orbit-space QCBM produces $S_D$-symmetric
distributions]\label{thm:orbit_symmetric}
For every parameter vector $\bm{\theta}$, the distribution
$P_{\bm{\theta}}$ defined by~\eqref{eq:orbit_to_grid} is a valid
$S_D$-symmetric probability distribution on $\mathcal{Q}_N$,
i.e., $P_{\bm{\theta}} \in \Delta_{\mathrm{sym}}$. Moreover, the
map $\Phi: \Delta^{K-1} \to \Delta_{\mathrm{sym}}$ defined by
$[\Phi(\mathbf{q})](\mathbf{x})
= q_{j(\mathbf{x})} / s_{j(\mathbf{x})}$ is a bijective isometry
under the total variation metric, and the circuit
$U(\bm{\theta})$ requires no symmetry properties.
\end{theorem}

\begin{proof}\leavevmode

\textbf{$P_{\bm{\theta}}$ is $S_D$-symmetric.}
By Definition~\ref{def:augmented_grid}, $\mathcal{Q}_N$ is a
union of complete $S_D$-orbits, so for every
$\mathbf{x} \in \mathcal{Q}_N$ and every $\sigma \in S_D$,
$\sigma \cdot \mathbf{x} \in \mathcal{Q}_N$ and lies in the
same orbit as $\mathbf{x}$. Hence $j(\sigma \cdot \mathbf{x})
= j(\mathbf{x})$ and $s_{j(\sigma \cdot \mathbf{x})}
= s_{j(\mathbf{x})}$. Therefore:
\[
    P_{\bm{\theta}}(\sigma \cdot \mathbf{x})
    = \frac{q_{\bm{\theta}}(j(\sigma \cdot \mathbf{x}))}
    {s_{j(\sigma \cdot \mathbf{x})}}
    = \frac{q_{\bm{\theta}}(j(\mathbf{x}))}
    {s_{j(\mathbf{x})}}
    = P_{\bm{\theta}}(\mathbf{x}).
\]
This holds for every $\bm{\theta}$, with no constraint on the
circuit architecture.

\textbf{$P_{\bm{\theta}}$ is a valid probability
distribution.}
Non-negativity is immediate from $q_{\bm{\theta}}(j) \geq 0$ and
$s_j > 0$. For normalization:
\[
    \sum_{\mathbf{x} \in \mathcal{Q}_N}
    P_{\bm{\theta}}(\mathbf{x})
    = \sum_{j=1}^{K} \sum_{\mathbf{x} \in \mathcal{O}_j}
    \frac{q_{\bm{\theta}}(j)}{s_j}
    = \sum_{j=1}^{K} s_j \cdot
    \frac{q_{\bm{\theta}}(j)}{s_j}
    = \sum_{j=1}^{K} q_{\bm{\theta}}(j) = 1,
\]
where the second equality uses $|\mathcal{O}_j| = s_j$, and the
last uses $\sum_j q_{\bm{\theta}}(j) = 1$ (Born rule with
$K = 2^{n'}$), this establishes
$P_{\bm{\theta}} \in \Delta_{\mathrm{sym}}$.

\textbf{$\Phi$ is well-defined.}
The same arguments as (replacing $q_{\bm{\theta}}$ by an arbitrary
$\mathbf{q} \in \Delta^{K-1}$) show that $\Phi(\mathbf{q})$ is
$S_D$-symmetric and properly normalized, hence
$\Phi(\mathbf{q}) \in \Delta_{\mathrm{sym}}$. Note that $\Delta^{K-1}$ has dimension $K-1$, and
$\Delta_{\mathrm{sym}}$ has dimension
$\mathcal{O}_D(\mathcal{Q}_N) - 1 = K - 1$
(Proposition~\ref{prop:sym_dim}), so a bijection between them is
dimensionally consistent. We construct it explicitly via $\Phi$.

\textbf{$\Phi$ is injective.}
Suppose $\Phi(\mathbf{q}) = \Phi(\mathbf{q}')$. For each
$j \in \{1, \ldots, K\}$, pick any representative
$\mathbf{x}_j \in \mathcal{O}_j$. Then:
\[
    \frac{q_j}{s_j}
    = [\Phi(\mathbf{q})](\mathbf{x}_j)
    = [\Phi(\mathbf{q}')](\mathbf{x}_j)
    = \frac{q'_j}{s_j},
\]
hence $q_j = q'_j$ for all $j$, i.e., $\mathbf{q} = \mathbf{q}'$.

\textbf{$\Phi$ is surjective.}
Let $P \in \Delta_{\mathrm{sym}}$ be arbitrary. Since $P$ is
$S_D$-symmetric, $P$ is constant on each orbit
$\mathcal{O}_j$; denote this common value by
$P_j := P(\mathbf{x}_j)$ for any representative
$\mathbf{x}_j \in \mathcal{O}_j$. Define
$\mathbf{q} \in \Real^K$ by:
\[
    q_j := s_j\, P_j, \qquad j = 1, \ldots, K.
\]
We verify $\mathbf{q} \in \Delta^{K-1}$:
\begin{itemize}
    \item Non-negativity: $q_j = s_j P_j \geq 0$ since
    $s_j > 0$ and $P_j \geq 0$.
    \item Normalization:
    \[
        \sum_{j=1}^{K} q_j
        = \sum_{j=1}^{K} s_j P_j
        = \sum_{j=1}^{K} \sum_{\mathbf{x} \in \mathcal{O}_j}
        P(\mathbf{x})
        = \sum_{\mathbf{x} \in \mathcal{Q}_N} P(\mathbf{x})
        = 1,
    \]
    where the second equality uses orbit-constancy of $P$.
\end{itemize}
Then for any $\mathbf{x} \in \mathcal{Q}_N$:
\[
    [\Phi(\mathbf{q})](\mathbf{x})
    = \frac{q_{j(\mathbf{x})}}{s_{j(\mathbf{x})}}
    = \frac{s_{j(\mathbf{x})} P_{j(\mathbf{x})}}
    {s_{j(\mathbf{x})}}
    = P_{j(\mathbf{x})}
    = P(\mathbf{x}),
\]
so $\Phi(\mathbf{q}) = P$.

\textbf{$\Phi$ is an isometry under total variation.}
For $\mathbf{q}, \mathbf{q}' \in \Delta^{K-1}$:
\begin{align*}
    \TV\bigl(\Phi(\mathbf{q}), \Phi(\mathbf{q}')\bigr)
    &= \frac{1}{2}
    \sum_{\mathbf{x} \in \mathcal{Q}_N}
    \bigl|[\Phi(\mathbf{q})](\mathbf{x})
    - [\Phi(\mathbf{q}')](\mathbf{x})\bigr| \\
    &= \frac{1}{2}
    \sum_{j=1}^{K} \sum_{\mathbf{x} \in \mathcal{O}_j}
    \frac{|q_j - q'_j|}{s_j}
    = \frac{1}{2}
    \sum_{j=1}^{K} s_j \cdot \frac{|q_j - q'_j|}{s_j}
    = \frac{1}{2} \sum_{j=1}^{K} |q_j - q'_j| \\
    &= \TV(\mathbf{q}, \mathbf{q}').
\end{align*}

Since $\Phi$ is a bijective isometry,
$\Delta_{\mathrm{sym}}$ and $\Delta^{K-1}$ are isometrically
isomorphic under TV. The orbit-space QCBM with reachable set
$\mathcal{B}_{n'} \subseteq \Delta^{K-1}$ produces exactly the
distribution family $\Phi(\mathcal{B}_{n'}) \subseteq
\Delta_{\mathrm{sym}}$.

In Strategy~A, the circuit $U(\theta)$ acts on the full grid register
and satisfies $[U(\theta),\Pi(\sigma)]=0$ for the encoding-induced
representation $\Pi$ of Definition~\ref{def:induced_rep}, with a fixed invariant pure
initial state. This is the general grid-space setting, independent
of the qubit-permutation reference model of Section~\ref{sec:strategyA}. In
Strategy~B (the orbit-space formulation), the encoding $\psi$
identifies the computational basis $\{|j\rangle\}_{j=1}^K$ of
$\mathcal{H}_{n'} \cong (\Complex^2)^{\otimes n'}$ with the orbit
set $\{\mathcal{O}_1, \ldots, \mathcal{O}_K\}$.

The natural action of $S_D$ on this quotient sends each orbit to
itself: $\sigma \cdot \mathcal{O}_j
= \{\sigma \cdot \mathbf{x} : \mathbf{x} \in \mathcal{O}_j\}
= \mathcal{O}_j$ (by definition of orbit). Therefore the induced
unitary representation $\Pi'$ on $\mathcal{H}_{n'}$ permutes the
basis states $\{|j\rangle\}$ according to this trivial action,
i.e., $\Pi'(\sigma)|j\rangle = \ket{\pi'(\sigma)(j)}=\ket{\psi(\sigma\cdot\psi^{-1}(j))}=\ket{\psi(\sigma\cdot\mathcal{O}_j)}=\ket{\psi(\mathcal{O}_j)}=|j\rangle$ for all
$\sigma \in S_D$ and all $j$. By linearity, $\Pi'(\sigma) = I$ on
$\mathcal{H}_{n'}$.

Since $[U(\bm{\theta}), I] = 0$ for any unitary $U(\bm{\theta})$,
the equivariance condition is satisfied trivially, without any
constraint on the circuit architecture.
\end{proof}

\begin{proposition}[Strategy B dominates Strategy A
as an abstract ansatz]\label{prop:dominance}
Assume that $Q_N$ contains an $S_D$-fixed point $x_0$, and fix the
encoding so that $\phi(x_0)=0^n$. Let the Strategy~A ansatz have
$L$ abstract layers, each commuting with the induced representation
$\Pi$ of Definition~\ref{def:induced_rep}, and $N_p^{(A)}$ independent parameters.
At the level of abstract parameterized unitary families, there exists
an orbit-space ansatz with $N_p^{(B)}\leq N_p^{(A)}$ parameters and
at most $L$ abstract layers such that
\begin{equation}
\label{eq:dominance}
\mathcal B_L^{(A)}
\subseteq
\Phi\!\left(\mathcal B_{n',L'}^{(B)}\right),
\qquad L'\leq L.
\end{equation}
Every distribution reachable by Strategy~A is therefore reproduced
by the constructed Strategy~B family. The latter uses
$n'=\log_2K\leq n$ qubits, with strict inequality whenever
$K=O_D(Q_N)<2^n$.
\end{proposition}

\begin{proof}
By Theorem~\ref{thm:orbit_symmetric}, the bijective isometry
$\Phi: \Delta^{K-1} \to \Delta_{\mathrm{sym}}$ identifies
$\Delta_{\mathrm{sym}}$ with $\Delta^{K-1}$ as metric spaces.
Combining with
Proposition~\ref{prop:equiv_implies_inv},
$\mathcal{B}_L^{(\mathrm{A})}
\subseteq \Delta_{\mathrm{sym}}$, so each
$P \in \mathcal{B}_L^{(\mathrm{A})}$ corresponds to a unique
$\mathbf{q}_P := \Phi^{-1}(P) \in \Delta^{K-1}$.

It suffices to exhibit an orbit-space ansatz of abstract depth
$L' \leq L$ whose reachable set
$\mathcal{B}_{n',\,L'}^{(\mathrm{B})}$ contains every such
$\mathbf{q}_P$. We construct it explicitly.

The orbit-space encoding $\psi: \{1, \ldots, K\}
\leftrightarrow \{0,1\}^{n'}$ in
Definition~\ref{def:orbit_encoding} is an arbitrary bijection.
Without loss of generality, we choose $\psi$ such that the
singleton orbit $\{\mathbf{x}_0\}$ of the $S_D$-fixed grid point
guaranteed by the hypothesis is mapped
to $|0\rangle^{\otimes n'}$. This choice is a relabelling of the freely chosen orbit encoding
$\psi$ and fixes the compressed initial state without introducing
an additional circuit layer or a trainable parameter.

This choice ensures
\begin{equation}\label{eq:initial_state_alignment}
    \iota\bigl(|0\rangle^{\otimes n}\bigr)
    = |\psi(j_0)\rangle = |0\rangle^{\otimes n'},
\end{equation}
where $j_0 \in \{1, \ldots, K\}$ is the index of the singleton orbit and 
$\iota$ is defined below.

We construct $V_{\mathrm{sym}} \subseteq \mathcal{H}_n$ explicitly via 
the symmetrization operator
\[
    \mathrm{Sym} := \frac{1}{|S_D|} \sum_{\sigma \in S_D} \Pi(\sigma).
\]
Standard re-indexing arguments show that
$\mathrm{Sym}^2 = \mathrm{Sym}$, $\mathrm{Sym}^\dagger = \mathrm{Sym}$,
and $\mathrm{Im}(\mathrm{Sym}) = V_{\mathrm{sym}}$; hence $\mathrm{Sym}$
is the orthogonal projection $\mathcal{H}_n \to V_{\mathrm{sym}}$.

Applied to a computational basis vector $|\phi(\mathbf{x})\rangle$ with
$\mathbf{x} \in \mathcal{O}_j$,
\[
    \mathrm{Sym}\,|\phi(\mathbf{x})\rangle
    \;=\; \frac{1}{|S_D|} \sum_{\sigma \in S_D}
    |\phi(\sigma \cdot \mathbf{x})\rangle
    \;\stackrel{(\ast)}{=}\;
    \frac{|\mathrm{Stab}(\mathbf{x})|}{|S_D|}
    \sum_{\mathbf{x}' \in \mathcal{O}_j} |\phi(\mathbf{x}')\rangle
    \;\stackrel{(\dagger)}{=}\;
    \frac{1}{s_j} \sum_{\mathbf{x}' \in \mathcal{O}_j}
    |\phi(\mathbf{x}')\rangle,
\]
where $(\ast)$ groups the sum by orbit element---by the
orbit-stabilizer theorem each $\mathbf{x}' \in \mathcal{O}_j$ is the
image of exactly $|\mathrm{Stab}(\mathbf{x})|$ group elements---and
$(\dagger)$ uses $|S_D| = s_j \cdot |\mathrm{Stab}(\mathbf{x})|$.
Define the \emph{orbit-uniform superposition} states
\[
    |\Omega_j\rangle 
    := \frac{1}{\sqrt{s_j}} \sum_{\mathbf{x} \in \mathcal{O}_j}
    |\phi(\mathbf{x})\rangle, 
    \qquad j = 1, \ldots, K.
\]
We claim that $\{|\Omega_j\rangle\}_{j=1}^K$ is an orthonormal basis of
$V_{\mathrm{sym}}$:
\begin{itemize}
    \item \emph{Invariance.} For any $\sigma$,
    $\Pi(\sigma)|\Omega_j\rangle = \frac{1}{\sqrt{s_j}}
    \sum_{\mathbf{x} \in \mathcal{O}_j} |\phi(\sigma \cdot \mathbf{x})\rangle
    = |\Omega_j\rangle$, since $\sigma$ permutes $\mathcal{O}_j$ within
    itself; hence $|\Omega_j\rangle \in V_{\mathrm{sym}}$.
    \item \emph{Orthonormality.}
    $\langle \Omega_j | \Omega_{j'} \rangle 
    = \frac{1}{\sqrt{s_j s_{j'}}}
    \sum_{\mathbf{x} \in \mathcal{O}_j, \mathbf{x}' \in \mathcal{O}_{j'}}
    \delta_{\mathbf{x}\mathbf{x}'} = \delta_{jj'}$,
    since orbits are disjoint for $j \neq j'$ and the diagonal terms
    contribute $s_j$ ones when $j = j'$.
    \item \emph{Spanning.} Any $|\xi\rangle = \sum_{\mathbf{x}}
    c_{\mathbf{x}}|\phi(\mathbf{x})\rangle \in V_{\mathrm{sym}}$ satisfies
    $c_{\sigma \cdot \mathbf{x}} = c_{\mathbf{x}}$ for every $\sigma$,
    so $c_{\mathbf{x}}$ is constant on each orbit; writing
    $c_{\mathbf{x}} = c_j / \sqrt{s_j}$ for $\mathbf{x} \in \mathcal{O}_j$
    gives $|\xi\rangle = \sum_j c_j |\Omega_j\rangle$.
\end{itemize}
Therefore $\dim V_{\mathrm{sym}} = K = 2^{n'}$.

Since $\{|\Omega_j\rangle\}_{j=1}^K$ and
$\{|\psi(j)\rangle\}_{j=1}^K$ are orthonormal bases of two $K$-dimensional
Hilbert spaces, the linear map
\[
    \iota: V_{\mathrm{sym}} \xrightarrow{\;\cong\;} \mathcal{H}_{n'},
    \qquad
    \iota(|\Omega_j\rangle) := |\psi(j)\rangle,
\]
extends linearly to a unitary isomorphism (an orthonormal-basis-to-
orthonormal-basis bijection preserves all inner products by linearity).

The augmented grid contains a singleton orbit
$\mathcal{O}_{j_0} = \{\mathbf{x}_0\}$ with $s_{j_0} = 1$
(guaranteed by the fixed-point hypothesis).
For this orbit, $|\Omega_{j_0}\rangle = |\phi(\mathbf{x}_0)\rangle$ is
itself a computational basis vector. By the assumed encoding $\phi(x_0)=0^n$, and choosing the
orbit-label encoding so that $\psi(j_0)=0^{n'}$, we obtain
\[
\iota\bigl(|0\rangle^{\otimes n}\bigr)
=
\iota(|\Omega_{j_0}\rangle)
=
|\psi(j_0)\rangle
=
|0\rangle^{\otimes n'},
\]
establishing~\eqref{eq:initial_state_alignment}. In particular,
$|0\rangle^{\otimes n} \in V_{\mathrm{sym}}$, a fact used below.

\paragraph{$U^{(\mathrm{A})}$ preserves $V_{\mathrm{sym}}$.}
For any $|\xi\rangle \in V_{\mathrm{sym}}$ and any $\sigma \in S_D$,
\[
    \Pi(\sigma)\bigl[U^{(\mathrm{A})}(\bm\theta)|\xi\rangle\bigr]
    \;\stackrel{\text{equivariance}}{=}\;
    U^{(\mathrm{A})}(\bm\theta)\,\Pi(\sigma)|\xi\rangle
    \;\stackrel{|\xi\rangle \in V_{\mathrm{sym}}}{=}\;
    U^{(\mathrm{A})}(\bm\theta)|\xi\rangle,
\]
so $U^{(\mathrm{A})}(\bm\theta)|\xi\rangle \in V_{\mathrm{sym}}$. The
restriction $U^{(\mathrm{A})}(\bm\theta)|_{V_{\mathrm{sym}}}$ is therefore
a unitary on $V_{\mathrm{sym}}$.

Define
\[
    U^{(\mathrm{B})}(\bm\theta) 
    := \iota \cdot U^{(\mathrm{A})}(\bm\theta)\big|_{V_{\mathrm{sym}}}
       \cdot \iota^{-1} \;\in\; U(2^{n'}),
\]
unitary as the composition of three unitaries.

Since each abstract layer $U_l^{(A)}$ preserves $V_{\mathrm{sym}}$,
its compression
\[
U_l^{(B)}(\theta)
=
\iota
\left(
U_l^{(A)}(\theta)\big|_{V_{\mathrm{sym}}}
\right)
\iota^{-1}
\]
is unitary. Multiplying these compressed layers gives $U^{(B)}$,
so the construction preserves the number of abstract layers and
introduces no additional trainable parameters.

Let $P \in \mathcal{B}_L^{(\mathrm{A})}$ be realized by parameters
$\bm\theta^*$, i.e.,
$P(\mathbf{x}) = \bigl|\langle\phi(\mathbf{x})|
U^{(\mathrm{A})}(\bm\theta^*)|0\rangle^{\otimes n}\bigr|^2$.
Since $|0\rangle^{\otimes n} \in V_{\mathrm{sym}}$ and
$U^{(\mathrm{A})}$ preserves $V_{\mathrm{sym}}$, the output state
expands in the orbit basis:
\begin{equation}\label{eq:output_expansion}
    U^{(\mathrm{A})}(\bm\theta^*)|0\rangle^{\otimes n}
    \;=\; \sum_{j'=1}^K \alpha_{j'}\,|\Omega_{j'}\rangle
    \;\in\; V_{\mathrm{sym}},
\end{equation}
for complex amplitudes $\alpha_{j'} \in \mathbb{C}$.
For any representative $\mathbf{x}_j \in \mathcal{O}_j$, the overlap
\[
    \langle\phi(\mathbf{x}_j) | 
    U^{(\mathrm{A})}(\bm\theta^*) |0\rangle^{\otimes n}
    \;=\; \sum_{j'} \alpha_{j'} \langle\phi(\mathbf{x}_j) | \Omega_{j'}\rangle
    \;=\; \frac{\alpha_j}{\sqrt{s_j}},
\]
using $\langle\phi(\mathbf{x}_j) | \Omega_{j'}\rangle 
= \delta_{jj'} / \sqrt{s_j}$ (only the term with 
$\mathbf{x}_j \in \mathcal{O}_{j'}$ contributes). Therefore
\begin{equation}\label{eq:amplitude_to_probability}
    P(\mathbf{x}_j) 
    \;=\; \biggl|\frac{\alpha_j}{\sqrt{s_j}}\biggr|^2
    \;=\; \frac{|\alpha_j|^2}{s_j}
    \quad \Longleftrightarrow \quad
    |\alpha_j|^2 = s_j \cdot P(\mathbf{x}_j),
\end{equation}
where $P(\mathbf{x}_j)$ is well-defined by orbit-constancy of $P$
(Step~5 of Theorem~\ref{thm:orbit_symmetric}).

Now compute $q_{\bm\theta^*}(j)$. Applying $\iota$ to the expansion
\eqref{eq:output_expansion} and using $\iota(|\Omega_{j'}\rangle) =
|\psi(j')\rangle$,
\[
    \iota\, U^{(\mathrm{A})}(\bm\theta^*) |0\rangle^{\otimes n}
    \;=\; \sum_{j'} \alpha_{j'}\,|\psi(j')\rangle.
\]
Combining with $|0\rangle^{\otimes n'} = \iota(|0\rangle^{\otimes n})$
and $\iota^{-1}\iota |0\rangle^{\otimes n} = |0\rangle^{\otimes n}$
(since $|0\rangle^{\otimes n} \in V_{\mathrm{sym}}$):
\begin{align*}
    q_{\bm\theta^*}(j)
    &= \bigl|\langle\psi(j)|
        U^{(\mathrm{B})}(\bm\theta^*) |0\rangle^{\otimes n'}\bigr|^2 \\
    &= \bigl|\langle\psi(j)|
        \iota U^{(\mathrm{A})}(\bm\theta^*) \iota^{-1}\,
        \iota(|0\rangle^{\otimes n})\bigr|^2 \\
    &= \bigl|\langle\psi(j)|
        \iota U^{(\mathrm{A})}(\bm\theta^*) |0\rangle^{\otimes n}\bigr|^2 \\
    &= \biggl| \sum_{j'} \alpha_{j'} 
        \langle\psi(j)|\psi(j')\rangle\biggr|^2 \\
    &= |\alpha_j|^2 \\
    &\stackrel{\eqref{eq:amplitude_to_probability}}{=}\;
       s_j \cdot P(\mathbf{x}_j)
    \;=\; (\mathbf{q}_P)_j,
\end{align*}
where the final equality is the surjectivity construction
$(\mathbf{q}_P)_j := s_j \cdot P(\mathbf{x}_j)$ from Step~5 of
Theorem~\ref{thm:orbit_symmetric}.

Hence $\mathbf{q}_{\bm\theta^*} = \mathbf{q}_P$, so 
$\mathbf{q}_P \in \mathcal{B}_{n',L}^{(\mathrm{B})}$, and 
$\Phi(\mathbf{q}_P) = P$ confirms~\eqref{eq:dominance} with $L' = L$.

The constructed family reproduces the Strategy~A distributions
exactly. A larger orbit-space family may reach additional symmetric
distributions, but strict inclusion is not required for the proposition.
\end{proof}

\begin{remark}[Invariant initialization without a fixed grid point]\label{rem:no_fixed_point}
If $Q_N$ contains no $S_D$-fixed point, initialize Strategy~A in the
orbit-uniform state $|\Omega_{j_0}\rangle$ of any chosen orbit
$O_{j_0}$, as in Remark~\ref{rem:init-state}. Choose the orbit encoding so that
$\psi(j_0)=0^{n'}$. Then
\[
\iota(|\Omega_{j_0}\rangle)=|0\rangle^{\otimes n'},
\]
and the subsequent argument of Proposition~\ref{prop:dominance} is unchanged.

More generally, for any fixed invariant pure initial state
$|\psi_0\rangle\in V_{\mathrm{sym}}$, use the compressed initial
state $\iota|\psi_0\rangle$. Then
\[
U^{(B)}(\theta)\iota|\psi_0\rangle
=
\iota U^{(A)}(\theta)|\psi_0\rangle.
\]
Preparing the fixed initial states introduces no additional trainable
parameters. Their preparation costs are accounted for separately
from the abstract ansatz layers.
\end{remark}

\begin{remark}[Abstract depth versus compiled depth]\label{rem:compile}
Proposition~\ref{prop:dominance} compares ansatz families at the representation-theoretic
level. The compressed layers are abstract unitaries on the orbit
register, and compiling them into a native gate set may increase
circuit depth substantially. The proposition asserts inclusion of
reachable distribution families, not a bound on compiled gate counts
or depth. It does not assert that the compressed family belongs to a
prescribed fixed-depth hardware-efficient ansatz.
The layerwise comparison uses the individual-layer equivariance
assumed in Proposition~\ref{prop:dominance}; fixed initialization costs are treated
separately as in Remark~\ref{rem:no_fixed_point}.
\end{remark}

\begin{remark}[Optimization-quality preservation]
\label{rem:opt_preservation}
The bijective isometry $\Phi$ ensures that approximation
quality is preserved exactly: for any target
$P^* \in \Delta_{\mathrm{sym}}$ and any
$\mathbf{q}_{\bm{\theta}} \in \Delta^{K-1}$,
\[
    \TV\bigl(\Phi(\mathbf{q}_{\bm{\theta}}),\, P^*\bigr)
    = \TV\bigl(\mathbf{q}_{\bm{\theta}},\,
    \Phi^{-1}(P^*)\bigr).
\]
Consequently, minimizing TV in orbit space (Strategy~B) is
metrically identical to minimizing TV in the full grid space
(Strategy~A restricted to $\Delta_{\mathrm{sym}}$). This holds
without any expressivity assumption and validates the use of
orbit-space loss functions during training.
\end{remark}

\begin{definition}[$H$-equivariant CZ entangler class]\label{def:ESD}
For the qubit-permutation subgroup $H\leq S_n$ of Assumption~\ref{qubitpermutation},
let $\mathcal Q=\{0,\ldots,n-1\}$ and define
\[
\mathcal E_H
=
\left\{
E_S=\prod_{\{i,j\}\in S}\mathrm{CZ}_{i,j}:
S\subseteq\binom{\mathcal Q}{2},\quad
hS=S\ \text{for every }h\in H
\right\},
\]
where
\[
hS=\{\{h(i),h(j)\}:\{i,j\}\in S\}.
\]
\end{definition}

Every $E\in\mathcal E_H$ is a product of commuting involutive
Clifford gates, so $E^2=I$ and
$\Pi_hE\Pi_h^\dagger=E$ for every $h\in H$.
The complete edge set gives $E_{\mathrm{full}}$, which belongs
to $\mathcal E_H$ for every $H$ by Proposition~\ref{prop:full_cz_sym}.

\begin{definition}[Dynamical Lie algebra of the reference model]\label{def:dla}
For the shared-parameter reference ansatz of Theorem~\ref{fullentanglementPQC} with
$E\in\mathcal E_H$, let
\[
\mathfrak h_H
=
\operatorname{span}_{\mathbb R}
\left\{
\sum_{q\in O}Y_q,\ \sum_{q\in O}Z_q:
O\text{ is an }H\text{-orbit on qubit positions}
\right\}.
\]
Define
\[
\mathfrak g_H(E)
=
\operatorname{Lie}_{\mathbb R}
\left(
i\mathfrak h_H\cup iE\mathfrak h_HE
\right)
\subseteq\mathfrak u(2^n).
\]
\end{definition}

\begin{proposition}[Invariant-subspace ceiling for the reference model]\label{prop:ceiling}
Let
\[
W_H
=
\{
|\psi\rangle\in\mathcal H_n:
\Pi_h|\psi\rangle=|\psi\rangle
\text{ for every }h\in H
\}.
\]
Write $K_H=\dim W_H=|\{0,1\}^n/H|$.
For every $E\in\mathcal E_H$,
\[
\mathfrak g_H(E)\big|_{W_H}
\subseteq
\mathfrak{su}(W_H)
\cong
\mathfrak{su}(K_H).
\]
\end{proposition}

\begin{proof}
Each generator in $\mathfrak h_H$ commutes with every $\Pi_h$,
because it is a sum over an $H$-orbit. Since $E$ also commutes
with every $\Pi_h$, the same holds for $E\mathfrak h_HE$.
All generators therefore preserve $W_H$, and restriction is a
Lie algebra homomorphism into $\mathfrak u(W_H)$.

For each qubit orbit $O$, define the parity operators
\[
X_O=\prod_{q\in O}X_q,
\qquad
Z_O=\prod_{q\in O}Z_q,
\]
with identity on the remaining qubits. Both commute with the
qubit-permutation action and restrict to unitaries on $W_H$.
The sum $\sum_{q\in O}Z_q$ anticommutes with $X_O$, and
$\sum_{q\in O}Y_q$ anticommutes with $Z_O$.

For $E=E_S$, the identities
\[
EZ_qE=Z_q,
\qquad
EY_qE=Y_q\prod_{j:\{q,j\}\in S}Z_j
\]
show that $E(\sum_{q\in O}Z_q)E$ anticommutes with $X_O$,
while $E(\sum_{q\in O}Y_q)E$ anticommutes with $Z_O$.
If a Hermitian operator $A$ anticommutes with a unitary $V$
preserving $W_H$, then
\[
\operatorname{tr}_{W_H}(A)
=
\operatorname{tr}_{W_H}(VAV^\dagger)
=
-\operatorname{tr}_{W_H}(A)=0.
\]
All restricted generators are therefore traceless. Commutators
are traceless, and real linear combinations preserve tracelessness,
which proves the inclusion in $\mathfrak{su}(W_H)$.
\end{proof}

\begin{corollary}[Orbit-space coverage of equivariant outputs]\label{cor:dominance}
Let $\mathcal B_\infty^{(A)}$ denote the closure of the reachable
Born distributions of an $S_D$-equivariant ansatz under the induced
representation of Definition~\ref{def:induced_rep}, with a fixed invariant pure
initial state. Let $\mathcal B_{n',\infty}^{(B)}$ denote the closure
of the Born distributions obtainable by unrestricted universal
circuits on $n'=\log_2K$ orbit qubits. Then
\[
\mathcal B_\infty^{(A)}
\subseteq
\Delta_{\mathrm{sym}}
=
\Phi\!\left(\mathcal B_{n',\infty}^{(B)}\right).
\]
\end{corollary}

\begin{proof}
Equivariance and invariant initialization imply that each
Strategy~A output state is invariant under the induced representation.
The argument of Proposition~\ref{prop:equiv_implies_inv} therefore gives a Born distribution
in $\Delta_{\mathrm{sym}}$. Since this finite-dimensional probability
simplex is closed, taking the closure preserves the inclusion.

For every $q\in\Delta_{K-1}$, the normalized orbit-register state
\[
|\chi_q\rangle
=
\sum_{j=1}^K\sqrt{q_j}\,|\psi(j)\rangle
\]
has Born distribution $q$. Unrestricted universal circuits can
approximate this state arbitrarily closely; hence
\[
\mathcal B_{n',\infty}^{(B)}=\Delta_{K-1}.
\]
Theorem~\ref{thm:orbit_symmetric} then yields
\[
\Phi\!\left(\mathcal B_{n',\infty}^{(B)}\right)
=
\Phi(\Delta_{K-1})
=
\Delta_{\mathrm{sym}}.
\]
\end{proof}

\begin{remark}[Reachability and the role of the reference model]
Corollary~\ref{cor:dominance} requires no particular entangler or dynamical Lie
algebra. It does not assert that a specified equivariant ansatz
covers $\Delta_{\mathrm{sym}}$. Such coverage is a separate
reachability question. In the usual controllability setting, a
sufficient condition is that the restricted dynamical Lie algebra
of the grid-space controls contain
$\mathfrak{su}(V_{\mathrm{sym}})$.

Proposition~\ref{prop:ceiling} concerns the separate qubit-permutation reference
model, whose invariant space is $W_H$ and whose dimension is $K_H$.
It is not required for Corollary~\ref{cor:dominance}, and no identification of
$W_H$ with the grid invariant space $V_{\mathrm{sym}}$ is assumed.
The coverage statement for Strategy~B concerns unrestricted
universal circuits; a particular fixed-depth architecture requires
its own expressivity analysis.
\end{remark}

\subsection{General Framework}

\begin{theorem}[General orbit count]
\label{thm:general_burnside}
For any subgroup $G \leq S_D$, the number of $G$-orbits on a set
$\mathcal{X}$ is given by Burnside's lemma:
\begin{equation}\label{eq:general_burnside}
    \mathcal{O}_G(\mathcal{X})
    = \frac{1}{|G|}\sum_{\sigma \in G}
    |\mathrm{Fix}_{\mathcal{X}}(\sigma)|,
\end{equation}
where $\mathrm{Fix}_{\mathcal{X}}(\sigma)
= \{\mathbf{x} \in \mathcal{X} :
\sigma \cdot \mathbf{x} = \mathbf{x}\}$.
When $\mathcal{X} = \mathcal{X}_N
= \{\mathbf{k} \in \mathbb{Z}_{\geq 0}^D
: \sum k_i = N\}$ (a lattice at resolution $N$) and $\sigma$ has
cycle lengths $\lambda_1, \ldots, \lambda_s$
($\sum \lambda_i = D$):
\begin{equation}\label{eq:fix_lattice}
    |\mathrm{Fix}_{\mathcal{X}_N}(\sigma)|
    = \bigl|\bigl\{(a_1, \ldots, a_s)
    \in \mathbb{Z}_{\geq 0}^s :
    \textstyle\sum_{i=1}^s \lambda_i a_i = N
    \bigr\}\bigr|.
\end{equation}
\end{theorem}

\begin{proof}
Equation~\eqref{eq:general_burnside} is the classical Burnside's
lemma. For~\eqref{eq:fix_lattice}: the condition
$\sigma \cdot \mathbf{k} = \mathbf{k}$ requires all coordinates
within each cycle of $\sigma$ to take a common value. Denote
the common value on cycle~$i$ (of length $\lambda_i$) by
$a_i \in \mathbb{Z}_{\geq 0}$. The constraint
$\sum_{j=1}^D k_j = N$ becomes $\sum_{i=1}^s \lambda_i a_i = N$.
\end{proof}

\begin{remark}[Generating function for fixed-point counts]
\label{rem:fix_genfunc}
The number of solutions to $\sum_i \lambda_i a_i = N$ with
$a_i \geq 0$ is given by the generating function
\[
    \sum_{N \geq 0} |\mathrm{Fix}_{\mathcal{X}_N}(\sigma)|\,
    z^N = \prod_{i=1}^s \frac{1}{1 - z^{\lambda_i}},
\]
which provides an efficient computational route for individual
$\sigma$. This avoids enumerating solutions when only the orbit
count is needed.
\end{remark}

\subsubsection{Application to the Augmented Grid}

In the QCBM setting, the relevant set is the augmented grid
$\mathcal{Q}_N$ as defined in Definition~\ref{def:augmented_grid},
not the natural lattice $\mathcal{X}_N$ alone. We now extend the
framework to subgroup actions on $\mathcal{Q}_N$.

\begin{proposition}[Orbit count on the augmented grid]
\label{prop:augmented_orbit_general}
Let $G \leq S_D$ be a subgroup, and let
\[
    \mathcal{Q}_N = \mathcal{X}_N \;\cup\;
    \bigcup_{m \in \mathcal{M}} \mathcal{C}^{(m)}
\]
be an augmented grid as in
Definition~\ref{def:augmented_grid}, where each
$\mathcal{C}^{(m)} \subseteq \mathcal{T}_N^{(m)}$ is a union of
complete \textbf{$G$-orbits}. Then the orbit count decomposes
by component:
\begin{equation}\label{eq:augmented_orbit_decomp}
    \mathcal{O}_G(\mathcal{Q}_N)
    = \mathcal{O}_G(\mathcal{X}_N)
    + \sum_{m \in \mathcal{M}} \mathcal{O}_G(\mathcal{C}^{(m)}).
\end{equation}
By the orbit correspondence
(Proposition~\ref{prop:orbit_correspondence}, which holds for
every subgroup $G \leq S_D$),
$\mathcal{O}_G(\mathcal{T}_N^{(m)}) = \mathcal{O}_G(\mathcal{A}_N^{(m)})$,
where $\mathcal{A}_N^{(m)} = \{\mathbf{k} \in
\mathbb{Z}_{\geq 0}^D : \sum k_i = N - m\}$. When
$\mathcal{C}^{(m)} = \mathcal{T}_N^{(m)}$ for every
$m \in \mathcal{M}$ (full inclusion), we obtain:
\begin{equation}\label{eq:augmented_burnside}
    \mathcal{O}_G(\mathcal{Q}_N)
    = \mathcal{O}_G(\mathcal{X}_N) +
    \sum_{m \in \mathcal{M}} \mathcal{O}_G(\mathcal{X}_{N-m}),
\end{equation}
where each $\mathcal{O}_G(\mathcal{X}_{N-m})$ is computed via
Theorem~\ref{thm:general_burnside} with $N$ replaced by $N - m$.
\end{proposition}

\begin{proof}
The component decomposition follows from
Lemma~\ref{lem:components}: the $S_D$-action preserves each
component $\mathcal{X}_N, \mathcal{T}_N^{(1)}, \ldots,
\mathcal{T}_N^{(D-1)}$, hence so does any subgroup action
$G \leq S_D$. Therefore $G$-orbits do not cross components,
yielding~\eqref{eq:augmented_orbit_decomp}.

The orbit correspondence
$\mathcal{O}_G(\mathcal{T}_N^{(m)}) = \mathcal{O}_G(\mathcal{A}_N^{(m)})$
holds for any subgroup $G \leq S_D$ because the centroid map
$G_m: \mathbf{k} \mapsto \mathbf{k} + (m/D)\mathbf{1}$ commutes
with every coordinate permutation (Lemma~\ref{lem:components}),
hence with every $\sigma \in G$.

The set $\mathcal{A}_N^{(m)}$ is structurally identical to
$\mathcal{X}_{N-m}$ (both are non-negative integer $D$-tuples
summing to $N-m$), so $\mathcal{O}_G(\mathcal{A}_N^{(m)}) =
\mathcal{O}_G(\mathcal{X}_{N-m})$.

The hypothesis that $\mathcal{C}^{(m)}$ is a union of complete
$G$-orbits ensures that orbit-counting on $\mathcal{C}^{(m)}$ is
well-defined: $\mathcal{O}_G(\mathcal{C}^{(m)}) =
|\mathcal{C}^{(m)} / G|$. When
$\mathcal{C}^{(m)} = \mathcal{T}_N^{(m)}$, this reduces to
$\mathcal{O}_G(\mathcal{T}_N^{(m)}) =
\mathcal{O}_G(\mathcal{X}_{N-m})$, giving
\eqref{eq:augmented_burnside}.
\end{proof}

\begin{remark}[$G$-closure of $\mathcal{C}^{(m)}$: added flexibility for subgroups]
\label{rem:g-closure}
If $\mathcal{C}^{(m)}$ is a union of complete $S_D$-orbits, then it is automatically a
union of complete $G$-orbits for \emph{every} subgroup $G \le S_D$: each
$S_D$-orbit decomposes into finitely many $G$-orbits, all of which are
contained in $\mathcal{C}^{(m)}$. The hypothesis of
Proposition~\ref{prop:augmented_orbit_general} is therefore \emph{weaker} than $S_D$-closure. This is a
feature rather than a restriction: when $G \subsetneq S_D$, one may select
individual $G$-orbits that do not assemble into complete $S_D$-orbits,
gaining finer granularity when tuning the augmented orbit count
$\mathcal{O}_G(\mathcal{Q}_N)$ to a power of two
(cf.\ Remark~\ref{rem:power2}). For Dirichlet targets with
$G = S_{n_1}\times\cdots\times S_{n_r}$, the selection need only respect the
block structure of $G$, not the full symmetric group.
\end{remark}

\begin{proposition}[Orbit-space reduction on $\mathcal{Q}_N$]
\label{prop:orbit_reduction_QN}
For any target distribution $P$ on $\mathcal{Q}_N$ with symmetry
group $G_P \leq S_D$, the orbit-space QCBM uses
\begin{equation}\label{eq:orbit_qubits}
    n' = \lceil\log_2\mathcal{O}_{G_P}(\mathcal{Q}_N)\rceil
    \qquad\text{qubits},
\end{equation}
with no symmetry constraints on the circuit and a compression
ratio $|\mathcal{Q}_N| / \mathcal{O}_{G_P}(\mathcal{Q}_N)$. When
$\mathcal{Q}_N$ is selected (via the choice of $\mathcal{M}$ and
$\mathcal{C}^{(m)}$) such that
$\mathcal{O}_{G_P}(\mathcal{Q}_N)$ is a power of~$2$, the
ceiling is exact and no inactive states arise.
\end{proposition}

The general framework requires computing Burnside's
sum~\eqref{eq:general_burnside} for each lattice term in
\eqref{eq:augmented_burnside}. For distributions whose symmetry
group admits a direct-product structure, each term simplifies
to a product formula, avoiding enumeration of group elements
entirely. We now develop this simplification for Dirichlet
distributions.

\subsubsection{Application to Dirichlet Distributions}
\label{subsubsec:dirichlet}

\begin{proposition}[Symmetry group of Dirichlet distributions]
\label{prop:dirichlet_sym}
Let $\mathrm{Dir}(\alpha_1, \ldots, \alpha_D)$ be a Dirichlet
distribution whose parameter vector takes $r$ distinct values
$\beta_1 < \cdots < \beta_r$. Define the index blocks $I_k =
\{i \in \{1,\ldots,D\} : \alpha_i = \beta_k\}$ with sizes
$n_k = |I_k|$ and $\sum_{k=1}^r n_k = D$. Then:
\begin{equation}\label{eq:dirichlet_sym_group}
    G_P = S_{I_1} \times S_{I_2} \times \cdots
    \times S_{I_r}
    \cong S_{n_1} \times S_{n_2} \times \cdots
    \times S_{n_r},
    \qquad |G_P| = \prod_{k=1}^r n_k!
\end{equation}
\end{proposition}

\begin{proof}
We establish both inclusions and verify the direct-product
structure.

$G_P \subseteq S_{I_1} \times \cdots \times
S_{I_r}$:
Let $\sigma \in G_P$, i.e., $f(\sigma \cdot \mathbf{x}) =
f(\mathbf{x})$ for all $\mathbf{x}$ in the simplex interior.
The Dirichlet density is:
\[
    f(\mathbf{x}) \propto \prod_{i=1}^D
    x_i^{\alpha_i - 1}.
\]
Applying $\sigma$:
\[
    f(\sigma \cdot \mathbf{x})
    \propto \prod_{i=1}^D
    x_{\sigma^{-1}(i)}^{\alpha_i - 1}
    = \prod_{i=1}^D
    x_i^{\alpha_{\sigma(i)} - 1}.
\]
Since the $x_i$ are free positive variables on the simplex
interior (any assignment $x_i > 0$ with $\sum x_i = 1$ is a
valid point), the equality $\prod x_i^{\alpha_{\sigma(i)}-1}
= \prod x_i^{\alpha_i - 1}$ for all such $\mathbf{x}$ requires
exponent-wise agreement:
\begin{equation}\label{eq:exponent_match}
    \alpha_{\sigma(i)} = \alpha_i
    \qquad \forall\, i = 1, \ldots, D.
\end{equation}
This means $\sigma$ maps each index to another index with the
same parameter value: if $i \in I_k$ (i.e., $\alpha_i =
\beta_k$), then $\alpha_{\sigma(i)} = \beta_k$, so $\sigma(i)
\in I_k$. Therefore $\sigma(I_k) \subseteq I_k$ for every $k$.
Since $\sigma$ is a bijection and $I_k$ is finite,
$\sigma(I_k) = I_k$. Hence $\sigma$ restricts to a permutation
within each block: $\sigma \in S_{I_1} \times \cdots \times
S_{I_r}$.

$S_{I_1} \times \cdots \times S_{I_r}
\subseteq G_P$:
Let $\sigma \in S_{I_k}$ for some $k$ (i.e., $\sigma$ permutes
the indices within $I_k$ and fixes all others). For
$i \in I_k$: $\alpha_{\sigma(i)} = \beta_k = \alpha_i$ (since
both $i$ and $\sigma(i)$ belong to $I_k$). For $i \notin I_k$:
$\sigma(i) = i$, so $\alpha_{\sigma(i)} = \alpha_i$ trivially.
Therefore~\eqref{eq:exponent_match} holds, and
$\sigma \in G_P$. Since $G_P$ is a group and each generator
of each $S_{I_k}$ lies in $G_P$, the full product
$S_{I_1} \times \cdots \times S_{I_r} \subseteq G_P$.

We verify the four defining properties of an (internal) direct
product:

\emph{(i) Subgroup property.} Each $S_{I_k}$ is a subgroup of
$S_D$: it consists of permutations that act only on the indices
in $I_k$ and fix all others. The identity $e$ belongs to every
$S_{I_k}$, and $S_{I_k}$ is closed under composition and
inversion.

\emph{(ii) Pairwise trivial intersection.}
For $k \neq l$, $S_{I_k} \cap S_{I_l} = \{e\}$. Any non-identity
element of $S_{I_k}$ moves some index in $I_k$, while every
element of $S_{I_l}$ fixes all indices outside $I_l$. Since
$I_k \cap I_l = \emptyset$, an element moving an index in $I_k$
cannot belong to $S_{I_l}$.

\emph{(iii) Pairwise commutativity.} For $k \neq l$ and any
$\sigma_k \in S_{I_k}$, $\sigma_l \in S_{I_l}$: the supports
$\mathrm{supp}(\sigma_k) \subseteq I_k$ and
$\mathrm{supp}(\sigma_l) \subseteq I_l$ are disjoint. For any
index $i$:
\begin{itemize}
    \item If $i \in I_k \setminus I_l$: $\sigma_l(i) = i$, so
    $\sigma_k\sigma_l(i) = \sigma_k(i)$. Also
    $\sigma_k(i) \in I_k$ (support closed under $\sigma_k$) and
    $I_k \cap I_l = \emptyset$, so
    $\sigma_l(\sigma_k(i)) = \sigma_k(i)$, giving
    $\sigma_l\sigma_k(i) = \sigma_k(i)$.
    \item If $i \in I_l \setminus I_k$: symmetric argument
    gives $\sigma_k\sigma_l(i) = \sigma_l(i) =
    \sigma_l\sigma_k(i)$.
    \item If $i \notin I_k \cup I_l$: both permutations fix $i$,
    so $\sigma_k\sigma_l(i) = i = \sigma_l\sigma_k(i)$.
\end{itemize}
Therefore $\sigma_k\sigma_l = \sigma_l\sigma_k$.

\emph{(iv) Unique decomposition.} Every $\sigma \in G_P$
decomposes uniquely as $\sigma = \sigma_1 \cdots \sigma_r$ with
$\sigma_k = \sigma|_{I_k} \in S_{I_k}$ (restrict $\sigma$ to
each block, and their product recovers $\sigma$ since the
blocks partition $\{1,\ldots,D\}$).

Properties~(i)--(iv) establish that $G_P$ is the internal direct
product $S_{I_1} \times \cdots \times S_{I_r}$, with
$|G_P| = \prod_{k=1}^r |S_{I_k}| = \prod_{k=1}^r n_k!$.
\end{proof}

\begin{theorem}[Orbit count for Dirichlet distributions]
\label{thm:dirichlet_orbits}
For a Dirichlet distribution with symmetry group $G = S_{n_1}
\times \cdots \times S_{n_r}$, the number of $G$-orbits on
$\mathcal{X}_N$ is:
\begin{equation}\label{eq:dirichlet_orbit_count}
    \mathcal{O}_G(\mathcal{X}_N)
    = \sum_{\substack{j_1 + \cdots + j_r = N \\
    j_k \geq 0}}
    \prod_{k=1}^r p_{n_k}(j_k),
\end{equation}
where $p_{n_k}(j_k)$ is the number of partitions of $j_k$ into
at most $n_k$ non-negative parts.
\end{theorem}

\begin{proof}
We establish a bijection between $G$-orbits on $\mathcal{X}_N$
and certain $r$-tuples of partitions, then count the latter.

Since $G = S_{I_1} \times \cdots \times S_{I_r}$ acts
block-diagonally (each $S_{I_k}$ permutes coordinates within
block $I_k$ and fixes others), the $G$-orbit of $\mathbf{k} \in
\mathcal{X}_N$ is determined by the $r$-tuple of restrictions
$(\mathbf{k}|_{I_1}, \ldots, \mathbf{k}|_{I_r})$ \emph{up to
permutation within each block}. The $S_{I_k}$-orbit of
$\mathbf{k}|_{I_k} \in \mathbb{Z}_{\geq 0}^{n_k}$ is uniquely
labeled by the sorted tuple $\lambda^{(k)}$, which is a
partition of the block sum
$j_k := \sum_{i \in I_k} k_i$ into at most $n_k$ non-negative
parts.

The map
\[
    [\mathbf{k}]_G
    \;\longmapsto\;
    (\lambda^{(1)}, \ldots, \lambda^{(r)})
\]
is well-defined (independent of representative) and injective
(distinct $r$-tuples yield distinct $G$-orbits). Surjectivity
onto $r$-tuples with $\sum_k |\lambda^{(k)}| = N$ follows by
constructing $\mathbf{k}$ from any chosen representatives in
each block.

The global constraint $\sum_{i=1}^D k_i = N$ becomes
$\sum_{k=1}^r j_k = N$. For each composition $(j_1, \ldots,
j_r)$ of $N$ into $r$ non-negative parts, the number of
$r$-tuples of block partitions is
\[
    \prod_{k=1}^r p_{n_k}(j_k),
\]
because the block partitions are chosen independently. Summing
over all compositions:
\[
    \mathcal{O}_G(\mathcal{X}_N)
    = \sum_{\substack{j_1 + \cdots + j_r = N \\ j_k \geq 0}}
    \prod_{k=1}^r p_{n_k}(j_k).
    \qedhere
\]
\end{proof}

\begin{corollary}[Augmented orbit count for Dirichlet]
\label{cor:augmented_dirichlet}
For the augmented grid $\mathcal{Q}_N = \mathcal{X}_N \cup
\bigcup_{m \in \mathcal{M}} \mathcal{C}^{(m)}$ with symmetry
group $G = S_{n_1} \times \cdots \times S_{n_r}$ and each
$\mathcal{C}^{(m)} = \mathcal{T}_N^{(m)}$ (full inclusion of
selected types):
\begin{equation}\label{eq:augmented_dirichlet_orbits}
    \mathcal{O}_G(\mathcal{Q}_N)
    = \mathcal{O}_G(\mathcal{X}_N) +
    \sum_{m \in \mathcal{M}}\;
    \sum_{\substack{j_1+\cdots+j_r = N-m \\ j_k \geq 0}}
    \prod_{k=1}^r p_{n_k}(j_k),
\end{equation}
obtained by combining
Proposition~\ref{prop:augmented_orbit_general}
with~\eqref{eq:dirichlet_orbit_count}. For partial inclusion
($\mathcal{C}^{(m)} \subsetneq \mathcal{T}_N^{(m)}$), the term
$\mathcal{O}_G(\mathcal{C}^{(m)})$ counts only the selected
$G$-orbits and replaces the corresponding inner sum.
\end{corollary}

\begin{proposition}[Recovery of extreme cases]
\label{prop:extreme_cases}
Formula~\eqref{eq:dirichlet_orbit_count} recovers both endpoints
of the symmetry spectrum:
\begin{enumerate}
    \item[(a)] \emph{Full symmetry}
    ($\mathrm{Dir}(\alpha, \ldots, \alpha)$, $r = 1$,
    $n_1 = D$): the sum has a single term $j_1 = N$, giving
    $\mathcal{O}_G(\mathcal{X}_N) = p_D(N)$, consistent with
    Theorem~\ref{thm:augmented_orbits}.
    \item[(b)] \emph{No symmetry}
    (all $\alpha_i$ distinct, $r = D$, $n_k = 1$ for all $k$):
    $p_1(j_k) = 1$ for any $j_k \geq 0$, so
    $\mathcal{O}_G(\mathcal{X}_N) = |\{(j_1,\ldots,j_D) \geq 0 :
    \sum j_k = N\}| = \binom{N+D-1}{D-1} = |\mathcal{X}_N|$.
\end{enumerate}
\end{proposition}

\begin{proof}
Part~(a): With $r = 1$ and $n_1 = D$, the sum collapses to
$j_1 = N$, giving $\mathcal{O}_G(\mathcal{X}_N) = p_D(N)$.

Part~(b): With $r = D$ and $n_k = 1$, each block contains a
single coordinate. The only partition of $j$ into one
non-negative part is $(j)$ itself, so $p_1(j) = 1$ for all
$j \geq 0$. The product $\prod_{k=1}^D p_1(j_k) = 1$ for every
composition, so $\mathcal{O}_G(\mathcal{X}_N)$ equals the number
of compositions of $N$ into $D$ non-negative parts:
$\binom{N+D-1}{D-1} = |\mathcal{X}_N|$.
\end{proof}

\begin{remark}[Feasibility of exact power-of-two orbit counts]\label{rem:power2}
Since $G$-orbits of centroids can be added one at a time
(Remark~\ref{rem:g-closure}), the attainable augmented orbit counts form the
full integer interval
$\big[\,\mathcal{O}_G(\mathcal{X}_N),\; \mathcal{O}_G(\mathcal{X}_N) + \sum_{m=1}^{D-1} \mathcal{O}_G(\mathcal{X}_{N-m})\,\big]$.
Hence $\mathcal{O}_G(\mathcal{Q}_N) = 2^{\lceil \log_2 \mathcal{O}_G(\mathcal{X}_N)\rceil}$ is attainable iff
\[
2^{\lceil \log_2 \mathcal{O}_G(\mathcal{X}_N)\rceil} - \mathcal{O}_G(\mathcal{X}_N)
\;\le\; \sum_{m=1}^{D-1} \mathcal{O}_G(\mathcal{X}_{N-m}),
\]
which holds comfortably in the regimes of interest. If the criterion
fails, increase $N$ or accept $n' = \lceil \log_2 \mathcal{O}_G(\mathcal{Q}_N)\rceil$ with
masked states.
\end{remark}

\end{document}